%% file: article.tex
\documentclass[a4paper,UKenglish,cleveref, autoref, thm-restate]{lipics-v2021}

\usepackage{style}
\hypersetup{final}

\usepackage{caption}
\usepackage{subcaption}
\usepackage{stackengine}

\hypersetup{final}

\theoremstyle{remark}
\newtheorem{notation}[theorem]{Notation}
\newtheorem{convention}[theorem]{Convention}

\title{Intuitionistic Unitary Linear Logic: A Proof-Theoretical Approach to Purely Quantum Higher-Order} 

\titlerunning{Intuitionistic Unitary Linear Logic} 

\author{Julien Lamiroy}{Université Paris-Saclay, CentraleSupélec, CNRS, ENS Paris-Saclay, Inria, Laboratoire Méthodes Formelles, 91190, Gif-sur-Yvette, France.
 }{}{}{}

\author{Benoit Valiron}{Université Paris-Saclay, CentraleSupélec, CNRS, ENS Paris-Saclay, Inria, Laboratoire Méthodes Formelles, 91190, Gif-sur-Yvette, France.}{}{}{}

\author{Renaud Vilmart}{Université Paris-Saclay, CentraleSupélec, CNRS, ENS Paris-Saclay, Inria, Laboratoire Méthodes Formelles, 91190, Gif-sur-Yvette, France.}{}{}{}

\authorrunning{J. Lamiroy, B. Valiron, R. Vilmart}

\Copyright{Julien Lamiroy, Benoit Valiron, Renaud Vilmart}



\funding{This work has been partially funded by the European Union through the MSCA SE project QCOMICAL (Grant Agreement ID: 101182520), by the French National Research Agency (ANR): projects TaQC ANR-22-CE47-0012, and within the framework of ``Plan France 2030'' under the research projects EPIQ ANR-22-PETQ-0007 and HQI-R\&D ANR-22-PNCQ-0002.}

\hideLIPIcs
\nolinenumbers
\begin{document}

\maketitle
\begin{abstract}
Although the circuit model for quantum computation is well established, it is incapable of representing non-causal higher-order quantum processes such as the quantum switch. If several models of non-causal quantum computation have been considered in the literature, the approaches have so far only been focusing on the physicality of such processes, using matrices and other techniques from linear algebra. If these approaches are expressive, they however only provide a static and monolithic understanding of these processes.

In this article, we propose a new formalism for non-causal, higher-order quantum processes. Based on a Curry-Howard interpetation, our proposal offers a computational interpretation that is both compositional and modular. In particular, we present \logicname\ (\logicsymb), a logic based on linear logic focusing on conservation of unitarity for higher order terms. We prove the coherence of \logicsymb\, its completeness with regard to unitaries, and the admissibility of its cut rules. We finally discuss the validity of our approach by revisiting known non-causal quantum processes with \logicsymb.
\end{abstract}

\section{Introduction}

The typical model of quantum computation relies on the so-called coprocessor model: a classical computer equipped with a device holding a quantum memory: information, stored using quantum objects, can be in a superposed state. The device comes with a set of elementary instructions: the unitary gates. The interface to the classical computer gives access to initialization, reading, application of unitary gates on quantum registers.

In the description of quantum algorithms, unitary gates are packaged in quantum circuits: the quantum equivalent of classical circuits. Wires correspond to quantum registers---typically \emph{quantum bits}---and boxes to elementary unitary operations. The standard programming model has so far been based on quantum circuits, seen as classical structures storing instructions to be sent to the coprocessor. Following this approach, quantum programming languages are so far, for the most part, conventional, classical languages with domain-specific features dedicated to circuit description.

If quantum circuits can be built one gate at a time, their description also makes use of higher-order combinators. Indeed, they are typically described in a hierarchical manner, where circuits are built from subcircuits. This yields a natural representation of circuits with holes, or \emph{quantum combs}~\cite{qcombs}: placeholders waiting for concrete subcircuits~\cite{quipper,qcombs}. The circuit combinator that we will focus on in this paper is the \emph{control} operator: the ability to apply a unitary on a set of quantum registers conditioned by the state of another register. Since in general the state of the memory is in superposition, the corresponding subcircuit is only executed on some of the elements in superposition. The control of a subcircuit can be regarded as a restricted notion of \emph{quantum test}: the subcircuit is switched ``on'' or ``off'' in superposition, depending on a quantum state.

However, the capabilities of quantum computation go beyond the limited control offered by regular quantum circuits. As first described in the quantum {\qswitch} protocol~\cite{switch} followed by the Grenoble~\cite{grenoble} and the Lugano processes~\cite{Vanrietvelde2025consistentcircuits}, one can not only apply a fixed subcircuit in superposition, but one can also set in superposition the ordering sequence of subcircuits. This more general notion of control has been experimentally realized~\cite{taddei2021computational,procopio2015experimental}, and it has been proved to bring a competitive advantage~\cite{araujo2014computational,taddei2021computational,zhao2020quantum}.

Despite its promises, the design of a model of computation supporting superposition of executions is still an open question. As discussed above, the naive extension of quantum circuits to higher order, where one just cuts ``holes'' into the circuit for higher order input, does not capture faithfully the above processes. To be faithful, such an extension would require bending wires~\cite{switch}. Multiple graphical quantum languages make use of bending~\cite{Coecke2011,MW, Vanrietvelde2025consistentcircuits}. However, despite being able to capture quantum higher order processes, these languages also capture linear maps that are not representative of quantum theory. Only routed circuits~\cite{Vanrietvelde2025consistentcircuits} provide a framework validating whether a diagram is a valid quantum process, though said framework only handles second order, and is complex and not complete: it might reject diagrams with a valid semantics, even though it accepts them presented with a different syntax.

In parallel, the mathematical representation of higher-order quantum processes still poses difficulties. Based on linear algebra, one of the main approaches followed in the literature is the use of process matrices~\cite{grenoble,Vanrietvelde2025consistentcircuits,Wechs}. While expressive, this model describes more than what is physically realizable~\cite{Vanrietvelde2025consistentcircuits}, and in general lacks modularity and compositionality, in the sense that in general process matrices cannot easily be decomposed into sub-process matrices. The other main approach focuses instead on categorical models based on the semantics of linear logic~\cite{prof,caus}. However, if these categorical approaches are compositional, they fail to provide an operational understanding of quantum superposition of executions.

\medskip
\noindent
\textbf{\sffamily Contributions.}~
In this paper, we solve this tension by providing a computational, modular and expressive interpretation of quantum higher-order processes, with a focus on unitarity.
More specifically, we introduce {\logicname} (\logicsymb), a type system and logic for higher order quantum computation based on intuitionistic multiplicative additive linear logic (\IMALL). 

\begin{itemize}
	\item \textbf{Structured for superposition.} {\logicsymb} presents a two-tiered proof system, separating general higher order processes and processes handling quantum data. The latter enriches its sequent structure with worlds that render the superposition of executions.
	
	\item \textbf{A safer interpretation.} We provide an operational semantics for terms and proofs of {\logicsymb} based on linear algebra. It ensures unitarity, conservation of unitarity, conservation of conservation of unitarity, etc. along with additional constructions. This approaches the theory of superchannels as laid out in~\cite{grenoble, Vanrietvelde2025consistentcircuits} as an analogue to positivity.
	
	\item \textbf{Expressivity.} {\logicsymb} is capable of expressing multiple ICOs, like $\SWITCH$, the Grenoble process, and challenging higher order processes, like controlling an unitary.
	
	\item \textbf{Cut elimination.} We show that {\logicsymb} enjoys sound cut elimination, as a first step to use it as a type system for a quantum programming language for higher order and indefinite causal orders.
\end{itemize}
\section{Quantum Computation and Indefinite Causal Ordering}
\label{sec:quantumcomp}

Unlike classical data that is often modeled by elements of sets (eg $\mathbb B = \{0, 1\}$), quantum data are represented as vectors of complex Hilbert spaces of norm $1$. In this paper, Hilbert spaces are always finite-dimensional, and we use the letter $\cH$ to refer to them. We write the set of linear maps between the vector space $\cV$ and $\cV'$ with $\mathrm{Lin}(\cV,\cV')$. 

The quantum equivalent of the bit, the \emph{quantum bit}, or \emph{qubit}, is a vector of dimension 2: it is written in the canonical basis as $\alpha \ket 0 + \beta \ket 1$, with $|\alpha|^2 + |\beta|^2 = 1$. Most quantum computing manipulates more than one qubit: the state of $n$ qubits lives in the \emph{Kronecker product} (or \emph{tensor}) of copies of $\bC^2$: it corresponds to linear combinations of strings of $n$ bits, and it lives in $\bC^{2^n}$. In general, the tensor product of two Hilbert spaces $\cH_a$ and $\cH_b$ is denoted with $\cH_a\otimes\cH_b$.

Beside initialization and reading, \emph{unitary operations} form the main class of quantum operations on quantum information. A unitary operation is a linear morphism $f\in\text{Lin}(\cH_a,\cH_b)$ such that $f^\dag \circ f= \Id A$ and $f \circ f^\dag = \Id B$, where $f^\dag$ is the conjugate transpose of $f$. The typical quantum algorithm is a sequence of elementary unitary operations called \emph{quantum gates} on one or two qbits at a time. A noteworhty combinator on unitary operations is the \emph{control}: given $f:\cH_a\to\cH_b$, we define $\text{C-}f:\bC^2\otimes\cH_a\to\bC^2\otimes\cH_b$ as follows: $\text{C-}f(\ket 0\otimes v) = \ket 0\otimes v$ and $\text{C-}f(\ket 1\otimes v) = \ket 1\otimes (f v)$. If the direct sum of two spaces is written $\oplus$, the operation $\text{C-}f$ can be seen as block-diagonal operation on $(\ket 0\otimes\cH_a)\oplus(\ket 1\otimes\cH_a)$: the subspace of all linear combinations of basis vectors starting with $\ket0$ (resp., starting with $\ket1$).

The tensor product of two unitaries, and their compositions are also unitaries. This makes them naturally suited to be represented as circuits. Quantum circuits are the standard notation to represent first order quantum programs: wires corresponds to quantum bits, and boxes corresponds to unitary operations, read from left to right. \Cref{fig:circuits} presents circuits of increasing complexity.

		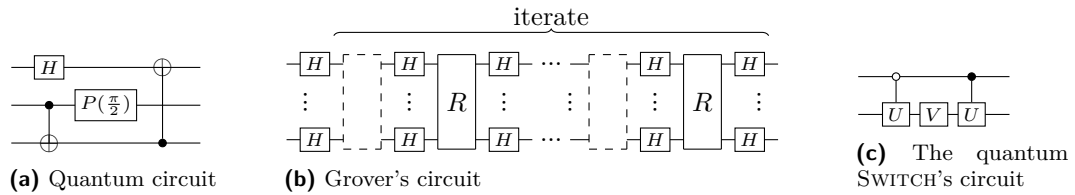
\begin{figure}[ht]
			\begin{subfigure}{0.25\columnwidth}
				\input{example-circuit.tikz}
				\caption{Quantum circuit}
				\label{fig:example-circuit}
			\end{subfigure}
			\hfill
			\begin{subfigure}{0.53\columnwidth}
				\input{Grover.tikz}
				\caption{Grover's circuit}
				\label{fig:grover}
			\end{subfigure}
			\hfill
			\begin{subfigure}{0.2\columnwidth}
				\input{switch-circuit.tikz}
				\caption{The quantum $\SWITCH$'s circuit}
				\label{fig:SwitchCirc}
			\end{subfigure}
			\caption{\ref{fig:example-circuit}: Circuit built over the $\langle H, \text{C-}X, P(\alpha)\rangle_{\alpha\in\mathbb R}$ gate set. The black dot corresponds to the controlling qubit. \ref{fig:grover}: Circuit for Grover's algorithm. The dashed boxes represent holes to be filled with the same oracle. The circuit $R$ can be built using the aforementioned gate set. \ref{fig:SwitchCirc}: Circuit representation of $\SWITCH(U,V)$ described below. The box $U$ connected to the first wire by a black (resp.~white) dot represents the operator $U$ controlled by $\ket1$ (resp.~$\ket0$).}
			\label{fig:circuits}
		\end{figure}

		Circuits with appropriate gate sets (such as that of \Cref{fig:example-circuit}) are universal for unitary computation \cite{nielsen_chuang_2010}. From there, one might then wonder how to take this formalism further, and describe operations that transform unitaries into unitaries. The naive approach would be to simply add ``holes'' to a standard circuit, and allow the user to plug any circuit with the right size into such holes. This is already done implicitly in many quantum algorithm, like Grover's algorithm (\Cref{fig:grover}) that requires an oracle.
		
		Such circuits with holes obey the laws laid out by quantum theory for higher order processes, especially in that, if the holes are filled with unitaries, then the result is obviously also unitary. However, they are not universal. Consider the following higher order algorithm. Given two unitaries $U$ and $V$ each acting on a single qubit, we can create the unitary $\SWITCH(U,V)$ \cite{switch} that acts on two qubits, where the first decides the order in which $U$ and $V$ are to be composed to act on the second, as described as follows:
		\begin{equation}\label{eq:qsw}\begin{cases}
			\SWITCH(U,V)(\ket 0\otimes \ket x) = \ket 0\otimes (VU \ket x)\\
			\SWITCH(U,V)(\ket 1\otimes \ket x) = \ket 1\otimes (UV \ket x)
		\end{cases}\end{equation}
  This operation can be regarded as a higher-order controlled operation. The effect of this operation is to place \emph{executions} in superposition.
  
		\Cref{fig:SwitchCirc} presents a circuit with the functional specification presented in \Cref{eq:qsw}. The controlled unitaries ensure that $U$ and $V$ are composed depending on the value of the control qubit. However, this representation of $\SWITCH$ is problematic from a theoretic point of view. Indeed, quantum theory requires higher order processes to also be linear in their higher order input. While this representation of $\SWITCH$ is in practice linear, in general circuits with two holes that must be filled with the same unitary are not.
		
		In fact, it is proven it is impossible to represent $\SWITCH$ as a circuit with holes using only two holes \cite{switch}. Moreover, $\SWITCH$ has been realized experimentally \cite{Goswani2018}, confirming that the circuit model is indeed too narrow. This is due to the ordering of $U$ and $V$ being dependent of the value of quantum data. In fact, $\SWITCH$ is an example of a wider class of processes, indefinite causal orders (ICOs), where the dependences between higher order input depend on quantum data. Other ICOs include the Grenoble process~\cite{grenoble} and the Lugano process~\cite{Vanrietvelde2025consistentcircuits} These ICOs are impossible to represent within the classic circuit model due to this challenging ordering \cite{switch}.
		
		Existing graphical languages \cite{Coecke2011,Clement2020,MW,Vanrietvelde2025consistentcircuits} can represent $\SWITCH$ and other indefinite causal orders. However, they tend to be too expressive: for instance, the ZX calculus and the Many Worlds calculus \cite{MW} also capture physically unrealistic processes, where data can be interpreted to travel back in time. To address this issue, Routed Circuits \cite{Vanrietvelde2025consistentcircuits} introduce a framework to check physical validity, but said framework is unwieldy, and rejects some semantically sound processes. The PBS-calculus \cite{Clement2020}, by construction, only represents physical processes that include indefinite causal orders. However, it only axiomatises the control part of the coherent control, and to do so requires a strict separation between the control and target systems. An alternative approach based on linear logic is the Caus construction~\cite{caus}, but the focus is the representation of general non-causal execution: the approach is more general than quantum computation, for example by including non-normalized processes, and it  does not come with an operational semantics.
    
    Finally, syntax-based approaches have also been explored to represent the purely unitary aspect of quantum computation, using either first-order approaches~\cite{grattage05functional,hainry:hal-05419254,barsse_et_al:LIPIcs.LICS.2026.14}, or lambda-calculus~\cite{diaz-caro2019realizability,hirata_et_al:LIPIcs.LICS.2026.57} or dedicated structures such as pattern-matching~\cite{sabry2018symmetric,iflet}. However, these works focus on expressivity from the perspective of the programmer, and in general do not explore the capabilities of quantum ICO.

\section{\logicsymb: a Logic Dedicated to ICO}
	\subsection{Terms}
	
	To determine how terms of \logicsymb\ are syntactically defined, we will consider how we wish to handle quantum data and processes. In keeping with our objective of safety over expressivity, we make one observation: it is useful to keep data and processes conceptually separate. Even though quantum theory allows us to view processes as states through the Choi-Jamiołkowski isomorphism \cite{jamiolkowski1972}, such states are typically not normalised, and still in a separate class from states representing quantum data. Furthermore, while freely allowing processes to be entangled or in superposition could lead to a richer system, we will see this degree of expressivity is not necessary to study most known higher order processes. In a more practical sense, our focus on unitaries requires any transformation of quantum data to be reversible, but the fundamental higher order operation of application is fundamentally non-reversible. Therefore, we strictly separate our terms into two layers: \emph{ground terms} ($A, B, C,...$) will be terms that represent quantum data, and \emph{higher terms} ($F, G, H,...$) will be terms that represent processes.
	
	To build higher terms, we will use Girard's Multiplicative Additive Linear Logic \cite{LL} as a base. Its usefulness as a starting point to describe quantum data and processes is well known, and leveraged in~\cite{MW,caus}. We deviate slightly by using the intuitionistic version of Multiplicative Additive Linear Logic, \IMALL, since intuitionistic logics are better adapted for type systems, following the Curry Howard correspondence. In particular, we want to study three basic kinds of composition of types. Given a term $F$ and a type $G$, we can refer to the linear maps between the values of $F$ and those of $G$, denoted $F \multimap G$. We can also refer to a pair composed of a process of type $F$ and a process of type $G$ by $F \otimes G$. Finally, we might refer to a process that behaves as $F$, or as $G$, but is only a single process (like a block diagonal unitary). We will denote these processes $F \with G$. Contrary to \IMALL, we will not include an $\oplus$ connective for an uncontrolled sum type over higher terms. Finally, we include a connective not found in Linear Logic to link higher terms to ground terms. If $A$ and $B$ are ground terms, $A \uarrow B$ is a higher term that represents unitaries between $A$ and $B$. 
	
	Ground terms are also derived from linear logic's connectives. We include the atoms $\unit$, representing $\bC$, and a set of other literals $\mathfrak A$ which will be used to represent states of arbitrary Hilbert spaces. $A \oplus B$ will represent superpositions between $A$ and $B$. For example, $\unit \oplus \unit$ represents qubits. For simplicity, in the following we may denote the $n$-fold sum of $\unit$ as $\underline{n} := \overbrace{\unit\oplus...\oplus\unit}^n$, so the qubit type may be written $\underline2$. Ground types can also be tensored, though their tensor allows entanglement. As such, it will be denoted $A \boxtimes B$, to distinguish it from the higher order tensor that only allows separate pairs. Note that under this grammar, the same underlying Hilbert space can be represented in different ways: for example $\unit \oplus \unit \oplus \unit \oplus \unit$ and $(\unit \oplus \unit) \boxtimes (\unit \oplus \unit)$. The full grammar is given in \Cref{fig:termGrammar}.

	\begin{figure}[h]
	\centering
	\fbox{
		$\quad\begin{aligned}
			A,B&:=\unit\mid a \mid A \boxtimes B \mid A \oplus B\\
			F,G&:=A \uarrow B \mid F \otimes G \mid F \multimap G \mid F \with G\\
		\end{aligned}\quad$
	}
	\caption{Grammar for terms of \logicsymb, where $a \in \mathfrak A$ a set of literals.}
	\label{fig:termGrammar}
	\end{figure}
	
	\begin{example}[The Quantum Switch]
		$\SWITCH$\ over a data type $A$ would have type
		\[((A \uarrow A) \otimes (A \uarrow A)) \multimap (((\unit \oplus \unit)\boxtimes A) \uarrow ((\unit \oplus \unit)\boxtimes A))\]
		
		It is a higher order process that takes two unitaries of $A$ to itself: $(A \uarrow A) \otimes (A \uarrow A)$, to a single unitary from a qubit wire tensored with an $A$ wire to the same space, $((\unit \oplus \unit)\boxtimes A) \uarrow ((\unit \oplus \unit)\boxtimes A)$.
	\end{example}
	
	\begin{example}[Controlled Unitaries]
		Normally, the process that transforms an unitary into its controlled version is not linear. However, it can be made linear by modifying the representation of the unitary so that it also describes its behavior on a vacuum state, and guarantees that if the process is not in use, it does not interfere with the rest of the computation~\cite{Chiribella_2019}. This is expressed by the connector $\with$, composing its ``on'' behavior of the unitary as $A \uarrow A$, and its ``off'' behavior over the one-dimensional vacuum state as $\unit \uarrow \unit$, and we obtain this type for the controlled unitary:
		\[((A \uarrow A) \with (\unit \uarrow \unit)) \multimap  ((\unit \oplus \unit)\boxtimes A) \uarrow ((\unit \oplus \unit)\boxtimes A)\]
	\end{example}
	
	Formally defining the interpretation of terms of \logicsymb\ presents a slight challenge. Usually, higher order quantum maps are described in the mixed state formalism as completely complete positivity preserving maps~\cite{grenoble}, but the presence of more fine grained types like $\with$, and our focus on pure state computing makes lifting this definition difficult. To remain within the framework of linear algebra, we adopt a laxer definition for our values:
	
	\begin{definition}
		\label{def:interp}
		Let $T$ a term of \logicsymb. We define its \emph{space} $\cV_T$ and its set of \emph{values} $\sem T$ inductively as follows, where $A, B$ are ground terms, $F, G$ are higher terms:
		\begin{itemize}
			\item For each $a \in \mathfrak A$, we fix $\cH_a$ a finite dimentional Hilbert space.
			\item $\cV_{\unit} = \bC$, $\cV_{a} = \cH_a$, $\cV_{A \boxtimes B} = \cV_A \otimes \cV_B$, $\cV_{A \oplus B} = \cV_A \oplus \cV_B$.
			\item $\sem A = \{\ket x \in \cV_A ; \braket{x\mid x} = 1 \}$.
			\item $\cV_{A \uarrow B} = \mathrm{Lin}(\cV_A, \cV_B)$, $\sem{A \uarrow B} = \{f \in \cV_{A \uarrow B}; f^\dag f = \Id{\cV_A}, ff^\dag = \Id{\cV_B}\}$.
			\item $\cV_{F \otimes G} = \cV_F \otimes \cV_G$, $\sem{F \otimes G} = \{f \otimes g; f \in \sem F, g \in \sem G\}$.
			\item $\cV_{F \multimap G} = \mathrm{Lin}(\cV_F, \cV_G)$, $\sem{F \multimap G} = \{s \in \cV_{F \multimap G}; \forall f \in \sem F, s(f) \in \sem G\}$.
			\item $\cV_{F \with G} = \cV_F \oplus \cV_G$, $\sem{F \with G} = \{(f,g); f \in \sem F, g \in \sem G\}$.
		\end{itemize}
		
		For convenience when we will define the interpretation of proofs, we also extend $\cV_-$ and $\sem -$ to ordered multisets of ground terms and ordered multisets of higher terms:
		\begin{itemize}
			\item In the context of ground terms, $\cV_\emptyset = \bC$, and $\sem \emptyset = \{z\in\bC; |z| = 1\}$.
			\item $\cV_{A, \Gamma} = \cV_A \otimes \cV_\Gamma$ if $\Gamma \neq \emptyset$
			\item In the context of higher terms, $\cV_\emptyset = \bC$ and $\sem \emptyset = \{1\}$
			\item $\cV_{F, \Phi} = \cV_F \otimes \cV_\Phi$, $\sem{F, \Phi} = \{f \otimes \varphi; f \in \sem F, \varphi \in \sem \Phi\}$ if $\Phi \neq \emptyset$.
		\end{itemize}
	\end{definition}
	
	\subsection{Proof Structure}

		Like terms, proofs of \logicsymb\ will be divided in two fragments, the $\lStile$ fragment describing proofs that have a higher term as a conclusion, and the $\uStile$ fragment describing proofs that have a ground term as a conclusion. The $\lStile$ fragment has a typical construction, where sequents are of the form $\Phi \lStile F$, where $\Phi$ is an ordered multiset of higher terms, and $F$ is a higher term, but $\uStile$ sequents are considerably more complex, due to the more demanding nature of unitaries.
		
		In Linear Logic, the \rOpR\ rule allows to discard any of the two sides of the conclusion if its topmost connective is $\oplus$, reminiscent of an injection. Dually, the \rOpL\ takes a hypothesis of the form $A \oplus B$, and splits the proof into two subproofs, one using $A$ and one using $B$, and the subproofs become independent. This behavior makes guaranteeing unitarity impossible: injection is inherently not surjective, and \rOpL\ would have us sum two processes without guaranteeing any kind of orthogonality that would conserve unitarity. We resolve this by introducing a notion of \emph{worlds} of hypotheses, separated by pipes ($\mid$), to represent the superposition of executions a process might find itself in. As an example, \cref{fig:Xproof} is a simplified derivation representing the $X$ gate. At the root, there is only a single world, containing the term $\unit_{\ket 0} \oplus \unit_{\ket 1}$. That world is then split in two worlds by the \rOpL\ rule, instead of creating two independent subproofs. These worlds are then redistributed to each side of the output by \rOpR, after being permuted by \rExchUnW\ so that $\ket 0$ maps to $\ket 1$ and vice versa. Unlike Linear Logic, it is \rOpR\ that is a binary rule, representing the fact that the direct sum of two unitaries remains an unitary, and \rOpL\ that is an unary rule.
		
		\begin{figure}[h]
		\begin{center}
		\fbox{
		\begin{prooftree}
			\infer0[\rAx]{\unit_{\ket 1} \uStile \unit_{\ket 0}}
			
			\infer0[\rAx]{\unit_{\ket 0} \uStile \unit_{\ket 1}}
			
			\infer2[\rOpR]{\unit_{\ket 1} \mid \unit_{\ket 0} \uStile \unit_{\ket 0} \oplus \unit_{\ket 1}}
			
			\infer1[\rExchUnW]{\unit_{\ket 0} \mid \unit_{\ket 1} \uStile \unit_{\ket 0} \oplus \unit_{\ket 1}}
			
			\infer1[\rOpL]{\unit_{\ket 0} \oplus \unit_{\ket 1} \uStile \unit_{\ket 0} \oplus \unit_{\ket 1}}
		\end{prooftree}}
		\end{center}
		
		\caption{Proof representing the $X$ gate. Subscripts have been added for clarity.}
		\label{fig:Xproof}
		\end{figure}
		
		However, by introducing worlds to solve issues of orthogonality in subproofs, we have introduced a new difficulty to depict higher order, where higher terms and ground terms have to coexist in a $\uStile$ sequent. Consider the naive sequent $A \uarrow A, A \oplus A \uStile A \oplus A$. If we were to simply apply \rOpL, we would obtain the sequent $(A \uarrow A, A )\mid (A \uarrow A, A) \uStile C$. In this second sequent, there is no witness to the fact that the two instances of $A \uarrow A$ implicitly represent the same object. The solution is to decorate all higher terms in a $\uStile$ sequent with a label, to witness the relationship between higher terms between worlds.
		
		\begin{definition}[Labeled Terms]
			Let $\mathfrak L$ be a set of label literals, distinct from $\mathfrak A$ the set of term literals. We define the set of labels as:
			\[\alpha, \beta := \ell \mid \otimes_1 \alpha \mid \otimes_2 \alpha \mid \oplus_1 \alpha \mid \oplus_2 \alpha \]
			where $\ell \in \mathfrak L$.
			\emph{Labeled terms} are pairs composed of a higher term $F$ and a label $\alpha$. They are denoted $F^\alpha$ in general, or $F \stackrel{\alpha}{\uarrow} G$, or $F \stackrel{\alpha}{\multimap} G$. A multiset of labeled higher terms will usually be denoted $\widehat \Phi$ or $\widehat \Psi$.
			We also define the unlabeling operation $\unlab$, that removes all labels from a multiset of labeled higher terms.		
			If $\Phi$ is an ordered multiset of higher terms, and $\widehat \Psi$ is an ordered multiset of labeled higher terms, we write $\Phi \labrel \widehat \Psi$ if $\Phi = \unlab(\widehat \Psi)$ and each label in $\widehat \Psi$ is a distinct literal.
		\end{definition}
		
    \begin{convention}
      For clarity, we separate the higher terms from ground terms in a world by a semicolon. As an example of how labels are manipulated by the rules of \logicsymb, consider the previous naive fragment. With labels, it would be rendered as this proof fragment:
      \begin{center}
        \begin{prooftree}
          \hypo{A \stackrel{\alpha}{\uarrow} A; A \mid A \stackrel{\alpha}{\uarrow} A; A \uStile C}
          \infer1[\rOpL]{A \stackrel{\alpha}{\uarrow} A; A \oplus A \uStile C}
        \end{prooftree}
      \end{center}
      Now keeping track of the origin of the two occurrences of $A\uarrow A$ at the top.
		\end{convention}
		
    The exact behavior of labels is described by the formal rules of \logicsymb\ in \Cref{fig:UnitRules2}.

		Worlds can now be formally defined.
		
		\begin{definition}[Worlds and world sets]
			A \emph{world} $\omega = (\widehat \Phi; \Gamma)$ is a pair composed of an ordered multiset of labeled higher terms $\widehat \Phi$, and an ordered multiset of ground terms $\Gamma$.\\
			A \emph{world set} $\Omega$ is an indexed, finite sequence of worlds.\\
			If $\Omega_1$ and $\Omega_2$ are world sets, $\Omega_1 \mid \Omega_2$ is defined as their concatenation.\\
			We define the product of two world sets as $(\widehat {\Phi_i}; \Gamma_i)_{i \in I} \times (\widehat {\Psi_j}; \Delta_j)_{j \in J} := (\widehat{\Phi_i}, \widehat{\Psi_j}; \Gamma_i, \Delta_j)_{ij \in I \times J}$.
		\end{definition}
		
		Worlds present a final difficulty with higher order. To deal with arbitrarily high order, we need to handle $\multimap$ terms within the $\uStile$ fragment. In particular, we need to specify how labels are to be used when a new higher term is generated as a result of a rule. In keeping with the intuition that two terms share the same label if they are semantically equal, whenever a term is created by applying a $\multimap$ rule, it must be either given a fresh label, or be generated by the same terms, with the same labels, as another, existing label. This is formalized by the concept of promises:
		
		\begin{definition}[Promises]
			A promise is a tuple $(\widehat \Phi, \pi_L, (F \multimap G)^\alpha, G^\beta)$ where:
			\begin{itemize}
				\item $\widehat \Phi$ is a labelled ordered multiset of higher terms;
				\item $\pi_L$ is a proof of $\unlab(\widehat \Phi) \lStile F$;
				\item $\alpha, \beta$ are a labels that do not appear in $\widehat \Phi$.
			\end{itemize}
			$G^\beta$ is called the \emph{output term} of the promise.
			
			A \emph{promise set} $P$ is a set of promises such that, for each label $\beta$, there is at most one $G$ such that there exists a promise in $P$ with output $G^\beta$, and no two promises share the same output term. For example, a promise set cannot contain both a a promise with output $G^\alpha$ and a promise with output $H^\alpha$.
		\end{definition}
		
		\begin{definition}[$\uStile$ sequents]
			A $\uStile$ sequent $P:\Omega \uStile C$ is a triple consisting of a promise set $P$, a non-empty set of worlds $\Omega$ and a ground term $C$, the conclusion.
		\end{definition}
		
		\begin{example}
                  An example of a $\uStile$ sequent is as follow.s.
                  \begin{center}
                    \scalebox{.8}{$\underbrace{\left(F^\alpha, \begin{prooftree}\infer0{F \lStile F}\end{prooftree}, F \stackrel{\beta}{\multimap} ((a \oplus a) \uarrow b), (a \oplus a) \stackrel \gamma \uarrow b\right)}_{\text{promise set}}: \underbrace{((a \oplus a) \stackrel \gamma \uarrow b; a) \mid (F \stackrel{\beta}{\multimap} ((a \oplus a) \uarrow b), F^\alpha; a)}_{\text{world set}} \uStile b$}
                  \end{center}
		\end{example}
		
		For labels to fulfil their role of witnesses of some sort of equality between worlds, they need to be coherent with each other. Intuitively, if $F^{\oplus_1 \alpha}$ appears in the sequent, it must imply any higher term with the label $\alpha$ to be of the form $(F \with G)^\alpha$.
		
		\begin{definition}[Consistent Labeling]
			Let $P: \Omega \uStile C$ a sequent.
			
			We say that $P: \Omega \uStile C$ is \emph{consistently labeled} if, for all labels $\alpha$, and all label prefixes $p$, if $F^\alpha$ and $G^{p \cdot \alpha}$ appear in the sequent, then $G = {\downarrow_p} F$, with $\downarrow_p$ a downcasting operation defined as follows:
			\begin{itemize}
				\item ${\downarrow_{\varepsilon}} F = F$, where $\varepsilon$ is the empty prefix.
				\item ${\downarrow_{p\cdot \otimes_i}}(F_1 \otimes F_2) = {\downarrow_p} F_i$
				\item ${\downarrow_{p\cdot \oplus_i}}(F_1 \with F_2) = {\downarrow_p} F_i$.
			\end{itemize}
			We also extend the use of $\downarrow_p$ to values of $\sem F$. For instance, ${\downarrow_{\oplus_2 \otimes_1}}((f_1, f_2) \otimes g) = f_2$.
		\end{definition}
		
		We will now only consider sequents that are consistently labeled.

		\subsection{Rules}
		From this, we can define a \logicsymb\ proof in the usual way: as a tree of sequents inductively constructed from a set of rules.
		
		\begin{figure}[t]
      \begin{subfigure}[b]{.56\textwidth}
			\scalebox{.6}{\fbox{\qquad
				$\addtolength{\jot}{.75em}
				\begin{aligned}
				\begin{prooftree}
					\hypo{\Phi, G, F, \Psi \lStile H}
					\infer1[\rExchLin]{\Phi, F, G, \Psi \lStile H}
				\end{prooftree}
				& &
				\begin{prooftree}
					\hypo{\Phi \labrel \widehat \Phi}
					\hypo{\emptyset:(\widehat \Phi; A) \uStile B}
					\infer2[\rUnitRight]{\Phi \lStile A \uarrow B}
				\end{prooftree} \\
				\begin{prooftree}
					\hypo{F, G, \Phi \lStile H}
					\infer1[\rOtLinL]{F \otimes G, \Phi \lStile H}
				\end{prooftree}
				& &
				\begin{prooftree}
					\hypo{\Phi \lStile F}
					\hypo{\Psi \lStile G}
					\infer2[\rOtLinR]{\Phi, \Psi \lStile F \otimes G}
				\end{prooftree}\\
				\begin{prooftree}
					\hypo{\Phi \lStile F}
					\hypo{G, \Psi \lStile H}
					\infer2[\rArLinL]{\Phi, F \multimap G, \Psi \lStile H}
				\end{prooftree}
				& &
				\begin{prooftree}
					\hypo{F, \Phi \lStile G}
					\infer1[\rArLinR]{\Phi \lStile F \multimap G}
				\end{prooftree}\\
				\begin{prooftree}
					\hypo{F, \Phi \lStile H}
					\infer1[\rWiLinL 1]{F \with G, \Phi \lStile H}
				\end{prooftree}
				& &
				\begin{prooftree}
					\hypo{G, \Phi \lStile H}
					\infer1[\rWiLinL 2]{F \with G, \Phi \lStile H}
				\end{prooftree} \\
				\begin{prooftree}
					\hypo{\Phi \lStile F}
					\hypo{\Phi \lStile G}
					\infer2[\rWiLinR]{\Phi \lStile F \with G}
				\end{prooftree}
				& &
				\begin{prooftree}
					\hypo{\Phi \lStile F}
					\hypo{F, \Psi \lStile H}
					\infer2[\rCutLin]{\Phi, \Psi \lStile H}
				\end{prooftree}
				\end{aligned}$\qquad
			}}
			\caption{Rules}
			\label{fig:LinRules}
    \end{subfigure}
    \begin{subfigure}[b]{.435\textwidth}

			\scalebox{.6}{\fbox{
          $\def\mynl{\\[1ex]}
          \begin{array}{l}
            \sem{\text{\rExchLin}} = \sem {\pi'} \circ (\Id{\Phi} \otimes \sigma_{F, G} \otimes \Id{\Psi})
            \mynl
            \sem{\text{\rUnitRight}} = \sem{\pi'}^{(1)}
            \mynl
            \sem{\text{\rOtLinL}} = \sem{\pi'}
            \mynl
            \sem{\text{\rOtLinR}} = \sem{\pi_1} \otimes \sem{\pi_2}
            \mynl
            \sem{\text{\rArLinL}} = \sem{\pi_2} \circ ((\app_{F, G} \circ (\sem {\pi_1} \otimes \Id{F \multimap G}))\otimes \Id{\Psi})
            \mynl
            \sem{\text{\rArLinR}} = \varphi \mapsto (f \mapsto \sem{\pi'}(f \otimes \varphi))
            \mynl
            \sem{\text{\rWiLinL x}} = \sem{\pi'}\circ(\proj{F_1, F_2}{x} \otimes \Id{\Phi})
            \mynl
            \sem{\text{\rWiLinR}} = \varphi \mapsto (\sem{\pi_1}(\varphi), \sem{\pi_2}(\varphi))
            \mynl
            \sem{\text{\rCutLin}} = \sem{\pi_2}\circ(\sem {\pi_1} \otimes \Id{\Psi})
          \end{array}$
      }} 
      \caption{Semantics (See Def.~\ref{def:LinSem})}
			\label{fig:LinSem}
        \end{subfigure}
        \caption{The $\lStile$ fragment}
			\label{fig:AllOfLin}
  \end{figure}
		The $\lStile$ fragment in \Cref{fig:LinRules} is essentially a fragment of \IMALL, where the $\oplus$ connective has been removed, and with the addition of the \rUnitRight\ rule. Notice the lack of axiom rule. Like in linear logic, it is derivable as long as it exists for literals in $\uStile$.
		\begin{figure}[t]
		\begin{subfigure}[b]{.56\textwidth}
      \scalebox{.6}{\fbox{\parbox{12cm}{
					\centering
					$\addtolength{\jot}{0.4em}\begin{aligned}
					\begin{prooftree}
						\infer0[\rAx]{P: (\emptyset; A) \uStile A}
					\end{prooftree}
					& \hspace*{0em} &
					\begin{prooftree}
						\hypo{P: \Omega \mid \omega' \mid \omega \mid \Omega' \uStile C}
						\infer1[\rExchUnW]{P: \Omega \mid \omega \mid \omega' \mid \Omega' \uStile C}
					\end{prooftree} 
\\
					\begin{prooftree}
						\hypo{P: (\widehat \Phi; \Gamma) \mid \Omega \uStile C}
						\infer1[\rOneL]{P: (\widehat \Phi; \unit, \Gamma) \mid \Omega \uStile C}
					\end{prooftree}
					& &
					\begin{prooftree}
						\hypo{P: (\widehat \Phi; \Gamma, B, A, \Delta) \mid \Omega \uStile C}
						\infer1[\rExchUnG]{P: (\widehat \Phi; \Gamma, A, B, \Delta) \mid \Omega \uStile C}
					\end{prooftree}
\\
					\begin{prooftree}
						\hypo{P:(\widehat \Phi; A, B, \Gamma) \mid \Omega \uStile C}
						\infer1[\rBtL]{P:(\widehat \Phi; A \boxtimes B, \Gamma) \mid \Omega \uStile C}
					\end{prooftree}
					& &
					\begin{prooftree}
						\hypo{P:\Omega_1 \uStile A}
						\hypo{P:\Omega_2 \uStile B}
						\infer2[\rBtR]{P:\Omega_1 \times \Omega_2 \uStile A \boxtimes B}
					\end{prooftree}
\\
					\begin{prooftree}
						\hypo{P:(\widehat \Phi; A, \Gamma) \mid (\widehat \Phi; B, \Gamma) \mid \Omega \uStile C}
						\infer1[\rOpL]{P:(\widehat \Phi; A \oplus B, \Gamma) \mid \Omega \uStile C}
					\end{prooftree}
					& &
					\begin{prooftree}
						\hypo{P: \Omega_1 \uStile A}
						\hypo{P: \Omega_2 \uStile B}
						\infer2[\rOpR]{P: \Omega_1 \mid \Omega_2 \uStile A \oplus B}
					\end{prooftree}
\\
            \begin{prooftree}
						\infer0[\rOneR]{P:(\emptyset; \emptyset) \uStile \unit}
					\end{prooftree}
\quad
					\begin{prooftree}
						\hypo{P:\Omega \uStile C}
						\infer1[\rPhTh]{P:\Omega \uStile C}
					\end{prooftree}
					& &
					\begin{prooftree}
						\hypo{P:(\widehat \Phi; (\unit \oplus \unit), \Gamma) \mid \Omega \uStile C}
						\infer1[\rRotTh]{P:(\widehat \Phi; \Gamma)\mid (\widehat \Phi; \Gamma) \mid \Omega \uStile C}
					\end{prooftree}
				\end{aligned}$
\\[0.5em]
				$\addtolength{\jot}{0.4em}\begin{aligned}
				\begin{prooftree}
					\hypo{P:(\widehat{\Phi_i}; \Gamma_i)_{i \in I} \uStile A}
					\hypo{P:(\widehat \Psi; B, \Delta) \mid \Omega \uStile C}
					
					\infer2[\rUnitUn]{(P:\widehat{\Phi_i}, A \stackrel{\alpha}{\uarrow} B, \widehat \Psi; \Gamma_i, \Delta)_{i \in I} \mid \Omega \uStile C}
				\end{prooftree}\\
				\begin{prooftree}
					\hypo{P:(\widehat {\Phi_i}; \Gamma_i)_{i \in I} \uStile A}
					\hypo{P:(\widehat \Psi; A, \Delta) \mid \Omega \uStile C}
					\infer2[\rCutUnG]{P:(\widehat {\Phi_i}, \widehat \Psi; \Gamma_i, \Delta)_{i \in I} \mid \Omega \uStile C}
				\end{prooftree}
			\end{aligned}$}}}
			\caption{Fragment for ground terms}
			\label{fig:UnitRules1}
		\end{subfigure}
		\begin{subfigure}[b]{.435\textwidth}
			\scalebox{.6}{\fbox{
				\parbox{9.5cm}{\centering
				$\addtolength{\jot}{0.4em}\begin{aligned}
					\begin{prooftree}
						\hypo{P: (\widehat \Phi, G^\beta, F^\alpha, \widehat \Psi; \Gamma) \mid \Omega \uStile C}
						\infer1[\rExchUnH]{P: (\widehat \Phi, F^\alpha, G^\beta, \widehat \Psi; \Gamma) \mid \Omega \uStile C}
					\end{prooftree}
					\\ 
					\begin{prooftree}
						\hypo{P:(F^{\otimes_1 \alpha}, G^{\otimes_2 \alpha}, \widehat \Phi; \Gamma) \mid \Omega \uStile C}
						\infer1[\rOtUn]{P:((F\otimes G)^\alpha, \widehat \Phi; \Gamma)}
					\end{prooftree} \\
					\begin{prooftree}
						\hypo{P:(F^{\oplus_1 \alpha}, \widehat \Phi; \Gamma) \mid \Omega \uStile C}
						\infer1[\rWiUn{1}]{P:((F\with G)^\alpha, \widehat \Phi; \Gamma) \mid \Omega \uStile C}
					\end{prooftree}
					\\ 
					\begin{prooftree}
						\hypo{P:(G^{(\oplus_2\alpha)}, \widehat \Phi; \Gamma) \mid \Omega \uStile C}
						\infer1[\rWiUn{2}]{P:((F \with G)^\alpha, \widehat \Phi; \Gamma) \mid \Omega \uStile C}
					\end{prooftree}
				\end{aligned}$\\[.4em]
				$\begin{aligned}
				\begin{prooftree}

          \hypo{\unlab(\widehat \Phi) \lStile F}
					\hypo{P \sqcup p_L: (G^\beta, \widehat \Psi; \Gamma) \mid \Omega \uStile C}
					\hypo{\beta \ \text{fresh}}
					\infer3[\rArUnNew]{P:(\widehat \Phi, F \stackrel{\alpha}{\multimap} G, \widehat \Psi; \Gamma) \mid \Omega \uStile C}
				\end{prooftree}\\
				\begin{prooftree}
					\hypo{P \sqcup p_L: (G^\beta, \widehat \Psi; \Gamma) \mid \Omega \uStile C}
					\infer1[\rArUnPro]{P\sqcup p_L:(\widehat \Phi, F \stackrel{\alpha}{\multimap} G, \widehat \Psi; \Gamma) \mid \Omega \uStile C}
				\end{prooftree}\\
				\begin{prooftree}
					\hypo{\unlab(\widehat \Phi) \lStile F}
					\hypo{P:(F^\alpha, \widehat{\Psi_i}; \Gamma_i)_{i \in I} \mid \Omega \uStile C}
					\hypo{\alpha \ \text{fresh}}
					\infer3[\rCutUnH]{P:(\widehat \Phi, \widehat{\Psi_i}; \Gamma_i)_{i \in I} \mid \Omega \uStile C}
				\end{prooftree}
				\end{aligned}$}}}
			\caption{Fragment for higher terms}
			\label{fig:UnitRules2}
		\end{subfigure}
    \caption{Rules for $\uStile$. In Fig.~\ref{fig:UnitRules2}, $p_L = (\widehat \Phi, \pi_L, \alpha, \beta)$.}
    \label{fig:UnitRules}
  \end{figure}
  
The rules of the $\uStile$ fragment in \Cref{fig:UnitRules1} and \Cref{fig:UnitRules2} describe the handling of ground terms and of worlds. Worlds can be permuted exactly like terms, as described by rule \rExchUnW. Worlds are created by the \rOpL\ rule, which splits a superposition of quantum data into two distinct worlds. This is what enables us to render quantum control, by applying different rules to distinct worlds.
		
		\rBtR\ is a generalisation of the fact that the tensor product of two unitaries is a unitary, and \rOpR\ that a block diagonal of unitaries is unitary. Contrary to \IMALL, \rOpR\ is a binary rule, to take into account that unitaries must be invertible.
				
		\rPhTh\ and \rRotTh\ are necessary for \logicsymb\ to be complete for unitaries. They represent the phase change, applied to the first world, and rotation, scrambling the first two worlds together. They are in a sense ``free'' unitaries that we can generate on demand, provided we know the angle used. They are both families of rules indexed by an angle $\theta$, fixed in the proof.
		
		\rUnitUn\ describes the application of a unitary during our procedure, consuming it as a higher order input. Since we aim to be reversible, the worlds used to create the unitary's input must become indistinguishable after the unitary is employed.
		
		\rCutUnG\ is much the same, combining existing data together without needing a unitary to apply. It is only one of the two cut rules of the $\uStile$ fragment, the other being described below in \Cref{fig:UnitRules2}.

		\rArLinL\ has been split in two rules to interact properly with labels: \rArUnNew\ creates an entirely new label, and adds a promise to the promise set to allow new creation of that label under the right condition. \rArUnPro\ makes use of that promise to create a new instance of that label.
		
		\rCutUnH\ is quite different from the arrow rule it is supposed to be associated with. It is not split in two, and affects multiple worlds at once to create the new higher order term (though these worlds do not interact between themselves). We could indeed have it behave like \rArUnNew\ and \rArUnPro, and adding cut proofs to promises. However, this approach renders cut elimination much more difficult to state: it would be impossible to entirely eliminate cuts from some $\uStile$ proofs where there are cut-rules employing a promise, but where the root sequent already carries terms with the same label. While the intention behind cut elimination could be conserved through a careful restating, to not get lost in highly specific statements of cut-elimination properties, \rCutUnH\ has been made less symmetrical to the $\multimap$ rules.
		
		We may also ask why not fuse \rArUnNew\ and \rArUnPro\ in a single multi-world rule like \rCutUnH. The consequences of this change are more impactful, and this version of \logicsymb\ loses the cut-elimination property due to indefinite causal orders. The ordering between the terms generated by cut rules and existing $\multimap$ terms might present incompatible orderings depending on worlds, which leads to label inconsistencies when we attempt to eliminate cut.

		\begin{example}
      In the proof of $\SWITCH$ shown in \Cref{fig:switchProof}, worlds are split according to the control qubit, and the two input unitaries are applied in one order or the other depending on the world.
		\end{example}

		\begin{figure}[t]
      \def\myrel#1{{\stackon[-.2ex]{$\uarrow$}{\scalebox{.7}{$#1$}}}}
      \centering
    \scalebox{.7}{\begin{prooftree}
			\infer0[\rAx]{A \uStile A}
			\infer0[\rAx]{A \uStile A}
			\infer0[\rAx]{A \uStile A}
			\infer0[\rAx]{A \uStile A}
			
			\infer0[\rAx]{\unit \uStile \unit}
			\infer0[\rAx]{\unit \uStile \unit}
			\infer2[\rOpR]{\unit \mid \unit \uStile \underline 2}
			
			\infer0[\rAx]{A \uStile A}
			
			\infer2[\rBtR]{\unit, A \mid \unit, A \uStile \underline 2 \boxtimes A}
			\infer2[\rUnitUn]{\unit, A \mid A \myrel{f} A; \unit, A \uStile \underline 2 \boxtimes A}
			\infer2[\rUnitUn]{\unit, A \mid A \myrel{f} A, A \myrel{g} A; \unit, A \uStile \underline 2 \boxtimes A}
			\infer2[\rUnitUn]{A \myrel{g} A; \unit, A \mid A \myrel{f} A, A \myrel{g} A; \unit, A \uStile \underline 2 \boxtimes A}
			\infer2[\rUnitUn]{A \myrel{f} A, A \myrel{g} A; \unit, A \mid A \myrel{f}, A \myrel{g} A; \unit, A \uStile \underline 2 \boxtimes A}
			\infer1[\rOpL]{A \myrel{f} A, A \myrel{g} A; \underline 2, A \uStile \underline 2 \boxtimes A}
			\infer1[\rBtL]{A \myrel{f} A, A \myrel{g} A; \underline 2 \boxtimes A \uStile \underline 2 \boxtimes A}
			\infer1[\rUnitRight]{A \uarrow A, A \uarrow A \lStile (\underline 2 \boxtimes A) \uarrow (\underline 2 \boxtimes A)}
			\infer1[\rOtLinL]{(A \uarrow A) \otimes (A \uarrow A) \lStile (\underline 2 \boxtimes A) \uarrow (\underline 2 \boxtimes A)}
			\infer1[\rArLinR]{\lStile ((A \uarrow A) \otimes (A \uarrow A)) \multimap ((\underline 2 \boxtimes A) \uarrow (\underline 2 \boxtimes A))}
		\end{prooftree}}
		\caption{Proof of $\SWITCH$}
		\label{fig:switchProof}
		\end{figure}

\section{Interpretation}
	
		We now define how to interpret \logicsymb\ proofs. The interpretation of the $\lStile$ fragment is straightforward, where the interpretation of proofs is simply linear maps. The multiple-world paradigm of the $\uStile$ requires a more complex approach. Similarly to circuits with holes, interpreting a proof with multiple worlds as a single map implies that some higher order input might not be used in all cases (when a label is present in one world and not another), or that two ``holes'' must be filled with the same value (when a label is present in multiple worlds). This inherently forbids a $\uStile$ proof being seen as a single linear map. Instead, we interpret these proofs as tuples of linear maps, one for each world. The intuition being that from higher order input, these interpretations output isometries that can be recomposed into an unitary.

\subsection{Semantics}

		\begin{notation}
			We often write $T$ instead of $\cV_T$ when the context is clear, like in the case of $\Id T$ instead of $\Id {\cV_T}$.
			We use the following notations for predefined maps. If $V, V', W, W'$ are vector spaces then:
		  The identity over $V$ is written $\Id V$.
		  The swap map is $\sigma_{V, W}: V \otimes W \rightarrow W \otimes V$.
			The application map is $\app_{V, W}: V \otimes \mathrm{Lin}(V, W) \rightarrow W$.
			The left and right injections are respectively $\inj{V_1,V_2}{i}: V_i \rightarrow V_1 \oplus V_2$ for $i \in \{1, 2\}$
		  The  left and right orthogonal projections are respectively $\proj{V_1, V_2}{i}: V_1 \oplus V_2 \rightarrow V_i$ for $i \in \{1, 2\}$.
			The map such that $1 \otimes x \mapsto x$ is denoted with $\lambda_V : \bC \otimes V \rightarrow V$.
			Note that when $V_1$ and $V_2$ are Hilbert spaces, $(\inj{V_1,V_2}{i})^\dag = \proj{V_1, V_2}{i}$.
		\end{notation}
		
		\begin{definition}[Semantics of \logicsymb]
			Let $\pi_L$ a proof of $\Phi \lStile H$, and $\pi_U$ a proof of $P:(\widehat{\Phi_i}; \Gamma_i)_{1 \leq i \leq n} \uStile C$.
			\begin{itemize}
				\item $\sem {\pi_L}$ is a linear map between $\cV_\Phi$ and $\cV_H$.
				\item $\sem {\pi_U}$ is a $n$-uple of linear maps, such that its $i$-th element $\sem{\pi_U}^{(i)}$ is a map between $\cV_{\unlab({\widehat{\Phi_i}})}$ and $\mathrm{Lin}(\cV_{\Gamma_i}, \cV_{C})$. 
			\end{itemize}
			
			They are defined inductively below:
		\end{definition}
		
		\begin{definition}[$\sem \pi$, $\lStile$ fragment]\label{def:LinSem}
			When $\pi$ is a proof of $\Phi \lStile H$, $\sem \pi$ is a linear map between $\cV_\Phi$ and $\cV_H$, defined inductively as shown in Figure~\ref{fig:LinSem}. In the figure, when a rule contains a single hypothesis, the corresponding proof is denoted with $\pi'$. When it contains two hypotheses, their proofs are written respectively with $\pi_1$ and $\pi_2$.
		In the rule (\rUnitRight), recall that $\pi'$ is a $\uStile$ proof, and as such is a tuple with one singular item.
				The interpretation of $\lStile$ proofs follows directly from the interpretation of higher terms given in \cref{def:interp}.
		\end{definition}

\begin{figure}[t]
  \centering
\scalebox{.7}{\begin{minipage}{1.3\textwidth}
			{
        \[\sem{\text{\rAx}} = \left(z \mapsto z \Id A \right)
          \qquad
        \sem{\text{\rOneR}} = \left(z \mapsto z \Id{\bC} \right)
        \qquad
        \sem{\text{\rBtL}} = \sem{\pi'}
        \qquad
        \sem{\text{\rBtR}}^{(ij)}_{\varphi \otimes \psi} = \sem{\pi_1}^{(i)}_\varphi \otimes \sem{\pi_2}^{(j)}_{\psi}
      \]

      \[\sem{\text{\rExchUnG}}^{(i)}_\varphi = \begin{cases*}
				 \sem{\pi'}^{(1)}_\varphi \circ (\Id \Gamma \otimes \sigma_{A, B} \otimes \Id \Delta) & if $i = 1$\\
				\sem{\pi'}^{(i)}_\varphi \quad\text{  otherwise} &
			\end{cases*} \]
			\[\sem{\text{\rOneL}}^{(i)}_\varphi = \begin{cases*} \sem{\pi'}^{(1)}_\varphi \circ \lambda_\Gamma & if $i=1$ \\
			\sem{\pi'}^{(i)}_\varphi & otherwise \end{cases*}
    \qquad\qquad
      \sem{\text{\rPhTh}}^{(i)} =
        \begin{cases}
          e^{i\theta} \sem{\pi'}^{(1)} & \text{if $i = 1$}\\
          \sem{\pi'}^{(i)} & \text{otherwise}
        \end{cases}
      \]
			\[\sem{\begin{prooftree}
					\hypo{\pi'}
					\infer[no rule]1{P:\Omega \mid \omega' \mid \omega \mid \Omega' \uStile C}
					\infer1[\rExchUnW]{P:\Omega \mid \underbrace{\omega}_{j} \mid \underbrace{\omega'}_{j+1} \mid \Omega' \uStile C}
				\end{prooftree}}^{(i)} = \begin{cases*}
				\sem {\pi'}^{(j+1)} & if $i = j$\\
				\sem {\pi'}^{(j)} & if $i = j+1$\\
				\sem {\pi'}^{(i)} & otherwise
			\end{cases*}
    \]
    \[
      \sem{\begin{prooftree}
          \hypo{\pi_1}
          \infer[no rule]1{P:(\widehat{\Phi_i}; \Gamma_i)_{i \in I} \uStile A}
          
          \hypo{\pi_2}
            \infer[no rule]1{P:(\widehat \Psi; B, \Delta) \mid \Omega \uStile C}
            
            \infer2[\rUnitUn]{P:(\widehat{\Phi_i}, A \stackrel{\alpha}{\uarrow} B, \widehat \Psi; \Gamma_i, \Delta)_{i \in I} \mid \Omega \uStile C}
					\end{prooftree}}^{(j)}_{\varphi_j \otimes f \otimes \psi}
				=\begin{cases}
					\sem{\pi_2}^{(1)} \circ \left((f \circ \sem{\pi_1}^{(j)}_{\varphi_j}) \otimes \Id \Delta \right) & \text{if $j \in I$}\\
					\sem{\pi_2}^{(j-|I|+1)}_{\varphi_j \otimes f \otimes \psi}& \text{otherwise}
				\end{cases}
			\]
			\[\sem{\begin{prooftree}
					\hypo{\pi'}
					\infer[no rule]1{P:(\widehat \Phi; A, \Gamma) \mid (\widehat \Phi; B, \Gamma) \mid \Omega \uStile C}
					\infer1[\rOpL]{P:(\widehat \Phi; A \oplus B, \Gamma) \mid \Omega \uStile C}
			\end{prooftree}}^{(i)}_\varphi =
			\begin{cases}
				\begin{array}{@{}l}
					\sem{\pi'}^{(1)}_\varphi \circ (\proj{A,B}1 \otimes \Id \Gamma)\\
					+ \sem{\pi'}^{(2)}_\varphi \circ (\proj{A,B}2 \otimes \Id\Gamma)
				\end{array}
				& \text{if $i = 1$}\\
				\sem{\pi'}^{i+1}_{\varphi} & \text{otherwise}
			\end{cases} \]
			\[\sem{\begin{prooftree}
					\hypo{\pi_1}
					\infer[no rule]1{P:\Omega_1 \uStile A}
					\hypo{\pi_2}
					\infer[no rule]1{P:\Omega_2 \uStile B}
					\infer2[\rOpR]{P:\Omega_1 \mid \Omega_2 \uStile A \oplus B}
			\end{prooftree}}^{(k)}_\varphi =
			\begin{cases}
				 \inj{A, B}{1} \circ \sem{\pi_1}^{(k)}_\varphi & \text{if $1 \leq k \leq n$}\\
				\inj{A, B}{2} \circ \sem{\pi_2}^{(k-n)}_\varphi & \text{if $n < k \leq n+m$}
			\end{cases}\]
			\[\sem{
				\begin{prooftree}
					\hypo{\pi'}
					\infer[no rule]1{P:(\widehat \Phi; (\unit \oplus \unit), \Gamma) \mid \Omega \uStile C}
					\infer1[\rRotTh]{P:(\widehat \Phi; \Gamma)\mid (\widehat \Phi; \Gamma) \mid \Omega \uStile C}
			\end{prooftree}}^{(i)}_\varphi =
			\begin{cases}
				 \sem{\pi'}^{(1)}_\varphi\left(\left(\begin{smallmatrix}
					\cos \theta \\ \sin \theta
				\end{smallmatrix}\right) \otimes \Id \Gamma\right)&\text{if $i = 1$}\\
				\sem{\pi'}^{(1)}_\varphi \left(\left(\begin{smallmatrix}
					-\sin \theta \\ \cos \theta
				\end{smallmatrix}\right) \otimes \Id \Gamma\right)&\text{if $i = 2$}\\
				\sem{\pi'}^{(i-1)}_\varphi&\text{otherwise}
			\end{cases}\]
			\[
        \sem{\begin{prooftree}
            \hypo{\pi_1}
            \infer[no rule]1{P:(\widehat {\Phi_i}; \Gamma_i)_{i \in I} \uStile A}
            
            \hypo{\pi_2}
            \infer[no rule]1{P:(\widehat \Psi; A, \Delta) \mid \Omega \uStile C}
            \infer2[\rCutUnG]{P:(\widehat {\Phi_i}, \widehat \Psi; \Gamma_i, \Delta)_{i \in I} \mid \Omega \uStile C}
					\end{prooftree}}^{(j)}_{\varphi_j \otimes \psi}
				=
				\begin{cases}
					\sem{\pi_2}^{(1)}_\psi \circ( \sem{\pi_1}^{(j)}_{\varphi_j} \otimes \Id \Delta) & \text{if $j \leq |I|$}\\
					\sem{\pi_2}^{(j-|I|+1)}_{\varphi_j \otimes \psi}&\text{if $j > |I|$}
				\end{cases}
			\]}
  \end{minipage}}
  \caption{$\sem \pi$, $\uStile$ fragment}\label{fig:semuStile}
\end{figure}

		We will now define $\sem -$ over $\uStile$ proofs. 
		\begin{definition}[$\sem \pi$, $\uStile$ fragment]
			If $\pi$ is a proof of $(\widehat{\Phi_i}; \Gamma_i)_{i \in I} \uStile C$, $\sem \pi$ is a collection of linear maps of type $\cV_{\Phi_i} \rightarrow (\cV_{\Gamma_i} \rightarrow \cV_C)$.
      The rules are given in Figure~\ref{fig:semuStile}.
			
			Higher order inputs of $\sem \pi^{(i)}$ will be rendered as a subscript: for example, $\sem \pi^{(1)}_{f}$ is a map from a space of ground terms to the space of the conclusion.

			In $\sem{\text{ax}}$, wecall that $\sem{\pi}^{(1)}: \cV_\emptyset \rightarrow (\cV_A \rightarrow \cV_A)$, and that $\cV_\emptyset = \mathbb C$. Therefore, the $z$ variable is needed to keep $\sem \pi^{(1)}$ linear.
      
      Some rules follow a particular convention.
			In $\sem{\text{\rUnitUn}}$, we consider $\varphi_j \in \sem{\unlab(\widehat{\Phi_j})}$, $f \in \sem{A \uarrow B}$ and $\psi \in \sem{\unlab(\widehat \Psi)}$.
      In $\sem{\text{\rBtR}}$ and $\sem{\text{\rOpR}}$, we assume  $\Omega_1 = (\widehat{\Phi_i}, \Gamma_i)_{i \leq 1 \leq n}$ , and $\Omega_2 = (\widehat {\Psi_j}, \Delta_j)_{1 \leq j \leq m}$.
		\end{definition}

\begin{figure}[ht]
			\scalebox{.66}{\begin{minipage}{1.5\textwidth}
			\[\sem{\begin{prooftree}
				\hypo{\pi'}
				\infer[no rule]1{P: (\widehat \Phi, G^\beta, F^\alpha, \widehat \Psi; \Gamma) \mid \Omega \uStile C}
				\infer1[\rExchUnH]{P: (\widehat \Phi, F^\alpha, G^\beta, \widehat \Psi; \Gamma) \mid \Omega \uStile C}
			\end{prooftree}}^{(i)} \!\!\!\!\!=
			\begin{cases}
				\sem{\pi'}^{(1)}\circ(\Id \Phi \otimes \sigma_{F, G} \otimes \Id \Psi) & \text{if $i = 1$}\\
				\sem{\pi'}^{(i)}&\text{otherwise}
			\end{cases}\]
		\[\sem{\begin{prooftree}
			\hypo{\pi'}
			\infer[no rule]1{P:(F^{\otimes_1 \alpha}, G^{\otimes_2 \alpha}, \widehat \Phi; \Gamma) \mid \Omega \uStile C}
			\infer1[\rOtUn]{P:((F\otimes G)^\alpha, \widehat \Phi; \Gamma)}
		\end{prooftree}} = \sem{\pi'}\]
		\[\sem{
		\begin{prooftree}
			\hypo{\pi_L}
			\infer[no rule]1{\unlab(\widehat \Phi) \lStile F}
			\hypo{\pi'}
			\infer[no rule]1{P \sqcup (\widehat \Phi, \pi_L, F \stackrel{\alpha}{\multimap} G, G^\beta): (G^\beta, \widehat \Psi; \Gamma) \mid \Omega \uStile C}
			\infer2[\rArUnNew]{P:(\widehat \Phi, F \stackrel{\alpha}{\multimap} G, \widehat \Psi; \Gamma) \mid \Omega \uStile C}
		\end{prooftree}}^{(i)}
		= \begin{cases}
			\sem{\pi'}^{(1)} \circ (\app_{F,G} \circ (\sem {\pi_L} \otimes \Id{F \multimap G}) \otimes \Id \Psi) & \text{if $i = 1$}\\
			\sem{\pi'}^{(i)} & \text{otherwise}
		\end{cases}
		\]
		\[\sem{\begin{prooftree}
			\hypo{\pi'}
			\infer[no rule]1{P \sqcup (\widehat \Phi, \pi_L, F \stackrel{\alpha}{\multimap} G, G^\beta): (G^\beta, \widehat \Psi; \Gamma) \mid \Omega \uStile C}
			\infer1[\rArUnPro]{P\sqcup (\widehat \Phi, \pi_L, F \stackrel{\alpha}{\multimap} G, G^\beta):(\widehat \Phi, F \stackrel{\alpha}{\multimap} G, \widehat \Psi; \Gamma) \mid \Omega \uStile C}
		\end{prooftree}}^{(i)}
		= \begin{cases}
			\sem{\pi'}^{(1)} \circ (\app_{F,G} \circ (\sem {\pi_L} \otimes \Id{F \multimap G}) \otimes \Id \Psi) & \text{if $i = 1$}\\
			\sem{\pi'}^{(i)} & \text{otherwise}
		\end{cases}
	 \]
		\[\sem{\begin{prooftree}
				\hypo{\pi'}
				\infer[no rule]1{P:(F_x^{\oplus_x \alpha}, \widehat \Phi; \Gamma) \mid \Omega \uStile C}
				\infer1[\rWiUn{x}]{P:((F_1\with G_2)^\alpha, \widehat \Phi; \Gamma) \mid \Omega \uStile C}
		\end{prooftree}}^{(i)} = \left \lbrace \begin{aligned}
			\sem{\pi'}^{(1)} \circ(\proj{F_1, F_2}x \circ \Id \Phi) && \text{if $i = 1$}\\
			\sem{\pi'}^{(i)} && \text{otherwise}
		\end{aligned}\right. \]
		\[\sem{\begin{prooftree}
				\hypo{\pi_1}
				\infer[no rule]1{\Phi \lStile F}
				
				\hypo{\pi_2}
				\infer[no rule]1{P:(F^\alpha, \widehat{\Psi_i}; \Gamma_i)_{i \in I} \mid \Omega \uStile C}
				\infer2[\rCutUnH]{P:(\widehat \Phi, \widehat{\Psi_i}; \Gamma_i)_{i \in I} \mid \Omega \uStile C}
		\end{prooftree}}^{(i)} = \left \lbrace\begin{aligned}
			\sem{\pi_2}^{(i)} \circ (\sem {\pi_1} \otimes \Id {\Psi_i}) && \text{if $i \in I$}\\
			\sem{\pi_2}^{(i)} && \text{otherwise}
		\end{aligned} \right. \]
		\end{minipage}}
\caption{$\sem \pi$, $\uStile$ fragment with higher-order inputs (See Def.~\ref{def:uStileSem2})}
\label{fig:uStileSem2}
\end{figure}

\begin{definition}[$\sem \pi$, $\uStile$ fragment 2]\label{def:uStileSem2}
  The semantics of the rules dealing with higher order input is shown in \cref{fig:uStileSem2}. In particular, $\sem \pi^{(i)} \circ s$ is composing the higher order input, as compared to $\sem \pi^{(i)}_\varphi \circ f$, which would be composing with $f$ acting on quantum data.
  
  We also define $p_L = \left(\widehat \Phi, \pi_L, F \stackrel{\alpha}{\multimap} G, G^\beta\right)$.
  
  Note how $\otimes$ is essentially syntaxic sugar when on the left.
\end{definition}
			
\subsection{Properties}
	For \logicsymb\ to be useful, its semantics must be coherent with the interpretation of types. For the $\lStile$ fragment, this is simply proving that if $\pi$ is a proof of $\Phi \lStile H$, then $\sem \pi \in \sem{\Phi \multimap H}$. Finding the right condition for $\uStile$ proofs requires introducing a notion of pre-unitarity, that ensures all individual $\sem{\pi}^{(i)}$ output isometries can be summed into an unitary when given the right higher order input. In order to define it, we need to formalize how labels restrict the choice of higher order input of the interpretation of a $\uStile$ proof. This is simply formalizing the intuition that labels witness semantic equality.
	
	\begin{definition}[Consistent Valuations]
		Let $P: (\widehat{\Phi_i}; \Gamma_i)_{i \in I} \uStile C$ be a consistently labeled sequent.
		Let $v$ be a map from labeled terms of the sequent to values.
		We also extend it to multisets inductively: $v(\emptyset) = 1$ and $v(F^\alpha, \widehat \Phi) = v(F^\alpha) \otimes v(\widehat \Phi)$ when $\widehat \Phi$ is nonempty.
		We then say $v$ is a \emph{consistent valuation} for this sequent if:
		\begin{enumerate}
			\item For all $F^\alpha$, $v(F^\alpha) \in \sem F$;
			\item If $F^\alpha \in \widehat{\Phi_i}$ and $(\downarrow_p F)^{p\cdot \alpha}$ are present in the sequent, then $v((\downarrow_p F)^{p\cdot \alpha}) = \downarrow_p v(F^\alpha)$;
			\item If $(\widehat \Phi, \pi_L, F \stackrel{\alpha}{\multimap} G, G^\beta) \in P$, then $v(G^\beta) = v(F \stackrel{\alpha}{\multimap} G) (\sem{\pi_L} v(\widehat \Phi))$.
		\end{enumerate}
	\end{definition}
	
	\begin{definition}[Pre-unitarity]
		Let $\pi$ be a proof of $P: (\widehat{\Phi_i}; \Gamma_i)_{i \in I} \uStile C$, a consistently labeled sequent. 
		We say that $\sem \pi$ is a \emph{pre unitary} if, for all consistent valuations $v$ of $\pi$, for all $i,j \in I; i \neq j$ we have:
    \[
			\sem \pi^{(i)\dag}_{v(\widehat{\Phi_i})} \circ \sem{\pi}^{(i)}_{v(\widehat{\Phi_i})} = \Id{\Gamma_i},
      \quad
			\sem \pi^{(i)\dag}_{v(\widehat{\Phi_i})} \circ \sem \pi^{(j)}_{v(\widehat{\Phi_j})} = \zeromap,
      \quad
			\sum_{i \in I} \sem \pi^{(i)}_{v(\widehat{\Phi_i})} \circ \sem \pi^{(i)\dag}_{v(\widehat{\Phi_i})} = \Id C.
	 \]
	\end{definition}

	Using this definition, we can state this coherence theorem:
\begin{restatable}[Coherence of \logicsymb]{theorem}{consistency}
	Let $\pi$ be an \logicsymb\ proof.\\
	If $\pi$ is a proof of $\Phi \lStile H$, then, for all $\varphi \in \sem \Phi$, $\sem \pi (\varphi) \in \sem H$.\\
	If $\pi$ is a proof of $P: (\widehat{\Phi_i}; \Gamma_i)_{i \in I} \uStile C$ then, $\sem{\pi}$ is a pre-unitary.
\end{restatable}
	
	\begin{proof}
		The proof, while rather straightforward, requires a careful analysis of all the rules of \logicsymb. It can be found in full in \cref{proofcoh}.
	\end{proof}
	
	From this, we can derive two trivial corollaries that link pre-unitarity and actual unitaries when the original sequent of $\pi$ is particularly simple, when all worlds share the same higher order terms:
	
	\begin{corollary}
		If $\pi$ is a proof of $P:(\widehat \Phi; \Gamma_i)_{i \in I} \uStile C$, and if $\varphi \in \sem{\unlab{(\widehat \Phi)}}$ then
		\[\sum_{i \in I} \sem{\pi}_{\varphi}^{(i)}\circ \proj{\oplus_{j \in I} \Gamma_j}i : \bigoplus_{i \in I} \cV_{\Gamma_i} \rightarrow \cV_C \]
		is a unitary.
		In particular:
		\begin{itemize}
			\item 
		If $\pi$ is a proof of $(\emptyset; \Gamma_i)_{i \in I} \uStile C$, then
		$\sum_{i \in I} \sem \pi^{(i)}_1 \circ \proj{}{i} : \bigoplus_{i} \cV_{\Gamma_i} \rightarrow \cV_C$ is a unitary.
		
		\item
		If $\pi$ is a proof of $(\widehat \Phi; \Gamma) \uStile C$, and if $\varphi \in \sem \Phi$, then $\sem \pi^{(1)}_{\varphi} : \cV_\Gamma \rightarrow \cV_C$ is a unitary.\qed
	\end{itemize}
	\end{corollary}
	
	With this interpretation of \logicsymb, we can write proofs that are interpreted as known higher order processes. The proofs can be found in \cref{reps}.
	
	\begin{theorem}
		For any of the $\SWITCH$, the Quantum $n$-Switch, the unitarized Grenoble Process and the mapping of a unitary (with a defined behavior on the vacuum state) to its controlled version, there is a proof $\pi$ such that $\sem \pi$ is equal to that process.\qed
	\end{theorem}
	
	Finally, while the interpretation of terms of the form $F \multimap G$ might be too coarse for a total and formal completeness result, the rules \rRotTh\ and \rPhTh\ are enough to prove that \logicsymb\ is complete for some output terms.
	\begin{restatable}[Universality for Unitaries]{theorem}{universality}
		Let $n \geq 1$ and $f : \bC^n \rightarrow \bC^n$ a unitary and $\underline n := \underbrace{\unit \oplus ... \oplus \unit}_{\text{$n$ times}}$.
		There exists $\pi$ a proof of $(\underline n \uarrow \underline n)$ such that $\sem \pi = f$.
	\end{restatable}
	
	\begin{proof}
		The proof relies on a result found e.g.~in \cite{brennen2005}, which proves that (complex versions of) Givens rotations are universal for unitaries over $\mathbb C^n$. 
		In \cref{proofComp}, given the indices and angles of a Givens rotation, we construct a proof with interpretation equal to that rotation. We then use \rCutUnG\ rules to compose multiple such proofs, and generate $f$.
	\end{proof}

\section{Cut Elimination}
	With semantics now defined, we can prove cut elimination for \logicsymb, and its soundness. Intuitively, since cut rules can be understood as composing processes together, eliminating them can be understood as computing the output of a process within another.
	
	We follow a typical pattern for proving cut elimination, by introducing a notion of cut rank that encodes the complexity of cuts found in the proof, and showing that it can always be decreased. Eliminating \rCutLin, the cut rule for the $\lStile$ fragment, is identical to eliminating cuts in \IMALL, provided cut elimination holds in the $\uStile$ fragment. However, \rCutUnH\ and \rCutUnG\ are made more complex with the introduction of worlds.
	
	First, \rCutUnH\ needs to be applied to multiple worlds at once, so labels are equal without having to resort to promises, which would complexify even stating a cut elimination theorem immensely. This makes commuting \rCutUnH\ up a proof sometimes challenging. Consider the proof fragment in \Cref{fig:cutHex}. While a cut-free fragment trivially exists, the \rBtL\ rule further above \rCutUnH\ needs to be commuted down. This means that eliminating \rCutUnH\ requires to look further than the root rule of its right subproof. To prove the general theorem, special attention needs to be drawn to commutations of rules acting on different worlds. This is solved by observing that two rules applying on distinct sets of worlds do indeed commute.

	\begin{figure}[h]
	\begin{center}
		\fbox{
		\begin{prooftree}
		\hypo{\unlab(\widehat \Phi) \lStile F}
		\hypo{\unlab(\widehat \Phi) \lStile G}
		\infer2[\rWiLinR]{\unlab(\widehat \Phi) \lStile F \with G}
		
		\hypo{P:F^{\oplus_1 \alpha}; A \mid F^{\oplus_1\alpha}; A, B \uStile C}
		\infer1[\rWiUn 1]{P:F^{\oplus_1 \alpha}; A \mid (F\with G)^\alpha; A, B \uStile C}
		\infer1[\rBtL]{P: F^{\oplus_1 \alpha}; A \mid (F\with G)^\alpha; A \boxtimes B \uStile C}
		\infer1[\rWiUn 1]{P:(F\with G)^\alpha; A \mid (F\with G)^\alpha; A \boxtimes B \uStile C}
		\infer2[\rCutUnH]{P: \widehat \Phi; A \mid \widehat \Phi; A \boxtimes B \uStile C}			
		\end{prooftree}}
	\end{center}
	\caption{An example of a challenging \rCutUnH\ elimination}
	\label{fig:cutHex}
	\end{figure}
	
	Second, the fusing of worlds of the \rUnitUn\ and \rRotTh\ rules makes commutations of instances of \rCutUnG\ difficult, like in \Cref{fig:cutGex}. However, the behavior of the $\boxtimes$ and $\oplus$ connectives is derived from \IMALL, where those are positive connectives, that can always be destructed at any point of the proof without changing the provability of the sequent. In \cref{sec:pos}, we prove that this is also the case in \logicsymb, and therefore all terms created by \rCutUnG\ can be destructed immediately, and the cut rule simplified.
	
	\begin{figure}[h]
		\centering
		\fbox{
		\begin{prooftree}
			\hypo{P: \Gamma \uStile A}
			\hypo{P: (\unit \oplus \unit), A \uStile C}
			\infer1[\rRotTh]{P:A \mid A \uStile C}
			\infer2[\rCutUnG]{P:\Gamma \mid A \uStile C}
		\end{prooftree}}
		\caption{An example of a challenging \rCutUnG\ commutation}
		\label{fig:cutGex}
	\end{figure}
	
	Through these two solutions, we can prove cut elimination (the full proof is found in \cref{sec:celim}).
	
	\begin{restatable}[Cut Elimination]{theorem}{cutelim}
		If $\pi$ is a proof of $\Phi \lStile H$, there exists a cut-free proof $\rho$ of the same sequent such that $\sem \pi = \sem {\rho}$. 
		If $\pi$ is a proof of $P: \Omega \uStile C$, there exists a cut-free promise set $Q$, and a cut-free proof $\rho$ of $Q: \Omega \uStile C$ such that $\sem \pi = \sem {\rho}$. \qed
	\end{restatable}

\section{Discussion}
	We have described a novel logic, \logicsymb, and equipped it with an interpretation for its terms and proofs suitable for quantum computation. Doing this, we have taken the first step in the creation of a type system for quantum processes. We have shown that \logicsymb\ is indeed coherent for this interpretation, and complete for the subset of terms representing $n$ dimensional unitaries. We also have shown that \logicsymb\ features cut-elimination, and that it is sound with regards to the provided interpretation.
	
	By building on this cut elimination result, it could be possible to build a lambda calculus, or another language, for \logicsymb, and provide a more robust language to describe higher order quantum processes. The many-worlded nature of the $\uStile$ fragment echoes behavior found in \cite{MW, Tape}, and makes these languages promising candidates for extension through \logicsymb.
	
	However, the interpretation given for $F \multimap G$ is coarse, and makes {\logicsymb} incomplete. In particular, there is no {\logicsymb} proof that can be intepreted as the Lugano process (also known as AF/BW process)~\cite{Vanrietvelde2025consistentcircuits}, a well known ICO. The Lugano process is a voting procedure in which three voters vote among themselves for who is to come last in their causal order, but unintuitively know the outcome of the vote before casting their own vote. Despite its seemingly paradoxical nature, the Lugano process follows the theory of quantum superchannels. We believe this is not a weakness of \logicsymb, but rather points to a fundamental difference between the Lugano process and other ICOs like $\SWITCH$ or the Grenoble process at a computational level. Further study and refinement of {\logicsymb} and its semantics could provide a syntax-based classification of ICOs, in parallel to the usual information theoretic framework.
	
	Finally, we have not proven that the semantics {\logicsymb} fully conform to the conditions of superchannels laid out in~\cite{grenoble}.  The conservation-based intepretation guarantee some compliance, but superchannels must also \emph{completely} conserve channels. This asks for processes to be extensible to accept input that interacts with a larger environment. For instance, any process $s$ in $\sem{(A \uarrow B) \multimap (C \uarrow D)}$ should be extensible to a process of $\sem{(X \boxtimes A \uarrow Y \boxtimes B) \multimap (X \boxtimes C \uarrow Y \boxtimes D)}$ for any ground types $X$ and $Y$ simply by tensoring $s$ with the identity of $\text{Lin}(\cV_X, \cV_Y)$. This definition is easy to lift for simple terms of {\logicsymb}, for example processes that map a number of unitaries to an unitary. In this case, this property, known as superunitarity~\cite{Vanrietvelde2025consistentcircuits}, seems easy to verify for intepretations of proofs of {\logicsymb} that represent transformations of multiple unitaries into a single unitary. However, these concepts are difficult to lift in the general case, when considering order $3$ or above processes, or the use of $\otimes$ and $\with$. This notion of extending interfaces is already present in~\cite{caus}, which also uses Linear Logic as a base, through the non-intuitionistic $\parr$ connective. Finding a way to leverage similar constructions in \logicsymb's intuitionistic setting could provide a guarantee that $\sem -$ corresponds to the superchannels of~\cite{grenoble}.

\bibliography{refs}

\newpage
\appendix

\begin{figure}
			{\small
				\[\sem{\begin{prooftree}
						\hypo{\pi'}
						\infer[no rule]1{\Phi, G, F, \Psi \lStile H}
						\infer1[\rExchLin]{\Phi, F, G, \Psi \lStile H}
				\end{prooftree}} = \sem {\pi'} \circ (\Id{\Phi} \otimes \sigma_{F, G} \otimes \Id{\Psi})\]
			\[\sem{\begin{prooftree}
				\hypo{\pi'}
				\infer[no rule]1{\emptyset:(\widehat \Phi; A) \uStile B}
				\infer1[\rUnitRight]{\Phi \lStile A \uarrow B}
			\end{prooftree}} = \sem{\pi'}^{(1)}\]
			\[\sem{\begin{prooftree}
				\hypo{\pi'}
				\infer[no rule]1{F, G, \Phi \lStile H}
				\infer1[\rOtLinL]{F \otimes G, \Phi \lStile H}
			\end{prooftree}}= \sem{\pi'}\]
			\[\sem{\begin{prooftree}
				\hypo{\pi_1}
				\infer[no rule]1{\Phi \lStile F}
				
				\hypo{\pi_2}
				\infer[no rule]1{\Psi \lStile G}
				\infer2[\rOtLinR]{\Phi, \Psi \lStile F \otimes G}
			\end{prooftree}} = \sem{\pi_1} \otimes \sem{\pi_2} \]
			\[\sem{\begin{prooftree}
				\hypo{\pi_1}
				\infer[no rule]1{\Phi \lStile F}
				
				\hypo{\pi_2}
				\infer[no rule]1{G, \Psi \lStile H}
				\infer2[\rArLinL]{\Phi, F \multimap G, \Psi \lStile H}
			\end{prooftree}} = \sem{\pi_2} \circ ((\app_{F, G} \circ (\sem {\pi_1} \otimes \Id{F \multimap G}))\otimes \Id{\Psi}) \]
			\[\sem{\begin{prooftree}
					\hypo{\pi'}
					\infer[no rule]1{F, \Phi \lStile G}
					\infer1[\rArLinR]{\Phi \lStile F \multimap G}
			\end{prooftree}}= \varphi \mapsto (f \mapsto \sem{\pi'}(f \otimes \varphi)) \]
			\[\sem{\begin{prooftree}
				\hypo{\pi'}
				\infer[no rule]1{F_x, \Phi \lStile H}
				\infer1[\rWiLinL x]{F_1 \with F_2, \Phi \lStile H}
			\end{prooftree}} = \sem{\pi'}\circ(\proj{F_1, F_2}{x} \otimes \Id{\Phi}) \]
			\[\sem{\begin{prooftree}
				\hypo{\pi_1}
				\infer[no rule]1{\Phi \lStile F}
				
				\hypo{\pi_2}
				\infer[no rule]1{\Phi \lStile G}
				\infer2[\rWiLinR]{\Phi \lStile F \with G}
			\end{prooftree}} = \varphi \mapsto (\sem{\pi_1}(\varphi), \sem{\pi_2}(\varphi)) \]
			\[\sem{\begin{prooftree}
				\hypo{\pi_1}
				\infer[no rule]1{\Phi \lStile F}
				
				\hypo{\pi_2}
				\infer[no rule]1{F, \Psi \lStile H}
				\infer2[\rCutLin]{\Phi, \Psi \lStile H}
			\end{prooftree}} = \sem{\pi_2}\circ(\sem {\pi_1} \otimes \Id{\Psi})\]
		}
    \caption{Complete semantics for the $\lStile$ fragment}
    \label{fig:CompleteLinSem}
  \end{figure}

\section{Representations}
	\label{reps}
	\subsection{The Controlled Unitary}
		Consider a process that takes some unitary $U$, and outputs a new unitary $C(U)$, acting on one more wire, such that $C(U)\ket 0 \ket x = \ket 0 \ket x$ and $C(U)\ket1 \ket x = \ket 1 U\ket x$. This is an important operation to perform, though it is not strictly linear in $U$. This is resolved by extending our view of $U$ to also include its behavior on the one dimensional space, as explained in~\cite{Chiribella_2019}. In particular, to control $U$ properly, it needs to always output a vacuum state when its input is the vacuum state.
		
		This naturally leads to the term $((A \uarrow A) \with (\unit \uarrow \unit)) \multimap (((\unit \oplus \unit) \boxtimes A) \uarrow ((\unit \oplus \unit) \boxtimes A))$ to represent the entire process of controlling $U$, and its proof in \cref{fig:Cproof}. In that proof, $\unit$ terms have been given subscript labels for readability.
		
		\begin{figure}[h]
			\centering\small
			\fbox{
				\begin{prooftree}
					\infer0[\rAx]{A \uStile A}
					\infer0[\rOneR]{\uStile \unit_v}
					
					\infer0[\rAx]{\unit_{\ket0} \uStile \unit_{\ket 0}}
					\infer0[\rAx]{\unit_{\ket1} \uStile \unit_{\ket 1}}
					
					\infer2[\rOpR]{\unit_{\ket 0} \mid \unit_{\ket 1} \uStile \unit_{\ket 0} \oplus \unit_{\ket 1}}
					
					\infer0[\rAx]{A \uStile A}
					
					\infer2[\rBtR]{\unit_{\ket 0}, A \mid; \unit_{v}, \unit_{\ket 1}, A \uStile (\unit_{\ket 0} \oplus \unit_{\ket 1}) \boxtimes A}
					\infer1[\rOneL]{\unit_{\ket 0}, A \mid; \unit_{v}, \unit_{\ket 1}, A \uStile (\unit_{\ket 0} \oplus \unit_{\ket 1}) \boxtimes A}
					\infer2[\rUnitUn]{\unit_{\ket 0}, A \mid \unit_{v} \stackrel{\oplus_2 u}{\uarrow} \unit_{v}; \unit_{\ket 1}, A \uStile (\unit_{\ket 0} \oplus \unit_{\ket 1}) \boxtimes A}
					\infer2[\rUnitUn]{A \stackrel{\oplus_1 u}{\uarrow} A; \unit_{\ket 0}, A \mid \unit_{v} \stackrel{\oplus_2 u}{\uarrow} \unit_{v}; \unit_{\ket 1}, A \uStile (\unit_{\ket 0} \oplus \unit_{\ket 1}) \boxtimes A}
					\infer1[\rWiUn 2]{A \stackrel{\oplus_1 u}{\uarrow} A; \unit_{\ket 0}, A \mid ((A \uarrow A) \with (\unit_{v} \uarrow \unit_{v}))^u; \unit_{\ket 1}, A \uStile (\unit_{\ket 0} \oplus \unit_{\ket 1}) \boxtimes A}
					\infer1[\rWiUn 1]{((A \uarrow A) \with (\unit_{v} \uarrow \unit_{v}))^u; \unit_{\ket 0}, A \mid ((A \uarrow A) \with (\unit_{v} \uarrow \unit_{v}))^u; \unit_{\ket 1}, A \uStile (\unit_{\ket 0} \oplus \unit_{\ket 1}) \boxtimes A}
					\infer1[\rOpL]{((A \uarrow A) \with (\unit_{v} \uarrow \unit_{v}))^u; \unit_{\ket 0} \oplus \unit_{\ket 1}, A \uStile (\unit_{\ket 0} \oplus \unit_{\ket 1}) \boxtimes A}
					\infer1[\rBtL]{((A \uarrow A) \with (\unit_{v} \uarrow \unit_{v}))^u; (\unit_{\ket 0} \oplus \unit_{\ket 1}) \boxtimes A \uStile (\unit_{\ket 0} \oplus \unit_{\ket 1}) \boxtimes A}	
					\infer1[\rUnitRight]{(A \uarrow A) \with (\unit \uarrow \unit) \lStile  ((\unit \oplus \unit) \boxtimes A) \uarrow ((\unit \oplus \unit) \boxtimes A)}
					\infer1[\rArLinR]{\lStile ((A \uarrow A) \with (\unit \uarrow \unit)) \multimap (((\unit \oplus \unit) \boxtimes A) \uarrow ((\unit \oplus \unit) \boxtimes A))}
				\end{prooftree}
			}
			\caption{Proof of the controlled unitary process}
			\label{fig:Cproof}
		\end{figure}
		
		The proof relies on treating the two worlds differently depending on whether they're on the left or on the right. One projects the extended unitary to its vacuum side, and the other on its main side. They both apply the resulting unitary, and since the vacuum state is of dimension 1 and represented by $\unit$, it can be freely created or discarded. The two worlds are then joined back together.
	
	\subsection{The Quantum Switch}
		The proof of $\SWITCH$ in \cref{fig:switchProof} relies on the same structure: worlds are split according to the control qubit, and the two input unitaries are applied in one order or the other depending on the world.
		\begin{figure}[h]
      \def\myrel#1{\stackrel{#1}{\uarrow}}
    \scalebox{.8}{\begin{prooftree}
			\infer0[\rAx]{A \uStile A}
			\infer0[\rAx]{A \uStile A}
			\infer0[\rAx]{A \uStile A}
			\infer0[\rAx]{A \uStile A}
			
			\infer0[\rAx]{\unit \uStile \unit}
			\infer0[\rAx]{\unit \uStile \unit}
			\infer2[\rOpR]{\unit \mid \unit \uStile \underline 2}
			
			\infer0[\rAx]{A \uStile A}
			
			\infer2[\rBtR]{\unit, A \mid \unit, A \uStile \underline 2 \boxtimes A}
			\infer2[\rUnitUn]{\unit, A \mid A \myrel{f} A; \unit, A \uStile \underline 2 \boxtimes A}
			\infer2[\rUnitUn]{\unit, A \mid A \myrel{f} A, A \myrel{g} A; \unit, A \uStile \underline 2 \boxtimes A}
			\infer2[\rUnitUn]{A \myrel{g} A; \unit, A \mid A \myrel{f} A, A \myrel{g} A; \unit, A \uStile \underline 2 \boxtimes A}
			\infer2[\rUnitUn]{A \myrel{f} A, A \myrel{g} A; \unit, A \mid A \myrel{f}, A \myrel{g} A; \unit, A \uStile \underline 2 \boxtimes A}
			\infer1[\rOpL]{A \myrel{f} A, A \myrel{g} A; \underline 2, A \uStile \underline 2 \boxtimes A}
			\infer1[\rBtL]{A \myrel{f} A, A \myrel{g} A; \underline 2 \boxtimes A \uStile \underline 2 \boxtimes A}
			\infer1[\rUnitRight]{A \uarrow A, A \uarrow A \lStile (\underline 2 \boxtimes A) \uarrow (\underline 2 \boxtimes A)}
			\infer1[\rOtLinL]{(A \uarrow A) \otimes (A \uarrow A) \lStile (\underline 2 \boxtimes A) \uarrow (\underline 2 \boxtimes A)}
			\infer1[\rArLinR]{\lStile ((A \uarrow A) \otimes (A \uarrow A)) \multimap ((\underline 2 \boxtimes A) \uarrow (\underline 2 \boxtimes A))}
		\end{prooftree}}
		\caption{Proof of $\SWITCH$}
		\label{fig:switchProofappendix}
		\end{figure}
		
	\subsection{The Quantum n-Switch}
		The Quantum $n$-Switch is a variant of $\SWITCH$, where instead $n$ unitaries of type $A \uarrow A$ are given, and composed in any of the $n!$ possible ways, controlled by quantum data.
		
		The term representing the $n$-Switch would therefore be $(A \uarrow A)^{\otimes n} \multimap ((\underline{n!} \boxtimes A) \uarrow (\underline{n!} \boxtimes A))$. A proof tree can be constructed the same way as for $\SWITCH$, by separating all worlds, and then applying the unitaries in the relevant order for each world.
	
	\subsection{The Grenoble Process}
		Described in~\cite{grenoble}, the Grenoble Process is a more complex ICO, featuring dynamic quantum control. In its pure version, it takes three unitaries $C_0, C_1, C_2$ each acting on a single qubit, and creates an isometry that takes a control qutrit, a source qubit, and outputs a qutrit, one ancillary qubit, and a target qubit. If the control qutrit is in state $\ket k$, $C_k$ is applied first to the target qubit. Then, if the output of $C_k$ is $\ket 0$, the output ancillary qubit is set to $\ket 0$, and $C_{k+1}$, then $C_{k+2}$ (with indices working mod $3$) are applied to the source qubit to obtain the target qubit. Otherwise, the ancilary qubit is set to $\ket 1$ and then $C_{k+2}$, then $C_{k+1}$ are applied to the source qubit. This is formalized in \cref{fig:GrenobleSpec}.
		
		\begin{figure}[h]
			\begin{center}
				\[
				\begin{cases}
					\ket 0 \ket x \mapsto (\ket 2 \ket 0 C_2 C_1 \circ \proj {}{\ket 0} + \ket 1 \ket 1 C_1 C_2 \circ \proj {} {\ket 1}) \circ C_0 \ket x\\
					\ket 1 \ket x \mapsto (\ket 0 \ket 0 C_0 C_2 \circ \proj {}{\ket 0} + \ket 2 \ket 1 C_2 C_0 \circ \proj {} {\ket 1}) \circ C_1 \ket x\\
					\ket 2 \ket x \mapsto (\ket 1 \ket 0 C_1 C_0 \circ \proj {}{\ket 0} + \ket 0 \ket 1 C_0 C_1 \circ \proj {}{\ket 1}) \circ C_2 \ket x
				\end{cases}
				\]
			\end{center}
			\caption{Specification of the Grenoble Process}
			\label{fig:GrenobleSpec}
		\end{figure}
		
		This statement of the Grenoble Process is not exactly unitary-preserving, as it produces an isometry. However, it can easily be extended in an unitary form, by making the unitary qubit part of the input as well, and building the unitarized Grenoble Process to be equal to the original when this ancillary input is set to $\ket 0$. We implement this by reversing the preference between the two remaining maps depending on the state of the ancillary qubit, and then continuing the Grenoble process as normal.
		
		Due to the more complex structure of the Grenoble Process, the full derivation of the term representing it contains more than seventy rules and manipulates up to twelve worlds at a time. Therefore, instead of presenting the full derivation, we give a graphical representation (formally defined in \cref{def:graphical}) of the core of the derivation in \cref{fig:GrenobleGraph} instead. $\pi_{\text{fin}}$ uses only \rExchUnW, \rBtR, \rOpR, and \rAx\ rules, to match each world with the correct value for the output qutrit and qubit. From left to right, the assigned values for the qutrit-ancilla pair are $\ket2 \ket 0, \ket 1 \ket 1, \ket 0 \ket 0, \ket 2 \ket 1, \ket 1 \ket 0$ and $\ket 0 \ket 1$. This leads to the semantics given in~\cref{fig:UnitGrenobleSpec}, which coincide with the typical Grenoble process when the middle qubit is set to $\ket 0$.
		
		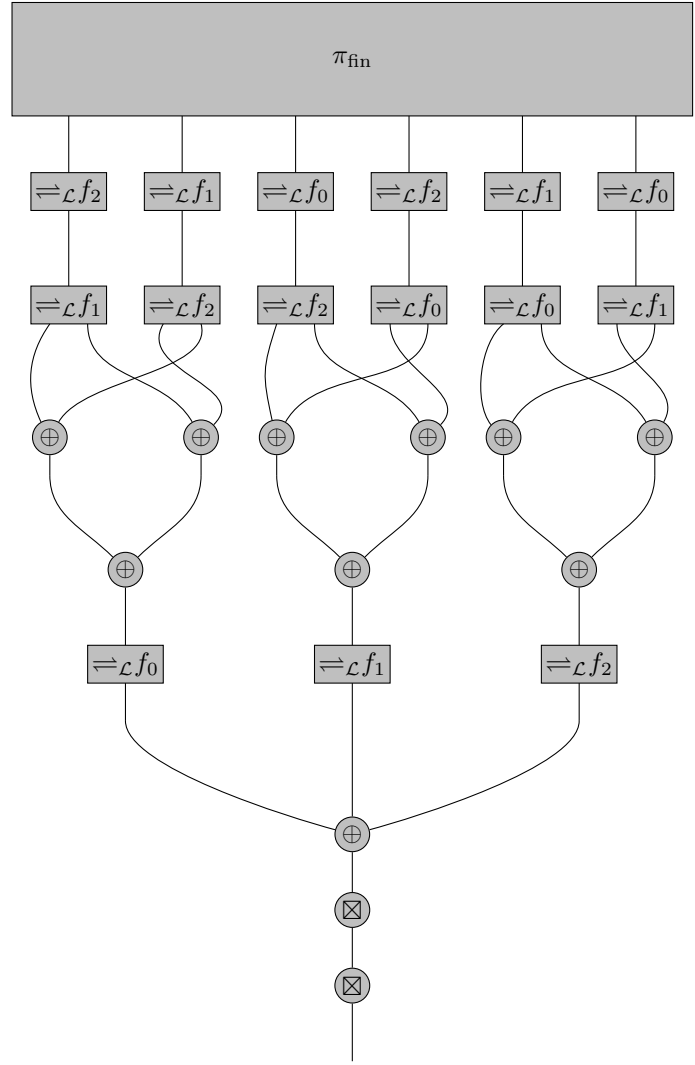
\begin{figure}[h]
			\begin{center}
		\input{grenoble-proof.tikz}
		\end{center}
		\caption{Graphical representation of the proof of the unitarized Grenoble process}
		\label{fig:GrenobleGraph}
		\end{figure}
		
		\begin{figure}[h]
			\begin{center}
				\[
				\begin{cases}
					\ket 0 \ket 0 \ket x \mapsto (\ket 2 \ket 0 C_2 C_1 \circ \proj {}{\ket 0} + \ket 1 \ket 1 C_1 C_2 \circ \proj {} {\ket 1}) \circ C_0 \ket x\\
					\ket 0 \ket 1 \ket x \mapsto (\ket 2 \ket 0 C_2 C_1 \circ \proj {}{\ket 1} + \ket 1 \ket 1 C_1 C_2 \circ \proj {} {\ket 0}) \circ C_0 \ket x\\
					\ket 1 \ket 0 \ket x \mapsto (\ket 0 \ket 0 C_0 C_2 \circ \proj {}{\ket 0} + \ket 2 \ket 1 C_2 C_0 \circ \proj {} {\ket 1}) \circ C_1 \ket x\\
					\ket 1 \ket 1 \ket x \mapsto (\ket 0 \ket 0 C_0 C_2 \circ \proj {}{\ket 1} + \ket 2 \ket 1 C_2 C_0 \circ \proj {} {\ket 0}) \circ C_1 \ket x\\
					\ket 2 \ket 0 \ket x \mapsto (\ket 1 \ket 0 C_1 C_0 \circ \proj {}{\ket 0} + \ket 0 \ket 1 C_0 C_1 \circ \proj {}{\ket 1}) \circ C_2 \ket x\\
					\ket 2 \ket 1 \ket x \mapsto (\ket 1 \ket 0 C_1 C_0 \circ \proj {}{\ket 1} + \ket 0 \ket 1 C_0 C_1 \circ \proj {}{\ket 0}) \circ C_2 \ket x
				\end{cases}
				\]
			\end{center}
			\caption{Semantics of the~\cref{fig:GrenobleGraph} proof}
			\label{fig:UnitGrenobleSpec}
		\end{figure}
		
\section{Useful Tools}
	\begin{definition}[Proof Equivalence]
		Let $\pi_A$ and $\pi_B$ be two proofs of the same sequent, up to the content of the promise set if $\pi_A$ and $\pi_B$ are $\uStile$ proofs.
		
		We say that $\pi_A$ and $\pi_B$ are \emph{equivalent}, noted $\pi_A \equiv \pi_B$ if and only if $\sem {\pi_A} = \sem {\pi_B}$.
	\end{definition}
	
	\subsection{Simple Proofs}
		Proofs that pass higher order input to the left subproof of rule \rUnitUn\ and right side rules of $\uStile$ can prove troublesome from a bureaucratic point of view.\\
		We prove that they can be ignored without loss of generality by showing that the original \rUnitUn, \rBtR, \rOpR\ and \rCutUnG\ rules are derivable in a system where stricter versions are introduced.
		
		\begin{definition}
			A proof $\pi$ is said \emph{simple} if every instance of the \rUnitUn, \rBtR, \rOpR\ and \rCutUnG\ rules is of the form:
			\[\begin{prooftree}
				\hypo{P:(\emptyset, \Gamma_i)_{i \in I} \uStile A}
				\hypo{P:(\widehat \Phi; B, \Delta) \mid \Omega \uStile C}
				\infer2[\rUnitUn]{P:(A \stackrel{\alpha}{\uarrow} B, \widehat \Phi; \Gamma_i, \Delta)_{i \in I} \mid \Omega \uStile C}
			\end{prooftree} \]
			
			\[
			\begin{prooftree}
				\hypo{P:(\emptyset; \Gamma_i)_{i \in I} \uStile C_1}
				\hypo{P:(\emptyset; \Delta_j)_{j \in J} \uStile C_2}
				\infer2[\rBtR]{P:(\emptyset; \Gamma_i, \Delta_j)_{ij} \uStile C_1 \boxtimes C_2}
			\end{prooftree}
			\]
			
			\[
			\begin{prooftree}
				\hypo{P:(\emptyset; \Gamma_i)_{i \in I} \uStile C_1}
				\hypo{P:(\emptyset; \Delta_j)_{j \in J} \uStile C_2}
				\infer2[\rOpR]{P:(\emptyset; \Gamma_i)_i \mid (\emptyset; \Delta_j)_{j} \uStile C_1 \oplus C_2}
			\end{prooftree}
			\]
			
			\[\begin{prooftree}
				\hypo{P:(\emptyset, \Gamma_i)_{i \in I} \uStile A}
				\hypo{P:(\widehat \Phi; B, \Delta) \mid \Omega \uStile C}
				\infer2[\rCutUnG]{P:(\widehat \Phi; \Gamma_i, \Delta)_{i \in I} \mid \Omega \uStile C}
			\end{prooftree} \]
		\end{definition}
		
		\begin{lemma}
			\label{lem:simple}
			For each proof $\pi$ of \logicsymb, there exists a simple proof $\pi_S$ that is equivalent to $\pi$.
		\end{lemma}
		
		\begin{proof}
			The proof is immediate when considering all non-right rules commute with \rUnitUn, \rBtR, \rOpR\ and \rCutUnG\ on the side we are interested in, up to duplication and relabeling.
			
			For example,
			
			\[\scalebox{0.8}{$\begin{prooftree}
				\hypo{P:\Omega_1 \mid \omega' \mid \omega \mid \Omega_2 \uStile A}
				\infer1[\rExchUnW]{P:\Omega_1 \mid \omega \mid \omega' \mid \Omega_2 \uStile A}
				
				\hypo{P:\Omega_3 \uStile B}
				\infer2[\rOpR]{P:\Omega_1 \mid \omega \mid \omega' \mid \Omega_2 \mid \Omega_3 \uStile A \oplus B}
			\end{prooftree}
			\rightarrow
			\begin{prooftree}
				\hypo{P:\Omega_1 \mid \omega' \mid \omega \mid \Omega_2 \uStile A}
				\hypo{P:\Omega_3 \uStile B}
				
				\infer2[\rOpR]{P:\Omega_1 \mid \omega' \mid \omega \mid \Omega_2 \mid \Omega_3 \uStile A \oplus B}
				
				\infer1[\rExchUnW]{P:\Omega_1 \mid \omega \mid \omega' \mid \Omega_2 \mid \Omega_3 \uStile A \oplus B}
			\end{prooftree}$}
			\]
			
			and,
			
			\scalebox{0.8}{
			$\begin{prooftree}
				\hypo{P:(\widehat \Phi; A, B, \Gamma) \mid \Omega_1 \uStile A}
				\infer1[\rBtL]{P:(\widehat \Phi; A \boxtimes B, \Gamma) \mid \Omega_1 \uStile A}
				
				\hypo{P:\Omega_2 \uStile B}
				\infer2[\rBtR]{P:(\widehat \Phi; A \boxtimes B, \Gamma)\times \Omega_2 \mid \Omega_1 \times \Omega_2 \uStile A \boxtimes B}
			\end{prooftree}
			\rightarrow
			\begin{prooftree}
				\hypo{P:(\widehat \Phi; A, B, \Gamma) \mid \Omega_1 \uStile A}
				\hypo{P:\Omega_2 \uStile B}
				\infer2[\rBtR]{P:(\widehat \Phi; A, B, \Gamma)\times \Omega_2 \mid \Omega_1 \times \Omega_2 \uStile A \boxtimes B}
				\infer[double]1[\rBtL]{P:(\widehat \Phi; A \boxtimes B, \Gamma)\times \Omega_2 \mid \Omega_1 \times \Omega_2 \uStile A \boxtimes B}
			\end{prooftree}$}
		
		\end{proof}

	\subsection{Graphical Representation}
		Some properties are more agreeable to state when using diagrams. Using simple proofs, we can define a rigorous diagrammatic representation of the $\uStile$ fragment of \logicsymb. Edges represent worlds, and vertices rules or subproofs, along with some decorations for promises.
		
		Note this is only a different presentation of proof trees, and not a meaningfully different form of proofs, like proof nets are for linear logic.
		
		\begin{definition}
			\label{def:graphical}
			Sections of a proof of a $\uStile$ sequent can be represented graphically with this correspondance:
			
			\scalebox{0.75}{
			\input{rep-graph.tikz}}
		\end{definition}
		
		This leads us to this basic informal property: two rules over ``independant'' sets of worlds commute painlessly.
		
		\begin{lemma}
			\label{lem:mono}
			Let $B_1$ and $B_2$ blocks of rules. We have this equivalence
			
			\ctikzfig{commut-lemma}
			
			Where, if $P_{12}$ is the set of promises created in $B_1$ that is employed in instances of \rArUnPro\ in $B_2$, $B_1^*$ is a block of rules equal to $B_1$, where all instances of \rArUnNew\ generating a promise of $P_{12}$ are replaced with \rArUnPro\ using the same promise, and $B_2^*$ a block equal to $B_2$ where the first instance of each \rArUnPro\ using a promise of $P_{12}$ is replaced by a \rArUnNew\ rule.
			
		\end{lemma}
		
		\begin{proof}
			This is immediate when both $B_1$ and $B_2$ each consist of a single rule.\\
			We then proceed by trivial induction over their combined size.
		\end{proof}
	
		\subsection{Positivity}
	\label{sec:pos}
	
	\begin{lemma}
		\label{lem:posBt}
		Let $\pi$ a proof of $P:(\widehat \Phi; A \boxtimes B, \Gamma) \mid \Omega \uStile C$.\\
		There exists $\pi'$ a proof of $P:(\widehat \Phi; A, B, \Gamma) \mid \Omega \uStile C$ such that \[\pi \equiv \begin{prooftree}
			\hypo{\pi'}
			\infer[no rule]1{P:(\widehat \Phi; A, B, \Gamma) \mid \Omega \uStile C}
			\infer1[\rBtL]{P:(\widehat \Phi; A \boxtimes B, \Gamma) \mid \Omega \uStile C}
		\end{prooftree}\]
	\end{lemma}
	
	\begin{proof}
		We proceed by induction over $\pi$.
		
		If $\pi = \begin{prooftree}
			\hypo{\pi_1}
			\infer[no rule]1{P:(\widehat{\Phi_i}; \Gamma_i)_i \uStile C}
			
			\hypo{\pi_2}
			\infer[no rule]1{P:(\widehat \Psi; D, A \boxtimes B, \Delta) \mid \Omega \uStile E}
			
			\infer2[\rUnitUn]{P:(\widehat{\Phi_i}, C \stackrel{\alpha}{\uarrow} D, \widehat \Psi; \Gamma_i, A \boxtimes B, \Delta)_i \mid \Omega \uStile E}
		\end{prooftree}$, then we apply the induction hypothesis to $\pi_2$.
		
		We obtain $\pi_2'$ such that $\pi_2 \equiv \begin{prooftree}
			\hypo{\pi_2'}
			\infer[no rule]1{P:(\widehat \Psi; D, A, B, \Delta) \mid \Omega \uStile C}
			\infer1[\rBtL]{P:(\widehat \Psi; D, A \otimes B, \Delta) \mid \Omega \uStile C}
		\end{prooftree}$
		
		With a trivial rules commutation, we have
		\[\pi \equiv \begin{prooftree}
			\hypo{\pi_1}
			\infer[no rule]1{P:(\widehat{\Phi_i}; \Gamma_i)_i \uStile C}
			
			\hypo{\pi_2'}
			\infer[no rule]1{P:(\widehat \Psi; D, A, B, \Delta) \mid \Omega \uStile E}
			
			\infer2[\rUnitUn]{P:(\widehat{\Phi_i}, C \stackrel{\alpha}{\uarrow} D, \widehat \Psi; \Gamma_i, A, B, \Delta)_i \mid \Omega \uStile E}
			\infer[double]1[\rBtL]{P:(\widehat{\Phi_i}, C \stackrel{\alpha}{\uarrow} D, \widehat \Psi; \Gamma_i, A \boxtimes B, \Delta)_i \mid \Omega \uStile E}
		\end{prooftree}  \]
		
		All other cases are either similar to the above, or trivial.
	\end{proof}
	
	Using similar proofs, we can easily prove the next two lemmas:
	\begin{lemma}
		\label{lem:posOp}
		Let $\pi$ a proof of $P:(\widehat \Phi; A \oplus B, \Gamma) \mid \Omega \uStile C$.\\
		There exists $\pi'$ a proof of $P:(\widehat \Phi; A, \Gamma) \mid (\widehat \Phi; B, \Gamma) \mid \Omega \uStile C$ such that \[\pi \equiv \begin{prooftree}
			\hypo{\pi'}
			\infer[no rule]1{P:(\widehat \Phi; A, \Gamma) \mid (\widehat \Phi; B, \Gamma) \mid \Omega \uStile C}
			\infer1[\rOpL]{P:(\widehat \Phi; A \oplus B, \Gamma)\mid \Omega \uStile C}
		\end{prooftree}\]
	\end{lemma}
	
	\begin{lemma}
		\label{lem:posOne}
		Let $\pi$ a proof of $P:(\widehat \Phi; \unit, \Gamma) \mid \Omega \uStile C$.\\
		There exists $\pi'$ a proof of $P:(\widehat \Phi; \Gamma) \mid \Omega \uStile C$ such that
		\[\pi \equiv \begin{prooftree}
			\hypo{\pi'}
			\infer[no rule]1{P:(\widehat \Phi; \Gamma) \mid \Omega \uStile C}
			\infer1[\rOneL]{P:(\unit, \widehat \Phi; \Gamma) \mid \Omega \uStile C}
		\end{prooftree}\]
	\end{lemma}

\section{Proofs of Theorems}

	\subsection{Proof of Coherence}
	
	\consistency*
	
	\begin{proof}
		\label{proofcoh}
		We proceed by induction over the structure of $\pi$.
		
		\proofsubparagraph{Initial Case}
		There are two axiom rules, \rAx\ and \rOneR. They are both trivial.
		
		\proofsubparagraph{Induction, $\lStile$ Fragment}
		Let $\pi$ a proof.
			\begin{itemize}
				\item If $\pi = \begin{prooftree}
						\hypo{\pi'}
						\infer[no rule]1{\Phi, G, F, \Psi \lStile H}
						\infer1[\rExchLin]{\Phi, F, G, \Psi \lStile H}
				\end{prooftree}$, then  $\sem \pi = \sem {\pi'} \circ (\Id{\Phi} \otimes \sigma_{F, G} \otimes \Id{\Psi})$.
				
				Let $\varphi \in \sem \Phi, f \in \sem F, g \in \sem G, \psi \in \sem \Psi$. We want to prove that $\sem \pi(\varphi \otimes f \otimes g \otimes \psi) \in \sem H$.
				
				\[\begin{aligned}\sem \pi(\varphi \otimes f \otimes g \otimes \psi)&=  \sem {\pi'} \circ (\Id{\Phi} \otimes \sigma_{F, G} \otimes \Id{\Psi})(\varphi \otimes f \otimes g \otimes \psi)\\
				&= \sem{\pi'}(\varphi \otimes g \otimes f \otimes \psi) \in \sem H \ \text{by induction hypothesis} \end{aligned}\]
				
				Therefore the \rExchLin\ case is verified.
				
				\item If $\pi = \begin{prooftree}
					\hypo{\pi'}
					\infer[no rule]1{\emptyset: (\widehat \Phi; A) \uStile B}
					\infer1[\rUnitRight]{\Phi \lStile A \uarrow B}
				\end{prooftree}$ with $\Phi \labrel \widehat \Phi$, then $\sem \pi = \sem{\pi'}^{(1)}$.\\
				Let $\varphi \in \sem \Phi$. Let us prove that $\sem {\pi'}^{(1)}_{\varphi} \in \sem {A \uarrow B}$.\\
				Let $\Phi = F_1, ..., F_n$\\
				Therefore, $\widehat \Phi = F_1^{\alpha_1},...,F_n^{\alpha_n}$, with the $\alpha_i$ distinct, and $\varphi = \varphi_1 \otimes ... \otimes \varphi_n$ with $\varphi_i \in \sem{F_i}$.\\
				Trivially, $v: F_i^{\alpha_i} \mapsto \varphi_i$ is a consistent valuation for $\pi'$.

				By induction hypothesis, we have $\sem{\pi'}^{(1)\dag}_\varphi \circ \sem {\pi'}^{(1)}_\varphi = \Id A$ and $\sem{\pi'}^{(1)}_\varphi \circ \sem {\pi'}^{(1)\dag}_{\varphi} = \Id B$\\
				Therefore $\sem{\pi'}^{(1)}_\varphi \in \sem{A \uarrow B}$.

				\item If $\pi =\begin{prooftree}
					\hypo{\pi'}
					\infer[no rule]1{F, G, \Phi \lStile H}
					\infer1[\rOtLinL]{F \otimes G, \Phi \lStile H}
				\end{prooftree}$ the case is trivial.
				
				\item If $\pi = \begin{prooftree}
						\hypo{\pi_1}
						\infer[no rule]1{\Phi \lStile F}
						
						\hypo{\pi_2}
						\infer[no rule]1{\Psi \lStile G}
						\infer2[\rOtLinR]{\Phi, \Psi \lStile F \otimes G}
				\end{prooftree}$ then $\sem \pi = \sem{\pi_1} \otimes \sem{\pi_2}$.
				
				Let $\varphi \in \sem \Phi$, $\psi \in \sem \Psi$.\\
				By induction hypothesis, $\sem \pi(\varphi \otimes \psi) = (\sem {\pi_1}(\varphi)) \otimes (\sem{\pi_2}(\psi)) \in \sem{F \otimes G}$.
				
				\item If $\pi = \begin{prooftree}
					\hypo{\pi_1}
					\infer[no rule]1{\Phi \lStile F}
					
					\hypo{\pi_2}
					\infer[no rule]1{G, \Psi \lStile H}
					\infer2[\rArLinL]{\Phi, F \multimap G, \Psi \lStile H}
				\end{prooftree}$ then $\sem \pi = \sem{\pi_2} \circ ((\app_{F, G} \circ (\sem {\pi_1} \otimes \Id{F \multimap G}))\otimes \Id{\Psi})$.
				
				Let $\varphi \in \sem \Phi$, $s \in \sem{F \multimap G}$, $\psi \in \sem \Psi$. Let us prove that $\sem{\pi}(\varphi \otimes s \otimes \psi) \in \sem H$.
				
				\[
					\sem{\pi}(\varphi \otimes s \otimes \psi)= (\sem{\pi_2} \circ ((\app_{F, G} \circ (\sem {\pi_1} \otimes \Id{F \multimap G}))\otimes \Id{\Psi}))(\varphi \otimes s \otimes \psi)\]
				\[
					= \sem{\pi_2}((\app_{F,G}((\sem{\pi_1}(\varphi))\otimes s))\otimes \psi)\]
				\[
					=\sem{\pi_2}(s(\sem{\pi_1}(\varphi)) \otimes \psi)\]
					By induction hypothesis over $\pi_2$, we know that $\sem{\pi_1}(\varphi) \in \sem F$, therefore, by definition $s(\sem{\pi_1}(\varphi)) \in \sem G$.\\
					Therefore, by induction hypothesis over $\pi_2$, $\sem{\pi_2}(s(\sem{\pi_1}(\varphi)) \otimes \psi) \in \sem H$.
				
				\item If $\pi = \begin{prooftree}
					\hypo{\pi'}
					\infer[no rule]1{F, \Phi \lStile G}
					\infer1[\rArLinR]{\Phi \lStile F \multimap G}
				\end{prooftree}$ then $\sem \pi = \varphi \mapsto (f \mapsto \sem{\pi'}(f \otimes \varphi))$.\\
				Let $\varphi \in \sem \Phi$. We want to prove that $f \mapsto \sem{\pi'}(f \otimes \varphi) \in \sem {F \multimap G}$\\
				Thus, let $f \in \sem F$. We simply have to prove that $\sem {\pi'}(f \otimes \varphi) \in \sem G$.\\
				This is simply the induction hypothesis.
				
				\item If $\pi = \begin{prooftree}
					\hypo{\pi'}
					\infer[no rule]1{F_x, \Phi \lStile H}
					\infer1[\rWiLinL x]{F_1 \with F_2, \Phi \lStile H}
				\end{prooftree}$ then $\sem \pi = \sem{\pi'}\circ(\proj{F_1, F_2}{x} \otimes \Id{\Phi})$\\
				Let $f \in \sem{F_1 \with F_2}$, $\varphi \in \sem \Phi$.\\
				By definition, we have $f = (f_1, f_2)$, with $f_1 \in \sem{F_1}$ and $f_2 \in \sem{F_2}$.
				
				By induction hypothesis over $\pi'$, we have $\sem \pi(f\otimes\varphi) = \sem{\pi'}(f_x \otimes \varphi) \in \sem H$.
				
				\item If $\pi = \begin{prooftree}
					\hypo{\pi_1}
					\infer[no rule]1{\Phi \lStile F}
					
					\hypo{\pi_2}
					\infer[no rule]1{\Phi \lStile G}
					\infer2[\rWiLinR]{\Phi \lStile F \with G}
				\end{prooftree}$ then $\sem \pi = \varphi \mapsto (\sem{\pi_1}(\varphi), \sem{\pi_2}(\varphi))$\\
					Let $\varphi \in \sem \Phi$.\\
					By induction hypothesis over $\pi_1$ and $\pi_2$, $\sem{\pi_1}(\varphi) \in \sem F$, and $\sem{\pi_2}(\varphi) \in \sem G$.\\
					Therefore, $\sem \pi (\varphi) \in \sem{F \with G}$.
					
				\item The \rCutLin\ case proceeds like the \rArLinL\ case.
				
			\end{itemize}
			
			\proofsubparagraph{Induction, $\uStile$ Fragment}
			\begin{itemize}
				\item If $\pi = \begin{prooftree}
					\hypo{\pi'}
					\infer[no rule]1{P:\Omega \mid \omega' \mid \omega \mid \Omega' \uStile C}
					\infer1[\rExchUnW]{P:\Omega \mid \underbrace{\omega}_{j} \mid \underbrace{\omega'}_{j+1} \mid \Omega' \uStile C}
				\end{prooftree}$ then $\sem{\pi}^{(i)} = \left\lbrace\begin{aligned}
				\sem {\pi'}^{(j+1)} && \text{if $i = j$}\\
				\sem {\pi'}^{(j)} && \text{if $i = j+1$}\\
				\sem {\pi'}^{(i)} && \text{otherwise}
				\end{aligned}\right.$
				
				Let us write $\Omega \mid \omega \mid \omega' \mid \Omega' = (\widehat{\Phi_i}; \Gamma_i)_{i \in I}$.
				
				Let $v$ a consistent valuation for $\pi$. It is trivially a consistent valuation for $\pi'$.
				
				Let $\varphi_i = v(\widehat{\Phi_i})$, and $i, i' \in I\setminus\{j, j+1\}$.
				
				By induction hypothesis:
				\begin{itemize}
					\item  $\sem{\pi}^{(i)\dag}_{\varphi_i} \circ \sem \pi^{(i)}_{\varphi_i} = \sem{\pi'}^{(i)\dag}_{\varphi_i} \circ \sem {\pi'}^{(i)}_{\varphi_i}=\Id {\Gamma_i}$
					\item $\sem{\pi}^{(i)\dag}_{\varphi_i} \circ \sem \pi^{(i')}_{\varphi_{i'}} = \sem{\pi'}^{(i)\dag}_{\varphi_i} \circ \sem {\pi'}^{(i')}_{\varphi_{i'}}=\zeromap$
					\item $\sum_{i \in I}\sem{\pi}^{(i)}_{\varphi_i} \circ \sem{\pi}^{(i)\dag}_{\varphi_i} = \sum_{i \in I}\sem{\pi'}^{(i)}_{\varphi_i} \circ \sem{\pi'}^{(i)\dag}_{\varphi_i} = \Id C$
				\end{itemize}
				
				Furthermore:
				\begin{itemize}
					\item $\sem{\pi}^{(j)\dag}_{\varphi_j} \circ \sem \pi^{(j)}_{\varphi_j} = \sem{\pi'}^{(j+1)\dag}_{\varphi_{j}} \circ \sem {\pi'}^{(j+1)}_{\varphi_{j}}=\Id {\Gamma_i}$
					\item $\sem{\pi}^{(j)\dag}_{\varphi_j} \circ \sem \pi^{(j+1)}_{\varphi_{j+1}} = \sem{\pi'}^{(j+1)\dag}_{\varphi_{j}} \circ \sem {\pi'}^{(j)}_{\varphi_{j+1}} = \zeromap $
					\item $\sem{\pi}^{(j)\dag}_{\varphi_j} \circ \sem \pi^{(i)}_{\varphi_i} = \sem{\pi'}^{(j+1)\dag}_{\varphi_{j}} \circ \sem {\pi'}^{(i)}_{\varphi_{i}} = \zeromap$
				\end{itemize}
				(And symetrically for the $j+1$ cases).
				
				Therefore, $\sem \pi$ is a pre-unitary.
				
				\item If $\pi = \begin{prooftree}
					\hypo{\pi'}
					\infer[no rule]1{P:(\widehat {\Phi_1}; \Gamma, B, A, \Delta) \mid \Omega \uStile C}
					\infer1[\rExchUnG]{P:(\widehat {\Phi_1}; \Gamma, A, B, \Delta) \mid \Omega \uStile C}
				\end{prooftree}$ then,
				\[\sem \pi^{(i)} = \left \lbrace \begin{aligned}
					\varphi \mapsto \sem{\pi'}^{(1)}_\varphi \circ (\Id \Gamma \otimes \sigma_{A, B} \otimes \Id \Delta) && \text{if $i = 1$}\\
					\sem{\pi'}^{(i)} && \text{otherwise}
				\end{aligned} \right.\]
				We also note $\Omega = (\widehat{\Phi_i}; \Lambda_i)_{2 \leq i \leq n}$.\\
				Trivially, if $v$ is a consistent valuation for $\pi$, it is a consistent assignment for $\pi'$.\\
				Again, we write $\varphi_i$ for $v(\widehat{\Phi_i})$.\\
				Let $i \in I \setminus \{1\}$. From the induction hypothesis, we have:
				\begin{itemize}
					\item First,
					\[\sem{\pi}^{(1)\dag}_{\varphi_1} \circ \sem{\pi}^{(1)}_{\varphi_1} = (\Id \Gamma \otimes \sigma_{B,A} \otimes \Id \Delta) \circ \sem{\pi'}^{(1)\dag}_{\varphi_1} \circ \sem{\pi'}^{(1)}_{\varphi_1} \circ (\Id \Gamma \otimes \sigma_{A, B} \otimes \Id \Delta)\]
					\[= (\Id \Gamma \otimes \sigma_{B,A} \otimes \Id \Delta) \circ \Id{\Gamma \otimes B \otimes A \otimes \Delta} \circ (\Id \Gamma \otimes \sigma_{A, B} \otimes \Id \Delta)\]
					\[= \Id{\Gamma \otimes A \otimes B \otimes \Delta}\]
					
					\item Let $i \in I \setminus \{1\}$:
					\[
					\sem{\pi}^{(1)\dag}_{\varphi_1} \circ \sem{\pi}^{(i)}_{\varphi_1} = (\Id \Gamma \otimes \sigma_{B,A} \otimes \Id \Delta) \circ \sem{\pi'}^{(1)\dag}_{\varphi_1} \circ \sem{\pi'}^{(i)}_{\varphi_i}\]
					\[= (\Id \Gamma \otimes \sigma_{B,A} \otimes \Id \Delta) \circ \zeromap\]
					\[=\zeromap\]
					
					\item Finally,
					\[
					\sum_{i \in I} \sem{\pi}^{(i)}_{\varphi_i} \circ \sem{\pi}^{(i)\dag}_{\varphi_i}\]
					\[= \sem{\pi'}^{(1)}_{\varphi_1} \circ (\Id \Gamma \otimes \sigma_{A, B} \otimes \Id \Delta) \circ (\Id \Gamma \otimes \sigma_{B,A} \otimes \Id \Delta) \circ \sem{\pi}^{(1)\dag}_{\varphi_1}+ \sum_{i\geq 2}^n \sem{\pi'}^{(i)}_{\varphi_i} \circ \sem{\pi'}^{(i)\dag}_{\varphi_i}\]
					\[= \sem{\pi'}^{(1)}_{\varphi_1} \circ \sem{\pi'}^{(1)\dag}_{\varphi_1} + \sum_{i\geq 2}^n \sem{\pi'}^{(i)}_{\varphi_i} \circ \sem{\pi'}^{(i)\dag}_{\varphi_i}\]
					\[= \sum_{i\in I} \sem{\pi'}^{(i)}_{\varphi_i} \circ \sem{\pi'}^{(i)\dag}_{\varphi_i}\]
					\[= \Id C\]

				\end{itemize}
				And this is enough to prove $\sem \pi$ is a pre-unitary.
				
				\item The cases of \rOneL, \rBtL\ and \rPhTh\ are all handled the same way and are trivial.
				
				\item If $\pi = \begin{prooftree}
					\hypo{\pi_1}
					\infer[no rule]1{P:(\widehat{\Phi_i}; \Gamma_i)_{1 \leq i \leq n} \uStile A}
					
					\hypo{\pi_2}
					\infer[no rule]1{P:(\widehat \Psi; B, \Delta) \mid \Omega \uStile C}
					
					\infer2[\rUnitUn]{P:(\widehat{\Phi_i}, A \stackrel{\alpha}{\uarrow} B, \widehat \Psi; \Gamma_i, \Delta)_{1 \leq i \leq n} \mid \Omega \uStile C}
				\end{prooftree}$ then:
				\begin{itemize}
					\item For all $1 \leq i \leq n$, $\sem{\pi}^{(i)}_{\varphi\otimes f \otimes \psi} = \sem{\pi_2}^{(1)}_\psi \circ ((f\circ \sem {\pi_1}^{(i)}_\varphi) \otimes \Id \Delta)$
					\item For all $i > n$, $\sem \pi^{(i)} = \sem {\pi_2}^{(i-n+1)}$.
				\end{itemize}
				We write $\Omega = (\widehat{\Theta_i}; \Lambda_i)_{n+1 \leq i \leq m}$.
				
				Let $v$ a consistent valuation for $\pi$.\\
				Let $f = v(A \stackrel{\alpha}{\uarrow} B)$, $\varphi_i = v(\widehat{\Phi_i})$, $\psi = v(\widehat \Psi)$ and $t_i = v(\widehat{\Theta_i})$.
				
				Let us prove the pre-unitarity of $\sem \pi$.
				\begin{itemize}
					\item Let $1 \leq i \leq n$.
					\[\sem{\pi}^{(i)\dag}_{\varphi_i \otimes f \otimes \psi} \circ \sem{\pi}^{(i)}_{\varphi_i \otimes f \otimes \psi}\]
					\[=((f\circ \sem {\pi_1}^{(i)}_{\varphi_i}) \otimes \Id \Delta)^\dag \circ \sem{\pi_2}^{(1)\dag}_\psi \circ \sem{\pi_2}^{(1)}_\psi \circ ((f\circ \sem {\pi_1}^{(i)}_{\varphi_i}) \otimes \Id \Delta)\]
					\[= ((f\circ \sem {\pi_1}^{(i)}_{\varphi_i}) \otimes \Id \Delta)^\dag \circ ((f\circ \sem {\pi_1}^{(i)}_{\varphi_i}) \otimes \Id \Delta)\]
					\[= ({\pi_1}^{(i)\dag}_{\varphi_i} \circ f^\dag \circ f \circ {\pi_1}^{(i)}_{\varphi_i}) \otimes \Id \Delta\]
					\[= ({\pi_1}^{(i)\dag}_{\varphi_i} \circ {\pi_1}^{(i)}_{\varphi_i})\]
					\[= \Id {\Gamma_i} \otimes \Id \Delta\]
					\[= \Id {\Gamma_i, \Delta}\]
					
					\item Let $n+1 \leq i \leq m$
						\[
							\sem{\pi}^{(i)\dag}_{t_i} \circ \sem{\pi}^{(i)}_{t_i} = \sem{\pi_2}^{(i-n+1)\dag}_{t_i} \circ \sem{\pi_2}^{(i-n+1)}_{t_i}
							= \Id{\Lambda_i} \]
						
					\item Let $1 \leq i \leq n$, $1 \leq j \leq n$ such that $i \neq j$
					\[
						\sem{\pi}^{(i)\dag}_{\varphi_i \otimes f \otimes \psi} \circ \sem{\pi}^{(j)}_{\varphi_j \otimes f \otimes \psi}\]
					\[=((f\circ \sem {\pi_1}^{(i)}_{\varphi_i}) \otimes \Id \Delta)^\dag
						\circ \sem{\pi_2}^{(1)\dag}_\psi \circ \sem{\pi_2}^{(1)}_\psi \circ ((f\circ \sem {\pi_1}^{(j)}_{\varphi_j}) \otimes \Id \Delta)\]
					\[= ((f\circ \sem {\pi_1}^{(i)}_{\varphi_i}) \otimes \Id \Delta)^\dag \circ ((f\circ \sem {\pi_1}^{(j)}_{\varphi_j}) \otimes \Id \Delta)\]
					\[= ({\pi_1}^{(i)\dag}_{\varphi_i} \circ f^\dag \circ f \circ {\pi_1}^{(j)}_{\varphi_j}) \otimes \Id \Delta\]
					\[= ({\pi_1}^{(i)\dag}_{\varphi_i} \circ {\pi_1}^{(j)}_{\varphi_i})\]
					\[= \zeromap \otimes \Id \Delta\]
					\[= \zeromap\]
					
					\item Let $1 \leq i \leq n < j \leq m$.
						\[
							\sem{\pi}^{(i)\dag}_{\varphi_i \otimes f \otimes \psi} \circ \sem{\pi}^{(j)}_{t_j} = ((f\circ \sem {\pi_1}^{(i)}_{\varphi_i}) \otimes \Id \Delta)^\dag \circ \sem{\pi_2}^{(1)\dag}_\psi \circ \sem{\pi_2}^{(j-n+1)}_{t_j}\]
						\[
							= ((f\circ \sem {\pi_1}^{(i)}_{\varphi_i}) \otimes \Id \Delta)^\dag \circ \zeromap\]
						\[= \zeromap
						\]
						(and similarly for $\sem{\pi}^{(j)\dag}_{t_j} \circ \sem \pi^{(i)}_{\varphi_i \otimes f \otimes \psi}$).
						
					\item For all $1 \leq i \leq n$,
						\[
						\sem{\pi}^{(i)}_{\varphi_i \otimes f \otimes \psi} \circ \sem{\pi}^{(i)\dag}_{\varphi_i \otimes f \otimes \psi}\]
						\[= \sem{\pi_2}^{(1)}_\psi \circ ((f\circ \sem {\pi_1}^{(i)}_{\varphi_i}) \otimes \Id \Delta) \circ ((f\circ \sem {\pi_1}^{(i)}_{\varphi_i}) \otimes \Id \Delta)^\dag \circ  \sem{\pi_2}^{(1)\dag}_\psi\]
						\[
						= \sem{\pi_2}^{(1)}_\psi \circ ((f \circ \sem{\pi_1}^{(i)}_{\varphi_i} \circ \sem{\pi_1}^{(i)\dag}_{\varphi_i} \circ f^\dag) \otimes \Id \Delta) \circ \sem{\pi_2}^{(1)\dag}_\psi
						\]
						
						Therefore:
						\[
							\sum_{i=1}^n \sem{\pi}^{(i)}_{\varphi_i \otimes f \otimes \psi} \circ \sem{\pi}^{(i)\dag}_{\varphi_i \otimes f \otimes \psi}\]
							\[= \sem{\pi_2}^{(1)}_\psi \circ \left(\left(f \circ \left( \sum_{i=1}^n\sem{\pi_1}^{(i)}_{\varphi_i} \circ \sem{\pi_1}^{(i)\dag}_{\varphi_i}\right) \circ f^\dag \right) \otimes \Id \Delta\right) \circ \sem{\pi_2}^{(1)\dag}_\psi\]
							\[= \sem{\pi_2}^{(1)}_\psi \circ \left(\left(f \circ f^\dag \right) \otimes \Id \Delta\right) \circ \sem{\pi_2}^{(1)\dag}_\psi\]
							\[
							= \sem{\pi_2}^{(1)}_\psi \circ \sem{\pi_2}^{(1)\dag}_\psi
						\]
						
						And we now have
						\[
							\sum_{i=1}^n \sem{\pi}^{(i)}_{\varphi_i \otimes f \otimes \psi} \circ \sem{\pi}^{(i)\dag}_{\varphi_i \otimes f \otimes \psi} + \sum_{i=n+1}^{m} \sem{\pi}^{(i)}_{t_i} \circ \sem{\pi}^{(i)\dag}_{t_i}\]
						\[=\sem{\pi_2}^{(1)}_\psi \circ \sem{\pi_2}^{(1)\dag}_\psi + \sum_{i=n+1}^{m} \sem{\pi_2}^{(i-n+1)}_{t_i} \circ \sem{\pi_2}^{(i-n+1)\dag}_{t_i}\]
						\[=\Id C\]
						
				\end{itemize}
				And we have proven that $\sem \pi$ is pre-unitary.
				
				\item If $\Omega_1 = (\widehat{\Phi_i}; \Gamma_i)_{1 \leq i \leq n}$, $\Omega_2 = (\widehat {\Psi_j}; \Delta_j)_{1 \leq j \leq m}$ and $\pi = \begin{prooftree}
					\hypo{\pi_1}
					\infer[no rule]1{P:\Omega_1 \uStile A}
					
					\hypo{\pi_2}
					\infer[no rule]1{P:\Omega_2 \uStile B}
					\infer2[\rBtR]{P:\Omega_1 \times \Omega_2 \uStile A \boxtimes B}
				\end{prooftree}$, then, for all $i, j$, $\sem{\pi}^{(ij)} = \sem {\pi_1}^{(i)} \otimes \sem {\pi_2}^{(j)}$
				
				Let $v$ a consistent valuation for $\pi$. It is trivially a consistent valuation for $\pi_1$ and $\pi_2$.\\
				We write $\varphi_i = v(\widehat{\Phi_i})$ and $\psi_j = v(\widehat{\Psi_j})$.
				
				Let us now check the conditions for pre-unitarity. 
				\begin{itemize}
					\item Let $1 \leq i \leq n$ and $1 \leq j \leq m$.
					\[
						\sem{\pi}^{(ij)\dag}_{\varphi_i \otimes \psi_j} \circ \sem{\pi}^{(ij)}_{\varphi_i \otimes \psi_j} =\left(\sem{\pi_1}^{(i)}_{\varphi_i} \otimes \sem{\pi_2}^{(j)}_{\psi_j}\right)^\dag \circ \left(\sem{\pi_1}^{(i)}_{\varphi_i} \otimes \sem{\pi_2}^{(j)}_{\psi_j}\right)\]
						\[= \left(\sem{\pi_1}^{(i)\dag}_{\varphi_i} \circ \sem{\pi_1}^{(i)}_{\varphi_i} \right) \otimes \left(\sem{\pi_2}^{(j)\dag}_{\psi_j} \circ \sem{\pi_2}^{(j)}_{\psi_j} \right)\]
						\[= \Id {\Gamma_i} \otimes \Id {\Delta_j}\]
						\[= \Id{\Gamma_i, \Delta_j}\]
					
					\item Let $1 \leq i' \leq n$ and $1 \leq j' \leq n$\\
						If $i \neq i'$ then
						\[\sem{\pi}^{(ij)\dag}_{\varphi_i \otimes \psi_j} \circ \sem{\pi}^{(i'j')}_{\varphi_{i'} \otimes \psi_{j'}}=\left(\sem{\pi_1}^{(i)}_{\varphi_i} \otimes \sem{\pi_2}^{(j)}_{\psi_j}\right)^\dag \circ \left(\sem{\pi_1}^{(i')}_{\varphi_{i'}} \otimes \sem{\pi_2}^{(j')}_{\psi_{j'}}\right)\]
						\[= \left(\sem{\pi_1}^{(i)\dag}_{\varphi_i} \circ \sem{\pi_1}^{(i')}_{\varphi_{i'}} \right) \otimes \left(\sem{\pi_2}^{(j)\dag}_{\psi_j} \circ \sem{\pi_2}^{(j')}_{\psi_{j'}} \right)\]
						\[= \zeromap \otimes \left(\sem{\pi_2}^{(j)\dag}_{\psi_j} \circ \sem{\pi_2}^{(j')}_{\psi_{j'}} \right)\]
						\[= \zeromap\]
						And similarly for the case where $j \neq j'$.
					
					\item And finally:
						\[
						\sum_{i,j} \sem{\pi}^{(ij)}_{\varphi_i \otimes \psi_j} \circ \sem{\pi}^{(ij)\dag}_{\varphi_i \otimes \psi_j}
						= \sum_{ij}
						\left(\sem{\pi_1}^{(i)}_{\varphi_i} \otimes \sem{\pi_2}^{(j)}_{\psi_j}\right) \circ \left(\sem{\pi_1}^{(i)}_{\varphi_{i}} \otimes \sem{\pi_2}^{(j)}_{\psi_{j}}\right)^\dag\]
						\[= \sum_{i,j} \left(\sem{\pi_1}^{(i)}_{\varphi_i} \circ \sem {\pi_1}^{(i)\dag}_{\varphi_i}\right) \otimes \left(\sem{\pi_2}^{(j)}_{\psi_j} \circ \sem {\pi_2}^{(j)\dag}_{\psi_j}\right)\]
						\[= \left(\sum_{i} \sem{\pi_1}^{(i)}_{\varphi_i} \circ \sem {\pi_1}^{(i)\dag}_{\varphi_i} \right) \otimes \left(\sum_j \sem{\pi_2}^{(j)}_{\psi_j} \circ \sem {\pi_2}^{(j)\dag}_{\psi_j} \right)\]
						\[= \Id{A} \otimes \Id B\]
						\[= \Id{A \boxtimes B}\]
				\end{itemize}
				
				And therefore $\sem \pi$ is also pre-unitary.
				
				\item If $\pi = \begin{prooftree}
					\hypo{\pi'}
					\infer[no rule]1{(\widehat \Phi; A, \Gamma) \mid (\widehat \Phi; B, \Gamma) \mid \Omega \uStile C}
					\infer1[\rOpL]{(\widehat \Phi; A \oplus B, \Gamma) \mid \Omega \uStile C}
				\end{prooftree}$ then, for all $i > 1$, and for all $\varphi \in \cV_{\Phi}$:
				\[\left \lbrace \begin{aligned}
					\sem \pi^{(1)}_\varphi &= \sem{\pi'}^{(1)}_\varphi \circ (\proj{A,B}1 \otimes \Id \Gamma) + \sem{\pi'}^{(2)}_\varphi \circ (\proj{A,B}2 \otimes \Id\Gamma)\\
					\sem \pi^{(i)} &= \sem {\pi'}^{(i+1)}
				\end{aligned}\right.\]
				
				Let $v$ a consistent valuation for $\pi$. It is trivially a consistent valuation for $\pi'$.
				
				Let $\varphi_1 = v(\widehat \Phi)$, and for $i > 1$, $\varphi_i$ the valuation of the higher terms of the $i$-th world of $\pi$.
				
				Let us now prove pre-unitarity.
				\begin{itemize}
					\item First,
					 \[\begin{aligned}
						\sem{\pi}^{(1)\dag}_{\varphi_1} \circ \sem{\pi}^{(1)}_{\varphi_1} &= \left( \sem{\pi'}^{(1)}_{\varphi_1} \circ (\proj{A,B}1 \otimes \Id \Gamma) + \sem{\pi'}^{(2)}_{\varphi_1} \circ (\proj{A,B}2 \otimes \Id\Gamma) \right)^\dag\\&\phantom{=} \circ \left(\sem{\pi'}^{(1)}_{\varphi_1} \circ (\proj{A,B}1 \otimes \Id \Gamma) + \sem{\pi'}^{(2)}_{\varphi_1} \circ (\proj{A,B}2 \otimes \Id\Gamma) \right)\end{aligned}\]
						
						\[\begin{aligned}
						=&\ (\proj{A,B}1 \otimes \Id \Gamma)^\dag \circ \sem{\pi'}^{(1)\dag}_{\varphi_1} \circ \sem{\pi'}^{(1)}_{\varphi_1} \circ (\proj{A,B}1 \otimes \Id \Gamma)\\
						& + (\proj{A,B}1 \otimes \Id \Gamma)^\dag \circ \sem{\pi'}^{(1)\dag}_{\varphi_1} \circ \sem{\pi'}^{(2)}_{\varphi_1} \circ (\proj{A,B}2 \otimes \Id\Gamma)\\
						& + (\proj{A,B}2 \otimes \Id\Gamma)^\dag \circ \sem{\pi'}^{(2)\dag}_{\varphi_1} \circ \sem{\pi'}^{(1)}_{\varphi_1} \circ (\proj{A,B}1 \otimes \Id \Gamma)\\
						& +(\proj{A,B}2 \otimes \Id\Gamma)^\dag \circ \sem{\pi'}^{(2)\dag}_{\varphi_1} \circ \sem{\pi'}^{(2)}_{\varphi_1} \circ (\proj{A,B}2 \otimes \Id\Gamma)\end{aligned}\]
						\[\begin{aligned}
						=&\ (\proj{A,B}1 \otimes \Id \Gamma)^\dag \circ \Id{A, \Gamma} \circ (\proj{A,B}1 \otimes \Id \Gamma)\\
						&+ (\proj{A,B}1 \otimes \Id \Gamma)^\dag \circ \zeromap \circ (\proj{A,B}2 \otimes \Id\Gamma)\\
						&+ (\proj{A,B}2 \otimes \Id\Gamma)^\dag \circ \zeromap \circ (\proj{A,B}1 \otimes \Id \Gamma)\\
						& +(\proj{A,B}2 \otimes \Id\Gamma)^\dag \circ \Id{B, \Gamma}\circ (\proj{A,B}2 \otimes \Id\Gamma)\end{aligned}\]
						\[=(\inj{A,B}1 \circ \Id A \circ \proj{A,B}1)\otimes \Id \Gamma + (\inj{A,B}2 \circ \Id B \circ \proj{A,B}2) \otimes \Id \Gamma\]
						\[=\Id{A \oplus B} \otimes \Id{\Gamma}\]
						\[=\Id{A\oplus B, \Gamma}\]
					
					\item Let $i > 1$,
						\[
						\sem{\pi}^{(1)\dag}_{\varphi_1} \circ \sem{\pi}^{(i)}_{\varphi_i} =\left( \sem{\pi'}^{(1)}_\varphi \circ (\proj{A,B}1 \otimes \Id \Gamma) + \sem{\pi'}^{(2)}_\varphi \circ (\proj{A,B}2 \otimes \Id\Gamma) \right)^\dag \circ {\pi'}^{(i+1)}_{\varphi_i}\]
						\[=(\proj{A,B}1 \otimes \Id \Gamma)^\dag \circ \sem{\pi'}^{(1)\dag}_{\varphi_1} \circ \sem{\pi'}^{(i-1)}_{\varphi_i} + (\proj{A,B}2 \otimes \Id\Gamma)^\dag \circ \sem{\pi'}^{(2)\dag}_{\varphi_1} \circ \sem{\pi'}^{(i+1)}_{\varphi_i}\]
						\[= \zeromap + \zeromap\]
						\[= \zeromap\]
					
					\item Finally,
					\[\begin{aligned}
						\sem{\pi}^{(1)}_{\varphi_1} \circ \sem{\pi}^{(1)\dag}_{\varphi_1} &= \left( \sem{\pi'}^{(1)}_\varphi \circ (\proj{A,B}1 \otimes \Id \Gamma) + \sem{\pi'}^{(2)}_\varphi \circ (\proj{A,B}2 \otimes \Id\Gamma) \right)\\&\phantom{=} \circ \left( \sem{\pi'}^{(1)}_\varphi \circ (\proj{A,B}1 \otimes \Id \Gamma) + \sem{\pi'}^{(2)}_\varphi \circ (\proj{A,B}2 \otimes \Id\Gamma) \right)^\dag\end{aligned}\]
						\[= \sem{\pi'}^{(1)}_{\varphi_1} \circ \sem{\pi'}^{(1)\dag}_{\varphi_1} + \sem{\pi'}^{(1)}_{\varphi_1} \circ \sem{\pi'}^{(2)\dag}_{\varphi_1}\]
						due to properties of projections\\
						Therefore,
						\[\sum_{i=1}^n \sem{\pi}^{(i)}_{\varphi_i} \circ \sem{\pi}^{(i)\dag}_{\varphi_i} = \sem{\pi'}^{(1)}_{\varphi_1} \circ \sem{\pi'}^{(1)\dag}_{\varphi_1} + \sem{\pi'}^{(1)}_{\varphi_1} \circ \sem{\pi'}^{(2)\dag}_{\varphi_1} + \sum_{i=2}^{n} \sem{\pi}^{(i+1)}_{\varphi_i} \circ \sem{\pi}^{(i+1)\dag}_{\varphi_i} = \Id C\]
				\end{itemize}
				Thus, we have proved that $\sem \pi$ is pre-unitary in the \rOpL\ case.
				
				\item If $\pi = \begin{prooftree}
					\hypo{\pi_1}
					\infer[no rule]1{P: \Omega_1 \uStile A}
					\hypo{\pi_2}
					\infer[no rule]1{P:\Omega_2 \uStile B}
					
					\infer2[\rOpR]{P:\Omega_1 \mid \Omega_2 \uStile A \oplus B}
				\end{prooftree}$,
				
				where $\Omega_1 = (\widehat{\Phi_1}; \Gamma_i)_{1 \leq i \leq n}$ and $\Omega_2 = (\widehat{\Psi_j}; \Delta_j)_{1 \leq j \leq m}$ then, a consistent assignment $v$ for $\pi$ is a consistent assignment for $\pi_1$ and for $\pi_2$.\\
				Let $\varphi_i = v(\widehat{\Phi_i})$ and $\psi_j =v(\widehat{\Psi_j})$.
				
				Let us now prove that $\sem \pi$ is pre-unitary:
				\begin{itemize}
					\item Let $1 \leq i \leq n$. 
						\[\sem{\pi}^{(i)\dag}_{\varphi_i} \circ \sem{\pi}^{(i)}_{\varphi_i} = \sem{\pi_1}^{(i)\dag}_{\varphi_i} \circ \proj{A,B}{1} \circ \inj{A,B}{1} \circ \sem{\pi_1}^{(i)}_{\varphi_i}\]
						\[= \sem{\pi_1}^{(i)\dag}_{\varphi_i} \circ \sem{\pi_1}^{(i)}_{\varphi_i}\]
						\[= \Id {\Gamma_i}\]
						
						And similarly for $1 \leq j \leq m$ and $\sem \pi^{(n+j)\dag}_{\psi_j} \circ \sem \pi^{(n+j)}_{\psi_j}$.
							
					\item Let $1 \leq i \leq n$ and $1 \leq i' \leq n$ with $i \neq i'$.\\
						It is trivial to see that $\sem{\pi}^{(i)\dag}_{\varphi_i} \circ \sem{\pi}^{(i')}_{\varphi_{i'}} = \sem{\pi_1}^{(i)\dag}_{\varphi_i} \circ \sem{\pi_1}^{(i')}_{\varphi_{i'}} = \zeromap$\\
						This is also true for the same property regarding $\pi_2$.
						
						Furthermore, for $1 \leq j \leq m$:
						\[\sem{\pi}^{(i)\dag}_{\varphi_i} \circ \sem{\pi}^{(n+j)}_{\varphi_j}=\sem{\pi_1}^{(i)\dag}_{\varphi_i} \circ \proj{A,B}{1} \circ \inj{A,B}{1} \circ \sem{\pi_2}^{(j)}_{\psi_j}\]
						\[=\sem{\pi_1}^{(i)\dag}_{\varphi_i} \circ \zeromap \circ \sem{\pi_2}^{(j)}_{\psi_j}\]
						\[= \zeromap\]
						And vice versa.
					
					\item We now verify the last condition for pre-unitarity:
						\[\sum_{i=1}^{n} \sem\pi^{(i)}_{\varphi_i} \circ \sem\pi^{(i)\dag}_{\varphi_i} + \sum_{j=1}^{m}\sem\pi^{(n+j)}_{\psi_j} \circ \sem{\pi}^{(n+j)\dag}_{\psi_j}\]
						\[= \sum_{i=1}^n \inj{A, B}1 \circ \sem{\pi_1}^{(i)}_{\varphi_i} \circ \sem {\pi_1}^{(i)\dag}_{\varphi_i} \circ \proj{A,B} 1 +\sum_{j=1}^m \inj{A, B}2 \circ \sem{\pi_2}^{(j)}_{\psi_j} \circ \sem {\pi_2}^{(j)\dag}_{\psi_j} \circ \proj{A,B} 2\]
						\[=\inj{A, B}1 \circ \left( \sem{\pi_1}^{(i)}_{\varphi_i} \circ \sem {\pi_1}^{(i)\dag}_{\varphi_i} \right) \circ \proj{A,B} 1 + \inj{A, B}2 \circ \left(\sem{\pi_2}^{(j)}_{\psi_j} \circ \sem {\pi_2}^{(j)\dag}_{\psi_j} \right)\circ \proj{A,B} 2\]
						\[= \inj{A,B}1 \circ \proj{A,B}1 + \inj{A,B}1 \circ \proj{A, B}2\]
						\[= \Id{A \oplus B}\]
				\end{itemize}
					And therefore, $\sem \pi$ is a pre-unitary.
				
				\item If $\pi = \begin{prooftree}
					\hypo{\pi'}
					\infer[no rule]1{(\widehat \Phi; (\unit \oplus \unit), \Gamma) \mid \Omega \uStile C}
					\infer1[\rRotTh]{(\widehat \Phi; \Gamma)\mid (\widehat \Phi; \Gamma) \mid \Omega \uStile C}
				\end{prooftree}$ then
				\[\sem{\pi}^{(i)} = \left\lbrace\begin{aligned}
					\varphi \mapsto \sem{\pi'}^{(1)}_\varphi\left(\begin{pmatrix}
						\cos \theta \\ \sin \theta
					\end{pmatrix} \otimes \Id \Gamma\right)&&\text{if $i = 1$}\\
					\varphi \mapsto \sem{\pi'}^{(1)}_\varphi \left(\begin{pmatrix}
						-\sin \theta \\ \cos \theta
					\end{pmatrix} \otimes \Id \Gamma\right)&&\text{if $i = 2$}\\
					\sem{\pi'}^{(i-1)}&&\text{otherwise}
				\end{aligned}\right. \]
				
				Let $v$ a consistent assignment for $\pi$. Let $\varphi = v(\widehat \Phi)$. $v$ is trivially a consistent assignment for $\pi'$.
				
				Let us prove preunitarity:
				\begin{itemize}
					\item First,
					\[
						\sem{\pi}^{(1)\dag}_\varphi \circ \sem{\pi}^{(1)}_\varphi= \left(\begin{pmatrix}
							\cos \theta & \sin \theta
						\end{pmatrix} \otimes \Id \Gamma \right) \circ \sem{\pi'}^{(1)\dag}_\varphi \circ \sem{\pi'}^{(1)}_\varphi \circ \left( \begin{pmatrix} \cos \theta \\ \sin \theta \end{pmatrix} \otimes \Id \Gamma \right)\]
						\[= \left(\begin{pmatrix} \cos \theta & \sin \theta \end{pmatrix} \begin{pmatrix} \cos \theta \\ \sin \theta \end{pmatrix} \right) \otimes \Id \Gamma\]
						\[=(\cos^2 \theta + \sin^2 \theta) \otimes \Id \Gamma \]
						\[= (1) \otimes \Id \Gamma\] 
						\[=\Id \Gamma\]
					
					And similarly for $\sem{\pi}^{(2)\dag}_\varphi \circ \sem{\pi}^{(2)}_\varphi = \Id \Gamma$
					
					\item For orthogonlaity,
						\[\sem{\pi}^{(1)\dag}_\varphi \circ \sem{\pi}^{(2)}_\varphi = \left(\begin{pmatrix}
							\cos \theta & \sin \theta
						\end{pmatrix} \otimes \Id \Gamma \right) \circ \sem{\pi'}^{(1)\dag}_\varphi \circ \sem{\pi'}^{(1)}_\varphi \circ \left( \begin{pmatrix} -\sin \theta \\ \cos \theta \end{pmatrix} \otimes \Id \Gamma \right)\]
						\[= \left(\begin{pmatrix} \cos \theta & \sin \theta \end{pmatrix} \begin{pmatrix} -\sin \theta \\ \cos \theta \end{pmatrix} \right) \otimes \Id \Gamma\]
						\[= (\cos \theta \sin \theta - \cos \theta \sin \theta) \otimes \Id \Gamma\]
						\[= (0) \otimes \Id \Gamma\]
						\[= \zeromap\]

					\item Finally, \[\sem{\pi}^{(1)}_\varphi \circ \sem{\pi}^{(1)\dag}_\varphi =\sem{\pi'}^{(1)} \circ \left(\begin{pmatrix}\cos \theta \\ \sin \theta\end{pmatrix} \otimes \Id \Gamma \right) \circ \left(\begin{pmatrix} \cos \theta & \sin \theta \end{pmatrix} \otimes \Id \Gamma \right) \circ \sem{\pi'}^{(1)\dag}_\varphi \]
						\[= \sem{\pi'}^{(1)}_\varphi \circ \left( \begin{pmatrix}
							\cos^2 \theta & \cos \theta \sin \theta \\ \cos \theta \sin \theta & \sin^2 \theta
						\end{pmatrix} \otimes \Id \Gamma \right) \circ \sem{\pi'}^{(1)\dag}_\varphi\]
						
						And similarly, \[\sem{\pi}^{(2)}_\varphi \circ \sem{\pi}^{(2)\dag}_\varphi = \sem{\pi'}^{(1)}_\varphi \circ \left( \begin{pmatrix}
							\sin^2 \theta & -\cos \theta \sin \theta \\ -\cos \theta \sin \theta & \cos^2 \theta
						\end{pmatrix} \otimes \Id \Gamma \right) \circ \sem{\pi'}^{(1)\dag}_\varphi\]
						
						Therefore,
						\[\sem{\pi}^{(1)}_\varphi \circ \sem{\pi}^{(1)\dag}_\varphi + \sem{\pi}^{(2)}_\varphi \circ \sem{\pi}^{(2)\dag}_\varphi\]\[ = \sem{\pi'}^{(1)}_\varphi \circ \left( \begin{pmatrix}
							\cos^2 \theta + \sin^2 \theta & \cos \theta \sin \theta -\cos \theta \sin \theta \\ \cos \theta \sin \theta -\cos \theta \sin \theta & \cos^2 \theta + \sin^2 \theta
						\end{pmatrix} \otimes \Id \Gamma \right)\circ \sem{\pi'}^{(1)\dag}_\varphi\]
						\[ = \sem{\pi'}^{(1)}_\varphi \circ \left( \begin{pmatrix}
							1 & 0 \\ 0 & 1
						\end{pmatrix} \otimes \Id \Gamma \right)\circ \sem{\pi'}^{(1)\dag}_\varphi\]
						\[= \sem{\pi'}^{(1)}_\varphi \circ \sem{\pi'}^{(1)\dag}_\varphi\]

				\end{itemize}
				These three equalities are enough to prove the pre-unitarity of $\sem \pi$, through the induction hypothesis on $\sem{\pi'}$.
				
				\item The case for \rCutUnG\ proceeds the same as the case of \rUnitUn.
			\end{itemize}
			
			We finish with the fragment of $\uStile$ that manipulates higher terms:
			\begin{itemize}
				\item If $\pi =\begin{prooftree}
					\hypo{\pi'}
					\infer[no rule]1{P: (\widehat \Phi, G^\beta, F^\alpha, \widehat \Psi; \Gamma) \mid \Omega \uStile C}
					\infer1[\rExchUnH]{P: (\widehat \Phi, F^\alpha, G^\beta, \widehat \Psi; \Gamma) \mid \Omega \uStile C}
				\end{prooftree}$ then
				\[\sem{\pi}^{(i)} = \left \lbrace \begin{aligned}
					\sem{\pi'}^{(1)}\circ(\Id \Phi \otimes \sigma_{F, G} \otimes \Id \Psi) & \text{if $i = 1$}\\
					\sem{\pi'}^{(i)}&\text{otherwise}
				\end{aligned} \right.\]
				
				Let $v$ a consistent valuation for $\pi$. It is trivially a consistent evaluation for $\pi'$.\\
				From this, the pre-unitarity is trivial.
				
				\item If $\pi = \begin{prooftree}
						\hypo{\pi'}
						\infer[no rule]1{P:(F^{\otimes_1 \alpha}, G^{\otimes_2 \alpha}, \widehat \Phi; \Gamma) \mid \Omega \uStile C}
						\infer1[\rOtUn]{P:((F\otimes G)^\alpha, \widehat \Phi; \Gamma)}
				\end{prooftree}$, let $v$ a consistent valuation for $\pi$, then $v$ is a consistent valuation for $\pi'$.\\
					From this, the preunitarity of $\sem \pi$ is trivial since $\sem \pi = \sem {\pi'}$.
					
				\item If $\pi = \begin{prooftree}
					\hypo{\pi_L}
					\infer[no rule]1{\Phi \lStile F}
					
					\hypo{\pi'}
					\infer[no rule]1{P \sqcup (\widehat \Phi, \pi_L, F \stackrel{\alpha}{\multimap} G, G^\beta): (G^\beta, \widehat \Psi; \Gamma) \mid \Omega \uStile C}
					\infer2[\rArUnNew]{P:(\widehat \Phi, F \stackrel{\alpha}{\multimap} G, \widehat \Psi; \Gamma) \mid \Omega \uStile C}
				\end{prooftree}$, let $v$ a consistent valuation for $\pi$.\\
					Let $\varphi = v(\widehat \Phi)$, $f = v(F \stackrel{\alpha}{\multimap}g)$ and $\psi = v(\widehat \Psi)$.\\
					It is trivial that $v'$ a valuation equal to $v$, except that we also have $v(G^\beta) = f(\sem{\pi_L}(\varphi))$.
					
					Since $\sem{\pi}^{1}_{\varphi \otimes f \otimes \psi} = \sem{\pi'}^{(1)}_{v'(G^\beta) \otimes \psi}$ and, for all $i >1$, $\sem{\pi}^{(i)} = \sem{\pi'}^{(i)}$, we trivially obtain pre-unitarity from the induction hypothesis.
				
				\item If $\pi = \begin{prooftree}
					\hypo{\pi'}
					\infer[no rule]1{P \sqcup (\widehat \Phi, \pi_L, F \stackrel{\alpha}{\multimap} G, G^\beta): (G^\beta, \widehat \Psi; \Gamma) \mid \Omega \uStile C}
					\infer1[\rArUnPro]{P\sqcup (\widehat \Phi, \pi_L, F \stackrel{\alpha}{\multimap} G, G^\beta):(\widehat \Phi, F \stackrel{\alpha}{\multimap} G, \widehat \Psi; \Gamma) \mid \Omega \uStile C}
				\end{prooftree}$, let $v$ a consistent valuation of $\pi$.\\
				Let $\varphi = v(\widehat{\Phi}), f = v(F \stackrel{\alpha}{\multimap} G)$, $\psi = v(\widehat{\Psi})$.\\
				By definition, $v(G^\beta) = f(\sem \pi_L (\varphi))$.
				
				We are now in the same situation as the above case, and pre-unitarity becomes trivial.
				
			\item If $\pi = \begin{prooftree}
				\hypo{\pi'}
				\infer[no rule]1{P:(F_x^{\oplus_x \alpha}, \widehat \Phi; \Gamma) \mid \Omega \uStile C}
				\infer1[\rWiUn{x}]{P:((F_1\with F_2)^\alpha, \widehat \Phi; \Gamma) \mid \Omega \uStile C}
			\end{prooftree}$, let $v$ a consistent valuation of $\pi$.
			
			By definition, $v$ is a consistent valuation for $\pi'$, and $\proj{F_1, F_2} x(v((F_1 \with F_2)^\alpha)) = v(F^{\oplus_x \alpha}_x)$.
			
			Since $\sem{\pi}^{(i)} = \begin{cases*}
				\sem{\pi'}^{(1)} \circ(\proj{F_1, F_2}x \circ \Id \Phi)&if $i = 1$\\
				\sem{\pi'}^{(i)}& otherwise
			\end{cases*}$, we trivially have pre-unitarity.
			
			\item If $\pi = \begin{prooftree}
				\hypo{\pi_1}
				\infer[no rule]1{\unlab(\widehat \Phi) \lStile F}
				
				\hypo{\pi_2}
				\infer[no rule]1{P:(F^\alpha, \widehat{\Psi_i}; \Gamma_i)_{i \in I} \mid \Omega \uStile C}
				\infer2[\rCutUnH]{P:(\widehat \Phi, \widehat{\Psi_i}; \Gamma_i)_{i \in I} \mid \Omega \uStile C}
			\end{prooftree}$ then let $v$ a valuation of $\pi$.\\
			Let $\varphi = v(\widehat \Phi)$, $\psi_i = v(\widehat{\Psi_i})$.
			Let $v'$ a valuation equal to $v$ except that $v'(F^\alpha) = \sem{\pi_L}(\varphi)$.
			
			Since $\sem{\pi}^{(i)} = \begin{cases*}
				\sem{\pi_2}^{(i)} \circ (\sem \pi_1 \otimes \Id {\Psi_i})& if $i \in I$\\
				\sem{\pi_2}^{(i)} & otherwise
			\end{cases*}$, we obtain pre-unitarity directly.
			\end{itemize}
	
	We have covered all rules of \logicsymb, and shown that they comform to our theorem.
	\end{proof}
	
	\subsection{Proof of Universality/Completeness for Unitaries}
	\universality*
	\begin{proof}
		\label{proofComp}
		We rely on the result from \cite{brennen2005}: Givens rotations are universal for one qudit unitaries.
		
		In the following, $\pi$ may denote the constant and not a proof. This should be obvious in context.
		
		Let $n \in \mathbb N$ with $n \geq 1$.\\
		Let $1 \leq j < k \leq n$, let $\theta, \phi$ angles.\\
		We denote $G_{j,k,n}(\theta, \phi) = \begin{pmatrix}
			I_{j-1} & 0 & 0 & 0 & 0\\
			0 & \cos \theta & 0 & -ie^{i\phi}\sin \theta & 0\\
			0 & 0 & I_{k-j-1} & 0 & 0\\
			0 & -ie^{-i\phi}\sin(\theta) &0 & \cos \theta & 0 \\
			0 & 0 & 0 & 0 & I_{n-k}
		\end{pmatrix}$ the Givens rotation over $j$ and $k$ with angles $\theta$ and $\phi$.
		
		We aim to construct $\pi$ a proof of $\underline n \uarrow \underline n$ such that \[\sem \pi = G_{1,2,n}(\theta, \phi) = \begin{pmatrix}
			\cos \theta & -ie^{i\phi} \sin \theta &\\
			-i e^{-i\phi} \sin \theta & \cos \theta &\\
			&& I_{n-2}
		\end{pmatrix}\]
		
		Let $\nu = \phi/2 - \pi/4$.
		
		We define $\mathfrak{g}_{1,2,n}(\theta, \phi)$ as:
		\[\begin{prooftree}
			\infer0[\rOneR]{(\emptyset) \uStile \unit}
			
			\infer0[\rOneR]{(\emptyset) \uStile \unit}
			
			\infer0[\rAx]{(\underline{n-2}) \uStile \underline{n-2} }
			\infer2[\rOpR]{(\emptyset) \mid (\underline{n-2}) \uStile \underline{n-1}}
			\infer2[\rOpR]{(\emptyset) \mid (\emptyset) \mid (\underline{n-2}) \uStile \underline n}
			\infer1[\rPh{\nu}(2)]{(\emptyset) \mid (\emptyset) \mid (\underline{n-2}) \uStile \underline n}
			\infer1[\rPh{-\nu}(1)]{(\emptyset) \mid (\emptyset) \mid (\underline{n-2}) \uStile \underline n}
			\infer1[\rOneL(1,2)]{(\emptyset) \mid (\emptyset) \mid (\underline{n-2}) \uStile \underline n}
			\infer1[\rOpL]{(\unit) \mid (\unit) \mid (\underline{n-2}) \uStile \underline n}
			\infer1[\rRotTh]{(\unit \oplus \unit) \mid (\underline{n-2})\uStile \underline n}
			\infer1[\rPh{-\nu}(2)]{(\emptyset) \mid (\emptyset) \mid (\underline{n-2}) \uStile \underline n }
			\infer1[\rPh{\nu}(1)]{(\emptyset) \mid (\emptyset) \mid (\underline{n-2}) \uStile \underline n}
			\infer1[\rOneL(1,2)]{(\unit) \mid \unit \mid (\underline{n-2}) \uStile \underline n}
			\infer1[\rOpL]{(\unit) \mid (\underline{n-1}) \uStile \underline n}
			\infer1[\rOpL]{(\underline{n}) \uStile \underline n}
			\infer1[\rUnitRight]{\lStile \underline n \uarrow \underline n}
		\end{prooftree} \]
		Where any number next to a rule name makes explicit which worlds it is applied to. For instance, \rPh{-\nu}(2) indicates that \rPh{-\nu} is applied to the second world.
		
		Rather trivially:
		\[\sem{\mathfrak{g}_{1,2,n}(\theta, \phi)} = \begin{pmatrix}
			e^{-i\nu} & 0 &\\
			0 & e^{i\nu} & \\
			& & I_{n-2}
		\end{pmatrix}
		\begin{pmatrix}
		\cos \theta & -\sin \theta &\\
		\sin \theta & \cos \theta &\\
		& & I_{n-2}
		\end{pmatrix}
		\begin{pmatrix}
			e^{i\nu} & 0 &\\
			0 & e^{-i\nu} &\\
			& & I_{n-2}
		\end{pmatrix}\]
		\[=\begin{pmatrix}
			e^{-i\nu} & 0 &\\
			0 & e^{i\nu} & \\
			& & I_{n-2}
		\end{pmatrix}
		\begin{pmatrix}
			e^{i\nu} \cos \theta & -e^{-i\nu} \sin \theta &\\
			e^{i\nu} \sin \theta & e^{-i\nu} \cos \theta &\\
			& & I_{n-2}
		\end{pmatrix}\]
		\[= \begin{pmatrix}
			e^0 \cos \theta & -e^{-2i\nu} \sin \theta &\\
			e^{2i\nu}\sin \theta & e^{0} \cos \theta&\\
			&&I_{n-2}
		\end{pmatrix}\]
		\[= \begin{pmatrix}
			\cos \theta & -e^{-i(\phi - \pi/2)} \sin \theta &\\
			e^{i(\phi -\pi/2)} \sin \theta & \cos \theta &\\
			&&I_{n-2}
		\end{pmatrix}\]
		\[=\begin{pmatrix}
			\cos \theta & -e^{i\pi/2}e^{-i\phi} \sin \theta &\\
			e^{-i\pi/2}e^{i\phi} \sin \theta & \cos \theta &\\
			&&I_{n-2}
		\end{pmatrix}\]
		\[=\begin{pmatrix}
			\cos \theta & -ie^{-i\phi} \sin \theta &\\
			-ie^{i\phi} \sin \theta & \cos \theta &\\
			&&I_{n-2}
		\end{pmatrix}\]
		\[= G_{1,2,n}(\theta, \phi)\]
		
		We then use the \rExchUnW\ rule to create a proof $\mathfrak{s}_{j,k,n}$ of $\underline n \uarrow \underline n$ such that $\sem {\mathfrak{s}_{j,k,n}}$ is the matrix that exchanges the $i$th and $j$th values of a vector of $\mathbb C^n$. It consists of unfolding $(\emptyset, \underline n)$ into $(\emptyset, \emptyset)^n$, and exchanging the $j$th and $k$th world.
		
		Finally, we can construct $\mathfrak{g}_{j,k,n}(\theta, \phi)$ as:
		\[\begin{prooftree}
			\hypo{\mathfrak{s}_{2,k,n}}
			\infer[no rule]1{(\emptyset; \underline n) \uStile \underline n}
			
			\hypo{\mathfrak{s}_{1,j,n}}
			\infer[no rule]1{(\emptyset; \underline n) \uStile \underline n}
			
			\hypo{\mathfrak{g}_{1,2,n}(\theta, \phi)}
			\infer[no rule]1{(\emptyset; \underline n) \uStile \underline n}
			
			\hypo{\mathfrak{s}_{2,k,n}}
			\infer[no rule]1{(\emptyset; \underline n) \uStile \underline n}
			
			\hypo{\mathfrak{s}_{1,j,n}}
			\infer[no rule]1{(\emptyset; \underline n) \uStile \underline n}
			
			\infer0[\rAx]{(\emptyset; \underline n) \uStile \underline n}
			\infer2[\rCutUnG]{(\emptyset; \underline n) \uStile \underline n}
			\infer2[\rCutUnG]{(\emptyset; \underline n) \uStile \underline n}
			\infer2[\rCutUnG]{(\emptyset; \underline n) \uStile \underline n}
			\infer2[\rCutUnG]{(\emptyset; \underline n) \uStile \underline n}
			\infer2[\rCutUnG]{(\emptyset; \underline n) \uStile \underline n}
			\infer1[\rUnitRight]{\lStile \underline n \uarrow \underline n}
		\end{prooftree}\]
		
		It is immediate to verify that $\sem{\mathfrak{g}_{j,k,n}(\theta, \phi)} = G_{j,k,n}(\theta, \phi)$.
		
		Since any unitary operation on a $n$ dimentional qdit can be expressed as the finite product of Givens matrices, we can use the \rCutUnG\ rule to compose $\mathfrak{g}_{j,k,n}(\theta, \phi)$ proofs, and obtain a proof with semantics equal to the original unitary.
		
		Therefore, for each unitary $f: \bC^n \rightarrow \bC^n$, there exists a proof $\pi$ of $\underline n \uarrow \underline n$ such that $\sem\pi = f$.
	\end{proof}

	\subsection{Cut Elimination}
		\label{sec:celim}
	\begin{definition}[Term Depth]
		We define recursively the depth of a term $T$, $|T|$ as follows:
		\begin{itemize}
			\item $|a| = |\unit| = 0$
			\item $|T_1 \odot T_2| = \max(|T_1|, |T_2|) +1$ with $\odot$ any connector
		\end{itemize}
	\end{definition}
	
	\begin{definition}[Notation]
		Let $\pi$ a proof.\\
		We will write $l(\pi)$ to denote the root rule of $\pi$.\\
		Furthermore, if $l(\pi)$ is a unary rule, we will by convention write $\pi'$ its one direct subproof. If $l(\pi)$ is binary, we will write $\pi_1$ its left subproof, and $\pi_2$ its right subproof.
	\end{definition}
	
	\begin{definition}[Cut Rank]
		Let $\pi$ be a proof. We define $c(\pi)$ the \emph{cut rank} of $\pi$ inductively:
		\begin{itemize}
			\item If $l(\pi)$ is an axiom rule, $c(\pi) = 0$
			\item If $l(\pi)$ is a unary rule, $c(\pi) = c(\pi')$
			\item If $l(\pi)$ is a non-cut, binary rule, $c(\pi) = \max(c(\pi_1), c(\pi_2))$
			\item If $l(\pi)$ is a cut rule with output $T$, $c(\pi) = \max(|T|+1, c(\pi_1), c(\pi_2))$
		\end{itemize}
		Note that promise proofs are not included in the cut rank calculation if they are not employed as subproofs.
	\end{definition}
	
	\begin{definition}[Proof Depth]
		We define $d(\pi)$ the depth of a proof $\pi$ inductively as:
		\begin{itemize}
			\item If $\pi$ is an axiom, then $d(\pi) = 0$
			\item If $l(\pi)$ is an unary rule, then $d(\pi) = d(\pi') + 1$
			\item If $l(\pi)$ is a binary rule, then $d(\pi) = d(\pi_1) + d(\pi_2) +1$
		\end{itemize}
	\end{definition}
	
	We are now ready to state a small series of lemmas, one for each cut rule, that will allow us to reduce the cut rank of a proof.
	\begin{lemma}
		\label{lem:cutred}
		If $\pi_1$ is a cut-free proof of $\Phi \lStile F$, and $\pi_2$ is a cut-free proof of $F, \Psi \lStile H$, then, there exists $\pi$ a proof of $\Phi, \Psi \lStile H$ such that $c(\pi) \leq |F|$ and:
		\[ \pi \equiv \begin{prooftree}
			\hypo {\pi_1}
			\infer[no rule]1{\Phi \lStile F}
			
			\hypo {\pi_2}
			\infer[no rule]1{F, \Psi \lStile H}
			
			\infer2[\rCutLin]{\Phi, \Psi \lStile H}
		\end{prooftree}\]

		If $\widehat \Phi$ is an ordered multiset of labeled higher terms, $\pi_1$ is a cut-free proof of $\unlab(\widehat \Phi) \lStile F$, and $\pi_2$ is a proof of $P: (F^\alpha, \widehat{\Psi_i}; \Gamma_i)_{i \in I} \mid \Omega \uStile C$ with $c(\pi_2) = 0$, such that $\alpha$ is fresh in $P:(\widehat \Phi, \widehat{\Psi_i}; \Gamma_i)_{i \in I} \mid \Omega$ then, there exists $\pi$ a proof of $P:(\widehat \Phi, \widehat{\Psi_i}; \Gamma_i)_{i \in I} \mid \Omega \uStile C$ such that $c(\pi) \leq |F|$ and:
		\[\pi \equiv \begin{prooftree}
			\hypo{\pi_1}
			\infer[no rule]1{\unlab(\widehat \Phi) \lStile F}
			
			\hypo{\pi_2}
			\infer[no rule]1{P:(F^\alpha, \widehat{\Psi_i}; \Gamma_i)_{i \in I} \mid \Omega \uStile C}
			
			\infer2[\rCutUnH]{P:(\widehat \Phi, \widehat{\Psi_i}; \Gamma_i)_{i \in I} \mid \Omega \uStile C}
		\end{prooftree}\]
		
		Let $\pi_1$ be a proof of $P: (\widehat{\Phi_i}; \Gamma_i)_{i \in I} \uStile A$, and $\pi_2$ be a proof of $P: (\widehat \Psi; A, \Delta) \mid \Omega \uStile C$ such that $c(\pi_1) = c(\pi_2) = 0$.\\
		There exists $\pi$ a proof of $(\widehat{\Phi_i}, \widehat \Psi; \Gamma_i, \Delta)_{i \in I} \mid \Omega \uStile C$ such that $c(\pi) \leq |A|$ and:
		\[\pi \equiv \begin{prooftree}
			\hypo{\pi_1}
			\infer[no rule]1{P:(\widehat{\Phi_i}; \Gamma_i)_{i \in I} \uStile A}
			
			\hypo{\pi_2}
			\infer[no rule]1{P:(\widehat \Psi; A, \Delta) \mid \Omega \uStile C}
			
			\infer2[\rCutUnG]{P:(\widehat{\Phi_i}, \widehat \Psi; \Gamma_i, \Delta)_{i \in I} \mid \Omega \uStile C}
		\end{prooftree}\]
	\end{lemma}
	
	\begin{proof}
		We proceed by induction over $d(\pi_1) + d(\pi_2)$.
		
		\proofsubparagraph{Initial Case} When $d(\pi_1) + d(\pi_2) = 0$, we only have one case to consider:
		\[\pi_1 = \pi_2 = \begin{prooftree}
			\infer0[\rAx]{P: (\emptyset; A) \lStile A}
		\end{prooftree}\]
		Here, we trivially have $\pi = \pi_1 = \pi_2 = \begin{prooftree}
			\infer0[\rAx]{P: (\emptyset; A) \lStile A}
		\end{prooftree}$
		
		We then split the induction over the three possible cases:
		
		\proofsubparagraph{Induction, \rCutLin\  Case}
		
		We exhibit the following commutations and reductions. The preservation of semantics is trivial.
		\[\begin{prooftree}
			\hypo{\Phi, G, F, \Psi \lStile K}
			\infer1[\rExchLin]{\Phi, F, G, \Psi \lStile K}
			
			\hypo{K, \Theta \lStile H}
			\infer2[\rCutLin]{\Phi, F, G, \Psi, \Theta \lStile H}
		\end{prooftree} \rightarrow
		\begin{prooftree}
			\hypo{\Phi, G, F, \Psi \lStile K}
			
			\hypo{K, \Theta, \lStile H}
			
			\infer2[\rCutLin]{\Phi, G, F, \Psi, \Theta \lStile H}
			\infer1[\rExchLin]{\Phi, F, G, \Psi, \Theta \lStile H}
		\end{prooftree} \]
		
		\[\begin{prooftree}
			\hypo{\Phi \lStile K}
			
			\hypo{K, \Psi \lStile F}
			\hypo{\Theta \lStile G}
			
			\infer2[\rOtLinR]{K, \Psi, \Theta \lStile F \otimes G}
			\infer2[\rCutLin]{\Phi, \Psi, \Theta \lStile F \otimes G}
		\end{prooftree} \rightarrow
		\begin{prooftree}
			\hypo{\Phi \lStile K}
			\hypo{K, \Psi \lStile F}
			\infer2[\rCutLin]{\Phi, \Psi \lStile K}
			
			\hypo{\Theta \lStile G}
			
			\infer2[\rOtLinR]{\Phi, \Psi, \Theta \lStile F \otimes G}
		\end{prooftree}\]
		
		\[\begin{prooftree}
			\hypo{\Phi \lStile K}
			\hypo{F, K, \Psi \lStile G}
			\infer1[\rArLinR]{K, \Psi \lStile F \multimap G}
			\infer2[\rCutLin]{\Phi, \Psi \lStile F \multimap G}
		\end{prooftree} \rightarrow
		\begin{prooftree}
			\hypo{\Phi \lStile K}
			\hypo{F, K, \Psi \lStile G}
			\infer2[\rCutLin]{F, \Phi, \Psi \lStile G}
			\infer1[\rArLinR]{\Phi, \Psi \lStile F \multimap G}
		\end{prooftree}\]
		
		\begin{small}
		\[\begin{array}{l}
		\begin{prooftree}
			\hypo{\Phi \lStile K}
			
			\hypo{K, \Psi \lStile F}
			\hypo{K, \Psi \lStile G}
			\infer2[\rWiLinR]{K, \Psi \lStile F \with G}
			
			\infer2[\rCutLin]{\Phi, \Psi \lStile F \with G}
		\end{prooftree}\\\hspace{4cm}\longrightarrow
		\begin{prooftree}
			\hypo{\Phi \lStile K}
			
			\hypo{K, \Psi \lStile F}
			\infer2[\rCutLin]{\Phi, \Psi \lStile F}
			
			\hypo{\Phi \lStile K}
			
			\hypo{K, \Psi \lStile F}
			\infer2[\rCutLin]{\Phi, \Psi \lStile G}
			
			\infer2[\rWiLinR]{\Phi, \Psi \lStile F \with G}
		\end{prooftree}\end{array}\]
		\end{small}
		
		\begin{small}
		\[\begin{array}{l}
			\begin{prooftree}
			\hypo{\Phi \lStile K}
			
			\hypo{K, \Psi \labrel K^\alpha, \widehat \Psi}
			\hypo{\emptyset:(K^\alpha, \widehat \Psi; A) \uStile B}
			\infer2[\rUnitRight]{K, \Psi \lStile A \uarrow B}
			
			\infer2[\rCutLin]{\Phi, \Psi \lStile A \uarrow B}
		\end{prooftree} \\ \hspace{5.3cm} \longrightarrow
		\begin{prooftree}
			\hypo{\Phi, \Psi \labrel \widehat \Phi, \widehat \Psi}
			
			\hypo{\Phi \lStile K}
			\hypo{\emptyset:(K^\alpha, \widehat \Phi; A) \uStile b}
			\infer2[\rCutUnH]{\emptyset:(\widehat \Phi, \widehat \Psi; A) \uStile B}
			
			\infer2[\rUnitRight]{\Phi, \Psi \lStile A \uarrow B}	
		\end{prooftree}\end{array}
		 \]
		 \end{small}
		 With $\alpha$ not in the labels of $\widehat \Phi$.
		 
		 \begin{small}
		 \[
		 \begin{array}{l}
		 \begin{prooftree}
		 	\hypo{\Phi_1 \lStile K_1}
		 	\hypo{\Phi_2 \lStile K_2}
		 	\infer2[\rOtLinR]{\Phi_1, \Phi_2 \lStile K_1 \otimes K_2}
		 	
		 	\hypo{K_1, K_2, \Psi \lStile H}
		 	\infer1[\rOtLinL]{K_1 \otimes K_2, \Psi \lStile H}
		 	
		 	\infer2[\rCutLin]{\Phi_1, \Phi_2, \Psi \lStile H}
		 \end{prooftree}\\ \\ \hspace{5.8cm}\longrightarrow 
		 \begin{prooftree}
		 	\hypo{\Phi_1 \lStile K_1}
		 	
		 	\hypo{\Phi_2 \lStile K_2}
		 	
		 	\hypo{K_1, K_2, \Psi \lStile H}
		 	\infer2[\rCutLin]{K_1, \Phi_2, \Psi \lStile H}
		 	
		 	\infer2[\rCutLin]{\Phi_1, \Phi_2, \Psi \lStile H}
		 \end{prooftree}\end{array} \]\end{small}

		 \begin{small}
		 \[\begin{prooftree}
		 	\hypo{F, \Psi \lStile G}
		 	\infer1[\rArLinR]{\Psi \lStile F \multimap G}
		 	
		 	\hypo{\Phi \lStile F}
		 	\hypo{G, \Theta \lStile H}
		 	
		 	\infer2[\rArLinL]{\Phi, F \multimap G, \Theta \lStile H}
		 	
		 	\infer2[\rCutLin]{\Phi, \Psi, \Theta \lStile H}
		 \end{prooftree} \rightarrow
		 \begin{prooftree}
		 	\hypo{\Phi \lStile F}
		 	
		 	\hypo{F, \Psi \lStile G}
		 	\hypo{G, \Theta \lStile H}
		 	\infer2[\rCutLin]{F, \Psi, \Theta \lStile H}
		 	
		 	\infer2[\rCutLin]{\Phi, \Psi, \Theta \lStile H}
		 \end{prooftree} \]
		 \end{small}
		 
		 \[\begin{prooftree}
		 	\hypo{\Phi \lStile K_1}
		 	\hypo{\Phi \lStile K_2}
		 	\infer2[\rWiLinR]{\Phi \lStile K_1 \with K_2}
		 	
		 	\hypo{K_i, \Psi \lStile H}
		 	\infer1[\rWiLinL i]{K_1 \with K_2, \psi \lStile H}
		 	\infer2[\rCutLin]{\Phi, \Psi \lStile H}
		 \end{prooftree} \rightarrow
		 \begin{prooftree}
		 	\hypo{\Phi \lStile K_i}
		 	\hypo{K_i, \Psi \lStile H}
		 	\infer2[\rCutLin]{\Phi, \Psi \lStile H}
		 \end{prooftree} \]
		 
		 \proofsubparagraph{Induction, \rCutUnH\ Case}
		 Due to the multi-world nature of \rCutUnH, the output of a cut rule may be destroyed in multiple worlds at different times, thus requiring the graphical representation defined in~\cref{def:graphical}, and the notion of simple proofs, to keep the bureaucracy manageable.
		 
		 In this section, we will assume that proofs are simple.
		 
		 Let $\pi_1$ a proof of $\unlab(\widehat \Phi) \lStile K$.

		 \begin{itemize}
		 	\item If $\widehat \Phi = \widehat {\Phi_1}, F^\alpha, G^\beta, \widehat{\Phi_2}$ and
		 	$\pi_1 = \begin{prooftree}
		 		\hypo{\pi_1'}
		 		\infer[no rule]1{\Phi_1, G, F, \Phi_2 \lStile K}
		 		\infer1[\rExchLin]{\Phi_1, F, G, \Phi_2 \lStile K}
		 	\end{prooftree}$, let us notice this proof equivalence:
		 	
		 	\[\begin{prooftree}
		 		\hypo{\pi_1'}
		 		\infer[no rule]1{\Phi_1, G, F, \Phi_2 \lStile K}
		 		\infer1[\rExchLin]{\Phi_1, F, G, \Phi_2 \lStile K}
		 		
		 		\hypo{\pi_2}
		 		\infer[no rule]1{P:(K^\kappa, \widehat{\Psi_i};\Gamma_i)_{1 \leq i \leq n} \mid \Omega \uStile C}
		 		
		 		\infer2[\rCutUnH]{P:(\widehat {\Phi_1}, F^\alpha, G^\beta, \widehat{\Phi_2}, \widehat{\Psi_i}; \Gamma_i)_{1\leq i \leq n} \mid \Omega \uStile C}
		 	\end{prooftree}\]
		 	\[\equiv \begin{prooftree}
		 		\hypo{\pi_1'}
		 		\infer[no rule]1{\Phi_1, G, F, \Phi_2 \lStile K}
		 		
		 		\hypo{\pi_2}
		 		\infer[no rule]1{P:(K^\kappa, \widehat{\Psi_i};\Gamma_i)_{1 \leq i \leq n} \mid \Omega \uStile C}
		 		
		 		\infer2[\rCutUnH]{P:(\widehat {\Phi_1}, G^\beta, F^\alpha, \widehat{\Phi_2}, \widehat{\Psi_i}; \Gamma_i)_{1\leq i \leq n} \mid \Omega \uStile C}
		 		
		 		\infer[double]1[\rExchUnH]{P:(\widehat {\Phi_1}, F^\alpha, G^\beta, \widehat{\Phi_2}, \widehat{\Psi_i}; \Gamma_i)_{1\leq i \leq n} \mid \Omega \uStile C}
		 	\end{prooftree}\]
		 	
		 	From this, we can trivially apply the induction hypothesis and obtain $\pi$.
		 	
		 	\item The cases where $l(\pi_1)$ is \rOtLinL\ or \rWiLinL{i}\ proceed identically.\\
		 		When $l(\pi_1)$ is \rArLinL, we create a new label fresh in $\pi_2$ for the promise we will have to generate, and proceed identically.
		 		
		 	\item We now work graphically. Let us start with this case:
		 		\[\input{cutred-prf-sigma-1.tikz}\]
		 		
		 		Note we can't assume this \rExchUnH\ rule is necessarily $l(\pi_2)$, since $l(\pi_2)$ might be a rule that does not commute with \rCutUnH.
		 		
		 		By~\cref{lem:mono}, we have:
		 		\begin{align*}
		 			\input{cutred-prf-sigma-1.tikz}
		 			~~\equiv~~\input{cutred-prf-sigma-2.tikz}
		 			~~\equiv~~\input{cutred-prf-sigma-3.tikz}
		 		\end{align*}
		 		
		 		And we can apply the induction hypothesis.
		 	
		 	\item The cases where this rule is replaced by \rOtUn, \rWiUn, \rArUnNew, \rBtL\ or \rOpL\, assuming $K^\kappa$ is not the principal formula, behave identically.\\
		 		The \rArUnPro\ case is similarly trivial, we simply may have to replace this rule with an \rArUnNew\ rule if it requires commuting with the rule that generated its promise.
		 		
		 	\item The \rUnitUn\ case is more difficult, since it combines multiple worlds into one. If we are in this case
		 	\[\input{cutred-prf-rUnitUn-1.tikz}\]
		 	
		 	Since $\pi_2$ is simple, $K^\kappa$ has to be present among all the input worlds:
		 	\[\input{cutred-prf-rUnitUn-2.tikz}\]
		 	
		 	Therefore, only two cases are possible:
		 	\[\input{cutred-prf-rUnitUn-3.tikz}\quad\text{ or }\quad\input{cutred-prf-rUnitUn-4.tikz}\]
		 	With $r$ a rule whose principal formula is not $K^\kappa$.
		 	
		 	In the former case, consider the following equivalence:

		 	\scalebox{0.8}{\input{cutred-prf-rUnitUn-3.tikz}}
		 		~~$\equiv$~~\scalebox{0.8}{\input{cutred-prf-rUnitUn-3f.tikz}}
		 		~~$\equiv~$~\scalebox{0.8}{\input{cutred-prf-rUnitUn-3ff.tikz}}
		 	
		 	In the latter, we can trivially commute the rule $r$ with the cut as well.
		 	
		 	In both cases we can apply the induction hypothesis to find $\pi$.
		 	
		 	\item The \rRotTh\ case is identical.
		 	
		 	\item The \rExchUnW\ case is trivial if we work modulo permutation.

		 	\item With the prior cases, we can now assume that the first rule applied to each world with $K^\kappa$ is a rule that has it as a principal formula.
		 	
		 		\begin{itemize}
		 			
		 			\item Let us assume:
		 			\[\pi_2 = \begin{prooftree}
		 				\hypo{\pi_L}
		 				\infer[no rule]1{K, \unlab(\widehat{\Phi_L}) \lStile F}
		 				
		 				\hypo{\pi_2'}
		 				\infer[no rule]1{P \sqcup p_L: (G^\beta, \widehat{\Psi_i}; \Gamma_i)_{1 \leq i \leq n} \mid \Omega_K \mid \Omega \uStile C}
		 				
		 				\infer[double]2[$\multimap^*$]{P: (K^\kappa, \widehat{\Phi_L}, F \stackrel{\alpha}{\multimap} G, \widehat{\Psi_i}; \Gamma_i)_{1 \leq i \leq n} \mid \Omega_K \mid \Omega \uStile C}
		 			\end{prooftree} \]
						
						Where $p_L = ((K^\kappa, \widehat{\Phi_L}), \pi_L, F \stackrel{\alpha}{\multimap} G, G^\beta)$, $\Omega_K$ is a (potentially empty) set of worlds that each contain $K^\kappa$, $\Omega$ a set of worlds that does not contain $K^\kappa$, and $\multimap^*$ represents a \rArUnNew\ rule followed by \rArUnPro\ rules on the same promise.\\
						We also assume that $c(\pi_2) = 0$, so $c(\pi_L) = c(\pi_2') = 0$ and that no \rArUnPro\ rule in $\pi_2'$ makes use of $p_L$.
						
						By induction hypothesis, there exist $\pi_L'$ with $c(\pi_L') \leq |K|$ and \[\pi_L' \equiv \begin{prooftree}
							\hypo{\pi_1}
							\infer[no rule]1{\unlab(\widehat \Phi) \lStile K}
							
							\hypo{\pi_L}
							\infer[no rule]1{K, \unlab(\widehat {\Phi_L}) \lStile F}
							
							\infer2[\rCutLin]{\unlab(\widehat \Phi), \unlab(\widehat {\Phi_L}) \lStile F}
						\end{prooftree}\]
						
						Let $p_L' = ((\widehat \Phi, \widehat{\Phi_L}), \pi_L', F \stackrel{\alpha}{\multimap} G, G^\beta)$.
						
						Since $p_L$ is unused in $\pi_2'$, we obtain $\pi_2''$ by replacing $p_L$ with $p_L'$, and it remains a proof, is equivalent, and still has $c(\pi_2'')= 0$.
						
						If $\Omega_K$ is empty, we propose the following $\pi$:
						\[\pi = \begin{prooftree}
							\hypo{\pi_L'}
							\infer[no rule]1{\unlab(\widehat \Phi), \unlab(\widehat{\Phi_L}) \lStile F}
							
							\hypo{\pi_2''}
							\infer[no rule]1{P \sqcup p_L': (G^\beta, \widehat{\Psi_i}; \Gamma_i)_{1 \leq i \leq n} \mid \Omega \uStile C}
							
							\infer[double]2[$\multimap^*$]{P: (\widehat \Phi, \widehat{\Phi_L}, F \stackrel{\alpha}{\multimap} G, \widehat{\Psi_i}; \Gamma_i)_{1 \leq i \leq n} \mid \Omega \uStile C}
						\end{prooftree} \]
						
						If $\Omega_K$ is non-empty, we apply the induction hypothesis to $\pi_1$ and $\pi_2''$ to obtain $\pi'$ a proof of $P\sqcup p_L': (G^\beta, \widehat{\Psi_i}; \Gamma_i)_{1 \leq i \leq n} \mid \Omega_{\widehat \Phi} \mid \Omega \uStile C$, where $\Omega_{\widehat \Phi}$ is equal to $\Omega_K$ where all instances of $K^\kappa$ have been replaced by $\widehat \Phi$.\\
						We have $c(\pi') \leq |K|$, and:
						\[\pi' \equiv \begin{prooftree}
							\hypo{\pi_1}
							\infer[no rule]1{\unlab(\widehat \Phi) \lStile K}
							
							\hypo{\pi_2''}
							\infer[no rule]1{P \sqcup p_L': (G^\beta, \widehat{\Psi_i}; \Gamma_i)_{1 \leq i \leq n} \mid \Omega_K \mid \Omega \uStile C}
							
							\infer2[\rCutUnH]{P \sqcup p_L': (G^\beta, \widehat{\Psi_i}; \Gamma_i)_{1 \leq i \leq n}\mid \Omega_{\widehat \Phi} \mid \Omega \uStile C}
						\end{prooftree}\]
						
						We then propose this definition for $\pi$:
						\[\pi = \begin{prooftree}
							\hypo{\pi_L'}
							\infer[no rule]1{\unlab(\widehat \Phi), \unlab(\widehat{\Phi_L}) \lStile F}
							
							\hypo{\pi'}
							\infer[no rule]1{P\sqcup p_L': (G^\beta, \widehat{\Psi_i}; \Gamma_i)_{1 \leq i \leq n} \mid \Omega_{\widehat \Phi} \mid \Omega \uStile C}
							
							\infer[double]2[$\multimap^*$]{P: (\widehat \Phi, \widehat{\Phi_L}, F \stackrel{\alpha}{\multimap} G, \widehat{\Psi_i}; \Gamma_i)_{1 \leq i \leq n} \mid \Omega_{\widehat \Phi} \mid \Omega \uStile C}
						\end{prooftree} \]
						
					\item Let us assume $K = A \uarrow B$, and $\pi_1 = \begin{prooftree}
						\hypo{\tilde \Phi \labrel \unlab(\widehat \Phi)}
						
						\hypo{\pi_1'}
						\infer[no rule]1{\emptyset: (\tilde \Phi; A) \uStile B}
						
						\infer2[\rUnitRight]{\unlab(\widehat \Phi) \lStile A \uarrow B}
					\end{prooftree}$
					
					We also assume $\pi_2 = \begin{prooftree}
						\hypo{\pi_{21}}
						\infer[no rule]1{P:(\emptyset, \Gamma_i)_{i \leq n} \uStile A}
						
						\hypo{\pi_{22}}
						\infer[no rule]1{P:(\widehat{\Psi_1}; B, \Delta_1) \mid \Omega_K \mid \Omega \uStile C}
						
						\infer2[\rUnitUn]{P:(A \stackrel{\kappa}{\uarrow} B, \widehat{\Psi_1}; \Gamma_i, \Delta_1)_{1 \leq i \leq n} \mid \Omega_K \mid \Omega \uStile C}
					\end{prooftree}$
					
						Trivially, there exists $\pi_1^*$ a proof of $P: (\widehat \Phi; A) \uStile B$ such that it is equivalent to $\pi_1'$.
					
						We also use the induction hypothesis on $\pi_{22}$ to create $\pi'$ a proof of $P:(\widehat{\Psi_1}; B, \Delta_1) \mid \Omega_{\widehat \Phi} \mid \Omega \uStile C$ such that
					\[\pi' \equiv \begin{prooftree}
						\hypo{\pi_1}
						\infer[no rule]1{\unlab(\widehat \Phi) \lStile A \uarrow B}
						
						\hypo{\pi_{22}}
						\infer[no rule]1{P:(\widehat{\Psi_1}; B, \Delta_1) \mid \Omega_{K} \mid \Omega \uStile C}
						
						\infer2[\rCutUnH]{P:(\widehat{\Psi_1}; B, \Delta_1) \mid \Omega_{\widehat \Phi} \mid \Omega \uStile C}
					\end{prooftree}\]
					
						And the following proof suffices:
					\[\begin{prooftree}
						\hypo{\pi_{21}}
						\infer[no rule]1{P:(\emptyset; \Gamma_i)_i \uStile A}
						
						\hypo{\pi_1^*}
						\infer[no rule]1{P:(\widehat \Phi; A) \uStile B}
						
						\infer2[\rCutUnG]{P:(\widehat \Phi; \Gamma_i)_i \uStile B}
						
						\hypo{\pi'}
						\infer[no rule]1{P:(\widehat{\Psi_1}; B, \Delta_1) \mid \Omega_{\widehat \Phi} \mid \Omega \uStile C}
						
						\infer2[\rCutUnG]{P:(\widehat{\Phi}, \widehat{\Psi_1}; \Gamma_i, \Delta_1)_i \mid \Omega_K \mid \Omega \uStile C}
					\end{prooftree} \]
					
						The equivalence with the cut version is trivial with some basic computations.
						
					\item We assume $K = F \multimap G$ and $\pi_1 = \begin{prooftree}
						\hypo{\pi_1'}
						\infer[no rule]1{F, \unlab(\widehat \Phi) \lStile G}
						\infer1[\rArLinR]{\unlab(\widehat \Phi) \lStile G}
					\end{prooftree}$ and
					\[\pi_2 = \begin{prooftree}
						\hypo{\pi_L}
						\infer[no rule]1{\unlab(\widehat{\Phi_L}) \lStile F}
						
						\hypo{\pi_2'}
						\infer[no rule]1{P\sqcup p_L : (G^\alpha, \widehat{\Psi_i}; \Gamma_i)_{1 \leq i \leq n} \mid \Omega_K \mid \Omega \uStile C}
						
						\infer[double]2[$\multimap^*$]{P: (\widehat{\Phi_L}, F \stackrel{\kappa}{\multimap} G, \widehat{\Psi_i}; \Gamma_i)_{1 \leq i \leq n} \mid \Omega_k \mid \Omega \uStile C}
					\end{prooftree} \]
						With $p_L = (\widehat{\Phi_L}, \pi_L, F \stackrel{\kappa}{\multimap} G, G^\alpha)$, $c(\pi_2) = 0$, and $p_L$ unused in $\pi_2'$.
						
						Since $p_L$ is unused in $\pi_2'$, we can remove it from the set of promises to create the equivalent proof $\pi_2''$. We then apply the induction hypothesis to $\pi_1$ and $\pi_2''$ to obtain $\pi'$ a proof of $P: (G^\alpha, \widehat{\Psi_i}; \Gamma_i)_{1 \leq i \leq n} \mid \Omega_{\widehat \Phi} \mid \Omega \uStile C$ such that $c(\pi') \leq |F \multimap G|$ and
						\[\pi' \equiv \begin{prooftree}
							\hypo{\pi_1}
							\infer[no rule]1{\unlab(\widehat \Phi) \lStile F \multimap G}
							\hypo{\pi_2''}
							\infer[no rule]1{P:(G^\alpha, \widehat{\Psi_i}; \Gamma_i)_{1 \leq i \leq n} \mid \Omega_K \mid \Omega \uStile C}
							
							\infer2[\rCutUnH]{P: (G^\alpha, \widehat{\Psi_i}; \Gamma_i)_{1 \leq i \leq n} \mid \Omega_{\widehat \Phi} \mid \Omega \uStile C}
						\end{prooftree} \]
						
						We then propose this definition for $\pi$:
						\[\pi = \begin{prooftree}
							\hypo{\pi_L}
							\infer[no rule]1{\unlab(\widehat{\Phi_L}) \lStile F}
							
							\hypo{\pi_1'}
							\infer[no rule]1{F, \unlab(\widehat \Phi) \lStile G}
							
							\infer2[\rCutLin]{\unlab(\widehat \Phi_L), \unlab(\widehat \Phi) \lStile G}
							
							\hypo{\pi'}
							\infer[no rule]1{{P: (G^\alpha, \widehat{\Psi_i}; \Gamma_i)_{1 \leq i \leq n} \mid \Omega_{\widehat \Phi} \mid \Omega \uStile C}}
							
							\infer2[\rCutUnH]{P:(\widehat {\Phi_L}, \widehat{\Phi}, \widehat{\Psi_i}; \Gamma_i)_{1 \leq i \leq j} \mid \Omega_{\widehat \Phi} \mid \Omega \uStile C}
						\end{prooftree} \]
						
						The equivalence is immediate by computation.
						
					\item Let us assume $K = F \otimes G$, $\widehat{\Phi} = \widehat{\Phi_1}, \widehat{\Phi_2}$, as well as:
					\[\pi_1 = \begin{prooftree}
						\hypo{\pi_{11}}
						\infer[no rule]1{\unlab(\widehat{\Phi_1}) \lStile F}
						
						\hypo{\pi_{12}}
						\infer[no rule]1{\unlab(\widehat{\Phi_2}) \lStile G}
						\infer2[\rOtLinL]{\unlab(\widehat{\Phi_1}, \widehat{\Phi_2}) \lStile F \otimes G}
					\end{prooftree}\]
					
					\[\pi_2 = \begin{prooftree}
						\hypo{\pi_2'}
						\infer[no rule]1{P:(F^{\otimes_1 \kappa}, G^{\otimes_2\kappa}, \widehat{\Psi_i}; \Gamma_i)_i \mid \Omega \uStile C}
						\infer[double]1[\rOtUn]{P:((F\otimes G)^\kappa, \widehat{\Psi_i}; \Gamma_i)_i \mid \Omega \uStile C}
					\end{prooftree} \]
					With $\kappa$ not found in $P$ or $\Omega$.\\
					Let $\alpha, \beta$ labels not found in $\pi_2$.
					
					The following proof suffices:
					\[\pi = \begin{prooftree}
						\hypo{\pi_{11}}
						\infer[no rule]1{\unlab(\widehat{\Phi_1}) \lStile F}
						
						\hypo{\pi_{21}}
						\infer[no rule]1{\unlab(\widehat{\Phi_2}) \lStile G}
						
						\hypo{\pi_2''}
						\infer[no rule]1{P:(F^\alpha, G^\beta, \widehat{\Psi_i}; \Gamma_i)_i \mid \Omega \uStile C}
						
						\infer2[\rCutUnH]{P:(F^\alpha, \widehat{\Phi_2}, \widehat{\Psi_i}; \Gamma_i)_i \mid \Omega \uStile C}
						
						\infer2[\rCutUnH]{P:(\widehat{\Phi_1}, \widehat{\Phi_2}, \widehat{\Psi_i}; \Gamma_i)_i \mid \Omega \uStile C}
					\end{prooftree}\]
					
					Where $\pi_2''$ is $\pi_2'$ where all instances of $\otimes_1 \kappa$ are replaced with $\alpha$, and all instances of $\otimes_2 \kappa$ are replaced with $\beta$.
					
					The semantic equivalence is trivial.
					
					\item
						Let us assume $K = F \with G$, and $\pi_1 = \begin{prooftree}
							\hypo{\pi_{11}}
							\infer[no rule]1{\unlab(\widehat \Phi) \lStile F}
							
							\hypo{\pi_{12}}
							\infer[no rule]1{\unlab(\widehat \Phi) \lStile G}
							
							\infer2[\rWiLinR]{\unlab(\widehat \Phi) \lStile F\with G}
						\end{prooftree}$ as well as
						\[\pi_2 = \begin{prooftree}
							\hypo{\pi_2'}
							\infer[no rule]1{P:(F^{\oplus_1 \kappa}, \widehat {\Psi_{i}}; \Gamma_i)_{1 \leq i \leq n} \mid (G^{\oplus_2 \kappa}, \widehat{\Psi_j}; \Gamma_j)_{n < j \leq m} \mid \Omega \uStile C}
							\infer[double]1[\rWiUn{2}]{P:(F^{\oplus_1 \kappa}, \widehat {\Psi_{i}}; \Gamma_i)_{1 \leq i \leq n} \mid ((F \with G)^{\kappa}, \widehat{\Psi_j}; \Gamma_j)_{n < j \leq m} \mid \Omega \uStile C}
							\infer[double]1[\rWiUn 1]{P:((F \with G)^\kappa, \widehat {\Psi_{i}}; \Gamma_i)_{1 \leq i \leq n} \mid ((F \with G)^{\kappa}, \widehat{\Psi_j}; \Gamma_j)_{n < j \leq m} \mid \Omega \uStile C}
						\end{prooftree} \]
						
						Let $\alpha, \beta$ two labels that do not appear in $\pi_2$.\\
						Let $\pi_2''$ be equal to $\pi_2$ were each instance of $\oplus_1 \kappa$ is replaced by $\alpha$ and each instance of $\oplus_2 \kappa$ is replaced by $\beta$.
						
						The following proof suffices:
						\[\begin{prooftree}
							\hypo{\pi_{11}}
							\infer[no rule]1{\unlab(\widehat \Phi) \lStile F}
							
							\hypo{\pi_{12}}
							\infer[no rule]1{\unlab(\widehat \Phi) \lStile G}
							
							\hypo{\pi_2''}
							\infer[no rule]1{P:(F^{\alpha}, \widehat {\Psi_{i}}; \Gamma_i)_{1 \leq i \leq n} \mid (G^{\beta}, \widehat{\Psi_j}; \Gamma_j)_{n < j \leq m} \mid \Omega \uStile C}
							
							\infer2[\rCutUnH]{P:(F^{\alpha}, \widehat {\Psi_{i}}; \Gamma_i)_{1 \leq i \leq n} \mid (\widehat \Phi, \widehat{\Psi_j}; \Gamma_j)_{n < j \leq m} \mid \Omega \uStile C}
							
							\infer2[\rCutUnH]{P:(\widehat \Phi, \widehat {\Psi_{i}}; \Gamma_i)_{1 \leq i \leq n} \mid (\widehat \Phi, \widehat{\Psi_j}; \Gamma_j)_{n < j \leq m} \mid \Omega \uStile C}
						\end{prooftree} \]
		 				
		 				The equivalence is immediate.
		 		\end{itemize}
		 		
		 \end{itemize}
		 
		 \proofsubparagraph{\rCutUnG\ Case}
		 Let $\pi_1$ a proof of $P:(\widehat{\Phi_i}; \Gamma_i)_{1 \leq i \leq n} \uStile A$ with $c(\pi_1) = 0$.
		 
		 Let $\pi_2$ a proof of $P: (\widehat \Phi; A, \Delta) \mid \Omega \uStile C$ with $c(\pi_2) = 0$.
		 
		 We want to exhibit $\pi$ a proof of $P: (\widehat {\Phi_i}, \widehat \Psi; \Gamma_i, \Delta)_{1 \leq i \leq n} \mid \Omega \uStile C$ such that $c(\pi) \leq |A|$ and
		 \[\pi \equiv \begin{prooftree}
		 	\hypo{\pi_1}
		 	\infer[no rule]1{P:(\widehat{\Phi_i}; \Gamma_i)_{1 \leq i \leq n} \uStile A}
		 	
		 	\hypo{\pi_2}
		 	\infer[no rule]1{P:(\widehat \Psi; A, \Delta) \mid \Omega \uStile C}
		 	
		 	\infer2[\rCutUnG]{P:(\widehat{\Phi_i}, \widehat \Psi; \Gamma_i, \Delta)_{1 \leq i \leq n} \mid \Omega \uStile C}
		 \end{prooftree} \]
		 
		 As in \cref{lem:simple}, we notice that any rule that is not \rAx, \rOneR, \rBtR\ or \rOpR\ commutes under \rCutUnG\ from $\pi_1$. Therefore, we can restrict ourselves to the following cases:
		 
		 \begin{itemize}
		 	\item If $A = a$ a literal, and $\pi_1 = \begin{prooftree}
		 		\infer0[\rAx]{(\emptyset; a) \uStile a}
		 	\end{prooftree}$, and $\pi = \pi_2$ suffices.
		 	
		 	\item Let $A = \unit$, and $\pi_1 = \begin{prooftree}
		 		\infer0[\rOneR]{(\emptyset; \emptyset) \uStile A}
		 	\end{prooftree}$.

		 		We apply \cref{lem:posOne} to create $\pi_3 = \begin{prooftree}
		 			\hypo{\pi_3'}
		 			\infer[no rule]1{P:(\widehat \Psi; \Delta) \mid \Omega \uStile C}
		 			\infer1[\rOneL]{P:(\widehat \Psi; \unit, \Delta) \mid \Omega \uStile C}
		 		\end{prooftree}$ with $c(\pi_3) = c(\pi_2) = 0$ and $\pi_3 \equiv \pi_2$.
		 		
		 		There, $\pi_3'$ trivially suffices.
		 	
		 	\item Let $A = A_1 \boxtimes A_2$, and $\pi_1 = \begin{prooftree}
		 		\hypo{\pi_{11}}
		 		\infer[no rule]1{P:(\widehat{\Phi_i}; \Gamma_i)_{1 \leq i \leq n} \uStile A_1}
		 		
		 		\hypo{\pi_{12}}
		 		\infer[no rule]1{P:(\widehat{\Phi_j'}; \Gamma_j')_{1 \leq j \leq m} \uStile A_2}
		 		
		 		\infer2[\rBtL]{P:(\widehat{\Phi_i}, \widehat{\Phi_j'}; \Gamma_i, \Gamma_j')_{ij} \uStile A_1 \boxtimes A_2}
		 	\end{prooftree}$.
		 	
		 		According to \cref{lem:posBt}, there exists $\pi_3 = \begin{prooftree}
		 			\hypo{\pi_3'}
		 			\infer[no rule]1{P:(\widehat \Psi; A_1, A_2, \Delta) \mid \Omega \uStile C}
		 			\infer1[\rBtL]{P:(\widehat \Psi; A_1 \boxtimes A_2, \Delta) \mid \Omega \uStile C}
		 		\end{prooftree}$ with $c(\pi_3) = c(\pi_2) = 0$ and $\pi_3 \equiv \pi_2$.
		 		
		 		Trivially, the following $\pi$ suffices:
		 		\[\begin{prooftree}
		 			\hypo{\pi_{11}}
		 			\infer[no rule]1{P:(\widehat{\Phi_i}; \Gamma_i)_{1 \leq i \leq n} \uStile A_1}
		 			
		 			\hypo{\pi_{12}}
		 			\infer[no rule]1{P:(\widehat{\Phi_j'}; \Gamma_j')_{1 \leq j \leq m} \uStile A_2}
		 			
		 			\hypo{\pi_3'}
		 			\infer[no rule]1{P:(\widehat \Psi; A_1, A_2, \Delta) \mid \Omega \uStile C}
		 			
		 			\infer2[\rCutUnG]{P:(\widehat{\Phi_j'}; A_1, \Gamma_j', \Delta)_{1 \leq j \leq m} \mid \Omega \uStile C}
		 			
		 			\infer[double]2[\rCutUnG]{P:(\widehat{\Phi_i}, \widehat{\Phi_j'}, \widehat \Psi; \Gamma_i, \Gamma_j', \Delta)_{i, j} \mid \Omega \uStile C}
		 		\end{prooftree} \]
		 		With the double line denoting $m$ applications of \rCutUnG\ with $\pi_{11}$.
		 		
		 	\item
		 		Let $A = A_1 \oplus A_2$, $\pi_1 = \begin{prooftree}
		 			\hypo{\pi_{11}}
		 			\infer[no rule]1{P:(\widehat{\Phi_i}; \Gamma_i)_{1 \leq i \leq n} \uStile A_1}
		 			
		 			\hypo{\pi_{12}}
		 			\infer[no rule]1{P:(\widehat{\Phi_j'}; \Gamma_j')_{1 \leq j \leq m} \uStile A_2}
		 			
		 			\infer2[\rOpR]{P:(\widehat{\Phi_i}; \Gamma_i)_{1 \leq i \leq n} \mid (\widehat{\Phi_j'}; \Gamma_j')_{1 \leq j \leq m} \uStile A_1 \oplus A_2}
		 			
		 		\end{prooftree}$.
		 		
		 		Using \cref{lem:posOp}, we create $\pi_3 = \begin{prooftree}
		 			\hypo{\pi_3'}
		 			\infer[no rule]1{P:(\widehat \Psi; A_1, \Delta) \mid (\widehat \Psi; A_2, \Delta) \mid \Omega \uStile C}
		 			
		 			\infer1[\rOpL]{P:(\widehat \Psi; A_1 \oplus A_2, \Delta) \mid \Omega \uStile C}
		 		\end{prooftree}$
		 		
		 		with $c(\pi_3) = c(\pi_2) = 0$ and $\pi_3 \equiv \pi_2$.
		 		
		 		The following $\pi$ trivially suffices:
		 		\begin{small}
		 		\[\begin{prooftree}
		 			\hypo{\pi_{11}}
		 			\infer[no rule]1{P:(\widehat{\Phi_i}; \Gamma_i)_{i} \uStile A_1}
		 			
		 			\hypo{\pi_{12}}
		 			\infer[no rule]1{P:(\widehat{\Phi_j'}; \Gamma_j')_{j} \uStile A_2}
		 			
		 			\hypo{\pi_3'}
		 			\infer[no rule]1{P:(\widehat \Psi; A_1, \Delta) \mid (\widehat \Psi; A_2, \Delta) \mid \Omega \uStile C}
		 			
		 			\infer2[\rCutUnG]{P: (\widehat \Psi; A_1, \Delta) \mid (\widehat{\Phi_j'}, \widehat \Psi; \Gamma_j')_{1 \leq j \leq m} \mid \Omega \uStile C}
		 			
		 			\infer2[\rCutUnG]{P: (\widehat{\Phi_i}, \widehat \Psi; \Gamma_i A_1, \Delta)_{1 \leq i \leq n} \mid (\widehat{\Phi_j'}, \widehat \Psi; \Gamma_j')_{1 \leq j \leq m} \mid \Omega \uStile C}
		 		\end{prooftree} \]
		 		\end{small}
		 \end{itemize}
		 
		 Note: it follows easily that instances of \rCutUnG\ can be eliminated without any further lemmas.
		 
		 This covers all cases for this the last part of the lemmas.
	\end{proof}
	
	These lemmas allow us to state and prove a first version of the cut-elimination theorem. However, since cut rank does not take promises into account, it introduces a slight difference between fully cut-free proofs, and proofs with cut-rank zero. Therefore, we need to prove a slightly weaker theorem first.
	
	\begin{lemma}[Pre-Cut Elimination]
		\begin{itemize}
			\item If $\pi$ is a proof of $\Phi \lStile H$, there exists $\rho$ a cut-free proof of $\Phi \lStile H$ such that $\sem \pi = \sem \rho$.
			\item Let $P$ be a set of cut-free promises. If $\pi$ is a proof of $P: \Omega \uStile C$, there exists $\rho$ a cut-free proof of $P: \Omega \uStile C$ such that $\sem \pi = \sem \rho$.
		\end{itemize}
	\end{lemma}
	
	\begin{proof}
		Let us proceed by induction over $c(\pi)$.
		
		\proofsubparagraph{Initial Case}
		If $c(\pi) = 0$, it is trivial.
		
		\proofsubparagraph{Induction}
		Let $n \geq 0$. We assume that the theorem is true for all proofs of cut rank less than or equal to $n$.
		
		Let $\pi$ such that $c(\pi) = n+1$.
		
		We proceed by structural induction over $\pi$:
		\begin{itemize}
			\item If $\pi = \begin{prooftree}
				\hypo{\pi'}
				\infer[no rule]1{\Phi, G, F, \Psi \lStile H}
				\infer1[\rExchLin]{\Phi, F, G, \Psi \lStile F}
			\end{prooftree}$, then we apply the induction hypothesis to $\pi'$ and obtain $\rho$
			
			Then, $\rho = \begin{prooftree}
				\hypo{\rho'}
				\infer[no rule]1{\Phi, G, F, \Psi \lStile H}
				\infer1[\rExchLin]{\Phi, F, G, \Psi \lStile F}
			\end{prooftree}$ suffices.
			
			\item We proceed identically for \rUnitRight, \rOtLinL, \rArLinR, \rWiLinL{i}, \rExchUnW, \rExchUnH, \rExchUnG, \rOneL, \rBtL, \rOpL, \rPhTh, \rRotTh, \rOtUn, \rWiUn{i}\ and \rArUnPro.
			
			\item If $\pi = \begin{prooftree}
				\hypo{\pi_1}
				\infer[no rule]{1}{\Phi \lStile F}
				
				\hypo{\pi_2}
				\infer[no rule]{1}{\Psi \lStile G}
				
				\infer2[\rOtLinR]{\Phi, \Psi \lStile F \otimes G}
			\end{prooftree}$ then, by structural induction hypothesis over $\pi_1$ and $\pi_2$, we obtain $\rho_1$ and $\rho_2$.
			
			From this, $\rho = \begin{prooftree}
				\hypo{\rho_1}
				\infer[no rule]{1}{\Phi \lStile F}
				
				\hypo{\rho_2}
				\infer[no rule]{1}{\Psi \lStile G}
				
				\infer2[\rOtLinR]{\Phi, \Psi \lStile F \otimes G}
			\end{prooftree}$ suffices.
			
			\item We proceed identically for \rArLinL, \rWiLinR, \rBtR, \rOpR\ and \rUnitUn.
			
			\item If $\pi = \begin{prooftree}
				\hypo{\pi_1}
				\infer[no rule]1{\unlab(\widehat \Phi) \lStile F}
				
				\hypo{\pi_2}
				\infer[no rule]1{P \sqcup (\widehat \Phi, \pi_1, F \stackrel \alpha \multimap G, G^\beta): (G^\beta, \widehat{\Psi}; \Gamma) \mid \Omega \uStile C}
				
				\infer2[\rArUnNew]{P:(\widehat \Phi, F \stackrel{\alpha}{\multimap} G \widehat{\Psi}; \Gamma) \mid \Omega \uStile C}
			\end{prooftree}$ then, by induction hypothesis, there exists a cut free equivalent $\rho_1$ of $\pi_1$.
			
			We then consider $\pi_2'$, obtained by replacing $\pi_1$ by $\rho_1$ in the $G^\beta$ promise.
			The promises in the initial sequent of $\pi_2'$ are therefore all cut-free, and $d(\pi_2') = d(\pi_2) < d(\pi)$.
			
			Therefore, we can apply the induction hypothesis to $\pi_2'$ and obtain its equivalent $\rho_2$.
			
			We can introduce
			\[\rho = \begin{prooftree}
				\hypo{\rho_1}
				\infer[no rule]1{\unlab(\widehat \Phi) \lStile F}
				
				\hypo{\rho_2}
				\infer[no rule]1{P \sqcup (\widehat \Phi, \rho_1, F \stackrel \alpha \multimap G, G^\beta): (G^\beta, \widehat{\Psi_i}; \Gamma_i)_{i \in I} \mid \Omega \uStile C}
				
				\infer2[\rArUnNew]{P:(\widehat \Phi, F \stackrel{\alpha}{\multimap} G \widehat{\Psi_i}; \Gamma_i)_{i \in I} \mid \Omega \uStile C}
			\end{prooftree}\]
			
			which suffices.
			
			\item If $\pi = \begin{prooftree}
				\hypo{\pi_1}
				\infer[no rule]1{\Phi \lStile F}
				\hypo{\pi_2}
				\infer[no rule]1{F, \Psi \lStile H}
				
				\infer2[\rCutLin]{\Phi, \Psi \lStile H}
			\end{prooftree}$, then let us use the structural induction on $\pi_1$ and $\pi_2$ to obtain $\rho_1$ and $\rho_2$ cut-free.
			
			We then use \cref{lem:cutred} using $\rho_1$ and $\rho_2$, to create $\pi_3$ equivalent to \[\begin{prooftree}
				\hypo{\rho_1}
				\infer[no rule]1{\Phi \lStile F}
				\hypo{\rho_2}
				\infer[no rule]1{F, \Psi \lStile H}
				
				\infer2[\rCutLin]{\Phi, \Psi \lStile H}	
			\end{prooftree} \equiv \pi\] with cut rank less than or equal to $|F|$.
			
			We trivially have $|F| \leq n$, and we apply the cut-rank induction hypothesis to $\pi_3$ to create $\rho$.
			
			\item We proceed similarly for \rCutUnH\ and \rCutUnG.

		\end{itemize}
	\end{proof}
	
	\cutelim*

	\begin{proof}
		We obtain this from pre-cut elimination by first eliminating cuts from every proof in $P$ and turning it into $Q$. We then apply pre-cut elimination to the resulting proof.
	\end{proof}

\end{document}

%% file: example-circuit.tikz
\begin{tikzpicture}[scale=2]
	\begin{pgfonlayer}{nodelayer}
		\node [style=none] (0) at (-1, 0.5) {};
		\node [style=none] (1) at (1.5, 0.5) {};
		\node [style=none] (2) at (-1, 0) {};
		\node [style=none] (3) at (1.5, 0) {};
		\node [style=none] (4) at (-1, -0.5) {};
		\node [style=none] (5) at (1.5, -0.5) {};
		\node [style=dot] (6) at (-0.5, 0) {};
		\node [style=none] (7) at (-0.5, -0.5) {$\oplus$};
		\node [style=dot] (8) at (1, -0.5) {};
		\node [style=none] (9) at (1, 0.5) {$\oplus$};
		\node [style=box] (10) at (-0.5, 0.5) {$H$};
		\node [style=box] (11) at (0.25, 0) {$P(\frac\pi2)$};
	\end{pgfonlayer}
	\begin{pgfonlayer}{edgelayer}
		\draw (1.center) to (0.center);
		\draw (3.center) to (2.center);
		\draw (4.center) to (5.center);
		\draw (6) to (7.center);
		\draw (8) to (9.center);
	\end{pgfonlayer}
\end{tikzpicture}

%% file: Grover.tikz
\begin{tikzpicture}[scale=2]
	\begin{pgfonlayer}{nodelayer}
		\node [style=none] (0) at (-5.25, 0.5) {};
		\node [style=none] (1) at (-4.5, 0.5) {};
		\node [style=none] (2) at (-5.25, -0.5) {};
		\node [style=none] (3) at (-4.5, -0.5) {};
		\node [style=none] (4) at (-4.5, 0.625) {};
		\node [style=box] (5) at (-4.875, 0.5) {$H$};
		\node [style=box] (6) at (-4.875, -0.5) {$H$};
		\node [style=none] (7) at (-4.5, -0.625) {};
		\node [style=none] (8) at (-4, -0.625) {};
		\node [style=none] (9) at (-4, 0.625) {};
		\node [style=none] (10) at (-4, 0.5) {};
		\node [style=none] (11) at (-3.25, 0.5) {};
		\node [style=none] (12) at (-4, -0.5) {};
		\node [style=none] (13) at (-3.25, -0.5) {};
		\node [style=box] (14) at (-3.625, 0.5) {$H$};
		\node [style=box] (15) at (-3.625, -0.5) {$H$};
		\node [style=none] (16) at (-3.25, 0.625) {};
		\node [style=none] (17) at (-3.25, -0.625) {};
		\node [style=none] (18) at (-2.75, -0.625) {};
		\node [style=none] (19) at (-2.75, 0.625) {};
		\node [style=none] (20) at (-3, 0) {$R$};
		\node [style=none] (21) at (-2.75, 0.5) {};
		\node [style=none] (22) at (-2, 0.5) {};
		\node [style=none] (23) at (-2.75, -0.5) {};
		\node [style=none] (24) at (-2, -0.5) {};
		\node [style=box] (26) at (-2.375, 0.5) {$H$};
		\node [style=box] (27) at (-2.375, -0.5) {$H$};
		\node [style=none] (42) at (-1.5, 0.5) {};
		\node [style=none] (43) at (-1.25, 0.5) {};
		\node [style=none] (44) at (-1.5, -0.5) {};
		\node [style=none] (45) at (-1.25, -0.5) {};
		\node [style=none] (46) at (-1.25, 0.625) {};
		\node [style=none] (49) at (-1.25, -0.625) {};
		\node [style=none] (50) at (-0.75, -0.625) {};
		\node [style=none] (51) at (-0.75, 0.625) {};
		\node [style=none] (52) at (-0.75, 0.5) {};
		\node [style=none] (53) at (0, 0.5) {};
		\node [style=none] (54) at (-0.75, -0.5) {};
		\node [style=none] (55) at (0, -0.5) {};
		\node [style=box] (56) at (-0.375, 0.5) {$H$};
		\node [style=box] (57) at (-0.375, -0.5) {$H$};
		\node [style=none] (58) at (0, 0.625) {};
		\node [style=none] (59) at (0, -0.625) {};
		\node [style=none] (60) at (0.5, -0.625) {};
		\node [style=none] (61) at (0.5, 0.625) {};
		\node [style=none] (62) at (0.25, 0) {$R$};
		\node [style=none] (63) at (0.5, 0.5) {};
		\node [style=none] (64) at (1.25, 0.5) {};
		\node [style=none] (65) at (0.5, -0.5) {};
		\node [style=none] (66) at (1.25, -0.5) {};
		\node [style=box] (67) at (0.875, 0.5) {$H$};
		\node [style=box] (68) at (0.875, -0.5) {$H$};
		\node [style=none] (73) at (-1.75, 0.5) {...};
		\node [style=none] (74) at (-1.75, -0.5) {...};
		\node [style=none, rotate=90] (75) at (-5, 0) {...};
		\node [style=none, rotate=90] (76) at (1, 0) {...};
		\node [style=none, rotate=90] (77) at (-3.625, 0) {...};
		\node [style=none, rotate=90] (78) at (-2.375, 0) {...};
		\node [style=none, rotate=90] (80) at (-1.5, 0) {...};
		\node [style=none, rotate=90] (81) at (-0.375, 0) {...};
		\node [style=none] (82) at (-4.625, 0.875) {};
		\node [style=none] (83) at (1.125, 0.875) {};
		\node [style=none] (84) at (-1.75, 1.125) {iterate};
	\end{pgfonlayer}
	\begin{pgfonlayer}{edgelayer}
		\draw (1.center) to (0.center);
		\draw (3.center) to (2.center);
		\draw [style=dashed] (4.center) to (7.center);
		\draw [style=dashed] (7.center) to (8.center);
		\draw [style=dashed] (8.center) to (9.center);
		\draw [style=dashed] (9.center) to (4.center);
		\draw (11.center) to (10.center);
		\draw (13.center) to (12.center);
		\draw (16.center) to (17.center);
		\draw (17.center) to (18.center);
		\draw (18.center) to (19.center);
		\draw (19.center) to (16.center);
		\draw (22.center) to (21.center);
		\draw (24.center) to (23.center);
		\draw (43.center) to (42.center);
		\draw (45.center) to (44.center);
		\draw [style=dashed] (46.center) to (49.center);
		\draw [style=dashed] (49.center) to (50.center);
		\draw [style=dashed] (50.center) to (51.center);
		\draw [style=dashed] (51.center) to (46.center);
		\draw (53.center) to (52.center);
		\draw (55.center) to (54.center);
		\draw (58.center) to (59.center);
		\draw (59.center) to (60.center);
		\draw (60.center) to (61.center);
		\draw (61.center) to (58.center);
		\draw (64.center) to (63.center);
		\draw (66.center) to (65.center);
		\draw [decoration=brace, decorate] (82.center) to (83.center);
	\end{pgfonlayer}
\end{tikzpicture}

%% file: switch-circuit.tikz
\begin{tikzpicture}[scale=2]
	\begin{pgfonlayer}{nodelayer}
		\node [style=none] (0) at (0, 0.25) {};
		\node [style=none] (1) at (-2, 0.25) {};
		\node [style=none] (2) at (0, -0.25) {};
		\node [style=none] (3) at (-2, -0.25) {};
		\node [style=box] (4) at (-0.5, -0.25) {$U$};
		\node [style=box] (5) at (-1.5, -0.25) {$U$};
		\node [style=box] (6) at (-1, -0.25) {$V$};
		\node [style=dot] (7) at (-0.5, 0.25) {};
		\node [style=white dot] (8) at (-1.5, 0.25) {};
	\end{pgfonlayer}
	\begin{pgfonlayer}{edgelayer}
		\draw (2.center) to (3.center);
		\draw (1.center) to (0.center);
		\draw (7) to (4);
		\draw (8) to (5);
	\end{pgfonlayer}
\end{tikzpicture}

%% file: grenoble-proof.tikz
\begin{tikzpicture}
	\begin{pgfonlayer}{nodelayer}
		\node [style=rule dot] (1) at (0, -5) {$\oplus$};
		\node [style=none] (2) at (-6, -2) {};
		\node [style=none] (3) at (-7, 0) {};
		\node [style=none] (4) at (-7, -1) {};
		\node [style=none] (5) at (-5, -1) {};
		\node [style=none] (6) at (-5, 0) {};
		\node [style=none] (7) at (-6, -0.5) {\rUnitUn $f_0$};
		\node [style=none] (8) at (-6, -1) {};
		\node [style=none] (9) at (-6, 0) {};
		\node [style=rule dot] (10) at (-6, 2) {$\oplus$};
		\node [style=none] (11) at (-8, 4.5) {};
		\node [style=none] (12) at (-4, 4.5) {};
		\node [style=rule dot] (13) at (-8, 5.5) {$\oplus$};
		\node [style=rule dot] (18) at (-4, 5.5) {$\oplus$};
		\node [style=none] (19) at (-8.5, 9.5) {};
		\node [style=none] (20) at (-8.5, 8.5) {};
		\node [style=none] (21) at (-6.5, 8.5) {};
		\node [style=none] (22) at (-6.5, 9.5) {};
		\node [style=none] (23) at (-7.5, 9) {\rUnitUn $f_1$};
		\node [style=none] (24) at (-8, 8.5) {};
		\node [style=none] (25) at (-7.5, 9.5) {};
		\node [style=none] (26) at (-5.5, 9.5) {};
		\node [style=none] (27) at (-5.5, 8.5) {};
		\node [style=none] (28) at (-3.5, 8.5) {};
		\node [style=none] (29) at (-3.5, 9.5) {};
		\node [style=none] (30) at (-4.5, 9) {\rUnitUn $f_2$};
		\node [style=none] (31) at (-4, 8.5) {};
		\node [style=none] (32) at (-4.5, 9.5) {};
		\node [style=none] (33) at (-7, 8.5) {};
		\node [style=none] (34) at (-5, 8.5) {};
		\node [style=none] (35) at (-1, 0) {};
		\node [style=none] (36) at (-1, -1) {};
		\node [style=none] (37) at (1, -1) {};
		\node [style=none] (38) at (1, 0) {};
		\node [style=none] (39) at (0, -0.5) {\rUnitUn $f_1$};
		\node [style=none] (40) at (0, -1) {};
		\node [style=none] (41) at (0, 0) {};
		\node [style=rule dot] (42) at (0, 2) {$\oplus$};
		\node [style=none] (43) at (-2, 4.5) {};
		\node [style=none] (44) at (2, 4.5) {};
		\node [style=rule dot] (45) at (-2, 5.5) {$\oplus$};
		\node [style=rule dot] (50) at (2, 5.5) {$\oplus$};
		\node [style=none] (51) at (-2.5, 9.5) {};
		\node [style=none] (52) at (-2.5, 8.5) {};
		\node [style=none] (53) at (-0.5, 8.5) {};
		\node [style=none] (54) at (-0.5, 9.5) {};
		\node [style=none] (55) at (-1.5, 9) {\rUnitUn $f_2$};
		\node [style=none] (56) at (-2, 8.5) {};
		\node [style=none] (57) at (-1.5, 9.5) {};
		\node [style=none] (58) at (0.5, 9.5) {};
		\node [style=none] (59) at (0.5, 8.5) {};
		\node [style=none] (60) at (2.5, 8.5) {};
		\node [style=none] (61) at (2.5, 9.5) {};
		\node [style=none] (62) at (1.5, 9) {\rUnitUn $f_0$};
		\node [style=none] (63) at (2, 8.5) {};
		\node [style=none] (64) at (1.5, 9.5) {};
		\node [style=none] (65) at (-1, 8.5) {};
		\node [style=none] (66) at (1, 8.5) {};
		\node [style=none] (67) at (5, 0) {};
		\node [style=none] (68) at (5, -1) {};
		\node [style=none] (69) at (7, -1) {};
		\node [style=none] (70) at (7, 0) {};
		\node [style=none] (71) at (6, -0.5) {\rUnitUn $f_2$};
		\node [style=none] (72) at (6, -1) {};
		\node [style=none] (73) at (6, 0) {};
		\node [style=rule dot] (74) at (6, 2) {$\oplus$};
		\node [style=none] (75) at (4, 4.5) {};
		\node [style=none] (76) at (8, 4.5) {};
		\node [style=rule dot] (77) at (4, 5.5) {$\oplus$};
		\node [style=rule dot] (82) at (8, 5.5) {$\oplus$};
		\node [style=none] (83) at (3.5, 9.5) {};
		\node [style=none] (84) at (3.5, 8.5) {};
		\node [style=none] (85) at (5.5, 8.5) {};
		\node [style=none] (86) at (5.5, 9.5) {};
		\node [style=none] (87) at (4.5, 9) {\rUnitUn $f_0$};
		\node [style=none] (88) at (4, 8.5) {};
		\node [style=none] (89) at (4.5, 9.5) {};
		\node [style=none] (90) at (6.5, 9.5) {};
		\node [style=none] (91) at (6.5, 8.5) {};
		\node [style=none] (92) at (8.5, 8.5) {};
		\node [style=none] (93) at (8.5, 9.5) {};
		\node [style=none] (94) at (7.5, 9) {\rUnitUn $f_1$};
		\node [style=none] (95) at (8, 8.5) {};
		\node [style=none] (96) at (7.5, 9.5) {};
		\node [style=none] (97) at (5, 8.5) {};
		\node [style=none] (98) at (7, 8.5) {};
		\node [style=none] (99) at (6, -2) {};
		\node [style=none] (100) at (-8.5, 12.5) {};
		\node [style=none] (101) at (-8.5, 11.5) {};
		\node [style=none] (102) at (-6.5, 11.5) {};
		\node [style=none] (103) at (-6.5, 12.5) {};
		\node [style=none] (106) at (-7.5, 12.5) {};
		\node [style=none] (107) at (-7.5, 11.5) {};
		\node [style=none] (108) at (-5.5, 12.5) {};
		\node [style=none] (109) at (-5.5, 11.5) {};
		\node [style=none] (110) at (-3.5, 11.5) {};
		\node [style=none] (111) at (-3.5, 12.5) {};
		\node [style=none] (112) at (-4.5, 12) {\rUnitUn $f_1$};
		\node [style=none] (113) at (-4.5, 11.5) {};
		\node [style=none] (114) at (-4.5, 12.5) {};
		\node [style=none] (115) at (3.5, 12.5) {};
		\node [style=none] (116) at (3.5, 11.5) {};
		\node [style=none] (117) at (5.5, 11.5) {};
		\node [style=none] (118) at (5.5, 12.5) {};
		\node [style=none] (119) at (4.5, 12) {\rUnitUn $f_1$};
		\node [style=none] (120) at (4.5, 11.5) {};
		\node [style=none] (121) at (4.5, 12.5) {};
		\node [style=none] (122) at (-2.5, 12.5) {};
		\node [style=none] (123) at (-2.5, 11.5) {};
		\node [style=none] (124) at (-0.5, 11.5) {};
		\node [style=none] (125) at (-0.5, 12.5) {};
		\node [style=none] (126) at (-1.5, 12) {\rUnitUn $f_0$};
		\node [style=none] (127) at (-1.5, 11.5) {};
		\node [style=none] (128) at (-1.5, 12.5) {};
		\node [style=none] (129) at (6.5, 12.5) {};
		\node [style=none] (130) at (6.5, 11.5) {};
		\node [style=none] (131) at (8.5, 11.5) {};
		\node [style=none] (132) at (8.5, 12.5) {};
		\node [style=none] (133) at (7.5, 12) {\rUnitUn $f_0$};
		\node [style=none] (134) at (7.5, 11.5) {};
		\node [style=none] (135) at (7.5, 12.5) {};
		\node [style=none] (136) at (0.5, 12.5) {};
		\node [style=none] (137) at (0.5, 11.5) {};
		\node [style=none] (138) at (2.5, 11.5) {};
		\node [style=none] (139) at (2.5, 12.5) {};
		\node [style=none] (140) at (1.5, 12) {\rUnitUn $f_2$};
		\node [style=none] (141) at (1.5, 11.5) {};
		\node [style=none] (142) at (1.5, 12.5) {};
		\node [style=none] (149) at (-7.5, 12) {\rUnitUn $f_2$};
		\node [style=rule dot] (150) at (0, -7) {$\boxtimes$};
		\node [style=rule dot] (151) at (0, -9) {$\boxtimes$};
		\node [style=none] (152) at (0, -11) {};
		\node [style=none] (153) at (-9, 14) {};
		\node [style=none] (154) at (9, 14) {};
		\node [style=none] (155) at (-9, 17) {};
		\node [style=none] (156) at (9, 17) {};
		\node [style=none] (157) at (-7.5, 14) {};
		\node [style=none] (158) at (-4.5, 14) {};
		\node [style=none] (159) at (-1.5, 14) {};
		\node [style=none] (160) at (1.5, 14) {};
		\node [style=none] (161) at (4.5, 14) {};
		\node [style=none] (162) at (7.5, 14) {};
		\node [style=none] (163) at (0, 15.5) {$\pi_{\text{fin}}$};
	\end{pgfonlayer}
	\begin{pgfonlayer}{edgelayer}
		\draw [style=box edge] (4.center)
			 to (3.center)
			 to (6.center)
			 to (5.center)
			 to cycle;
		\draw [in=-90, out=165, looseness=0.50] (1) to (2.center);
		\draw (2.center) to (8.center);
		\draw [style=box edge] (20.center)
			 to (19.center)
			 to (22.center)
			 to (21.center)
			 to cycle;
		\draw [style=box edge] (27.center)
			 to (26.center)
			 to (29.center)
			 to (28.center)
			 to cycle;
		\draw (9.center) to (10);
		\draw [in=270, out=135] (10) to (11.center);
		\draw (11.center) to (13);
		\draw [in=45, out=-90] (12.center) to (10);
		\draw (12.center) to (18);
		\draw [style=box edge] (36.center)
			 to (35.center)
			 to (38.center)
			 to (37.center)
			 to cycle;
		\draw [style=box edge] (52.center)
			 to (51.center)
			 to (54.center)
			 to (53.center)
			 to cycle;
		\draw [style=box edge] (59.center)
			 to (58.center)
			 to (61.center)
			 to (60.center)
			 to cycle;
		\draw (41.center) to (42);
		\draw [in=-90, out=135] (42) to (43.center);
		\draw (43.center) to (45);
		\draw [in=45, out=-90] (44.center) to (42);
		\draw (44.center) to (50);
		\draw (1) to (40.center);
		\draw [style=box edge] (68.center)
			 to (67.center)
			 to (70.center)
			 to (69.center)
			 to cycle;
		\draw [style=box edge] (85.center)
			 to (84.center)
			 to (83.center)
			 to (86.center)
			 to cycle;
		\draw [style=box edge] (93.center)
			 to (92.center)
			 to (91.center)
			 to (90.center)
			 to cycle;
		\draw (73.center) to (74);
		\draw [in=-90, out=135] (74) to (75.center);
		\draw (75.center) to (77);
		\draw [in=45, out=-90] (76.center) to (74);
		\draw (76.center) to (82);
		\draw [in=15, out=-90, looseness=0.50] (99.center) to (1);
		\draw (99.center) to (72.center);
		\draw [style=box edge] (103.center)
			 to (102.center)
			 to (101.center)
			 to (100.center)
			 to cycle;
		\draw [style=box edge] (108.center)
			 to (111.center)
			 to (110.center)
			 to (109.center)
			 to cycle;
		\draw [style=box edge] (117.center)
			 to (116.center)
			 to (115.center)
			 to (118.center)
			 to cycle;
		\draw [style=box edge] (123.center)
			 to (122.center)
			 to (125.center)
			 to (124.center)
			 to cycle;
		\draw [style=box edge] (130.center)
			 to (129.center)
			 to (132.center)
			 to (131.center)
			 to cycle;
		\draw [style=box edge] (137.center)
			 to (136.center)
			 to (139.center)
			 to (138.center)
			 to cycle;
		\draw (25.center) to (107.center);
		\draw (113.center) to (32.center);
		\draw (127.center) to (57.center);
		\draw (141.center) to (64.center);
		\draw (120.center) to (89.center);
		\draw (134.center) to (96.center);
		\draw [bend left] (13) to (24.center);
		\draw [in=-90, out=120] (18) to (33.center);
		\draw [bend left=15, looseness=1.25] (45) to (56.center);
		\draw [in=-90, out=120] (50) to (65.center);
		\draw [bend left=45, looseness=0.75] (77) to (88.center);
		\draw [in=-75, out=60, looseness=0.75] (13) to (31.center);
		\draw [in=-120, out=45, looseness=1.25] (18) to (34.center);
		\draw [in=-90, out=60] (45) to (63.center);
		\draw [in=-90, out=45] (50) to (66.center);
		\draw [in=-90, out=60] (77) to (95.center);
		\draw [in=-90, out=120] (82) to (97.center);
		\draw [in=-90, out=60] (82) to (98.center);
		\draw (152.center) to (151);
		\draw (151) to (150);
		\draw (150) to (1);
		\draw [style=box edge] (156.center)
			 to (154.center)
			 to (153.center)
			 to (155.center)
			 to cycle;
		\draw (106.center) to (157.center);
		\draw (114.center) to (158.center);
		\draw (128.center) to (159.center);
		\draw (142.center) to (160.center);
		\draw (121.center) to (161.center);
		\draw (135.center) to (162.center);
	\end{pgfonlayer}
\end{tikzpicture}

%% file: rep-graph.tikz
\begin{tikzpicture}
	\begin{pgfonlayer}{nodelayer}
		\node [style=none] (0) at (-12, 3) {\rExchUnW:};
		\node [style=none] (1) at (-9, 2) {};
		\node [style=none] (2) at (-9, 4) {};
		\node [style=none] (3) at (-7, 2) {};
		\node [style=none] (4) at (-7, 4) {};
		\node [style=none] (5) at (-9, 1) {};
		\node [style=none] (6) at (-7, 1) {};
		\node [style=none] (7) at (-9, 5) {};
		\node [style=none] (8) at (-7, 5) {};
		\node [style=none] (9) at (-9, 0) {$\omega$};
		\node [style=none] (10) at (-7, 0) {$\omega'$};
		\node [style=none] (11) at (-9, 6) {$\omega'$};
		\node [style=none] (12) at (-7, 6) {$\omega$};
		\node [style=none] (13) at (-9.5, 1.5) {};
		\node [style=none] (14) at (-6.5, 1.5) {};
		\node [style=none] (15) at (-9.5, 4.5) {};
		\node [style=none] (16) at (-6.75, 4.5) {};
		\node [style=none] (17) at (-10.25, 1.5) {$P$};
		\node [style=none] (18) at (-10.25, 4.5) {$P$};
		\node [style=none] (19) at (-4, 3) {\rExchUnH:};
		\node [style=none] (20) at (0, 1) {};
		\node [style=none] (21) at (0, 5) {};
		\node [style=rule dot] (22) at (0, 3) {$\sigma_1$};
		\node [style=none] (23) at (0, 0) {$(\widehat{\Phi}, F^\alpha, G^\beta, \widehat \Psi; \Gamma)$};
		\node [style=none] (24) at (0, 6) {$(\widehat{\Phi}, G^\beta, F^\alpha, \widehat \Psi; \Gamma)$};
		\node [style=none] (25) at (-0.5, 1.5) {};
		\node [style=none] (26) at (0.5, 1.5) {};
		\node [style=none] (27) at (-0.5, 4.5) {};
		\node [style=none] (28) at (0.5, 4.5) {};
		\node [style=none] (29) at (-1.25, 1.5) {$P$};
		\node [style=none] (30) at (-1.25, 4.5) {$P$};
		\node [style=none] (31) at (4, 3) {\rExchUnG:};
		\node [style=none] (32) at (8, 1) {};
		\node [style=none] (33) at (8, 5) {};
		\node [style=rule dot] (34) at (8, 3) {$\sigma_2$};
		\node [style=none] (35) at (8, 0) {$(\widehat \Phi; \Gamma, A, B, \Delta)$};
		\node [style=none] (36) at (8, 6) {$(\widehat \Phi; \Gamma, B, A, \Delta)$};
		\node [style=none] (37) at (7.5, 1.5) {};
		\node [style=none] (38) at (8.5, 1.5) {};
		\node [style=none] (39) at (7.5, 4.5) {};
		\node [style=none] (40) at (8.5, 4.5) {};
		\node [style=none] (41) at (6.75, 1.5) {$P$};
		\node [style=none] (42) at (6.75, 4.5) {$P$};
		\node [style=none] (43) at (-11.5, -21) {\rOtUn:};
		\node [style=none] (44) at (-8.5, -23) {};
		\node [style=none] (45) at (-8.5, -19) {};
		\node [style=rule dot] (46) at (-8.5, -21) {$\otimes$};
		\node [style=none] (47) at (-8.5, -24) {$((F \otimes G)^\alpha, \widehat \Phi; \Gamma)$};
		\node [style=none] (48) at (-8.5, -18) {$(F^{\otimes_1 \alpha}, G^{\otimes_2 \alpha}, \widehat \Phi; \Gamma)$};
		\node [style=none] (49) at (-9, -22.5) {};
		\node [style=none] (50) at (-8, -22.5) {};
		\node [style=none] (51) at (-9, -19.5) {};
		\node [style=none] (52) at (-8, -19.5) {};
		\node [style=none] (53) at (-9.75, -22.5) {$P$};
		\node [style=none] (54) at (-9.75, -19.5) {$P$};
		\node [style=none] (55) at (-4, -21) {\rWiUn i:};
		\node [style=none] (56) at (0, -23) {};
		\node [style=none] (57) at (0, -19) {};
		\node [style=rule dot] (58) at (0, -21) {$\with_i$};
		\node [style=none] (59) at (0, -24) {$((F_1 \with F_2)^\alpha, \widehat \Phi; \Gamma)$};
		\node [style=none] (60) at (0, -18) {$(F_i^{\oplus_i \alpha}, \widehat \Phi; \Gamma)$};
		\node [style=none] (61) at (-0.5, -22.5) {};
		\node [style=none] (62) at (0.5, -22.5) {};
		\node [style=none] (63) at (-0.5, -19.5) {};
		\node [style=none] (64) at (0.5, -19.5) {};
		\node [style=none] (65) at (-1.25, -22.5) {$P$};
		\node [style=none] (66) at (-1.25, -19.5) {$P$};
		\node [style=none] (67) at (3.5, -21) {\rArUnNew:};
		\node [style=none] (68) at (8, -23) {};
		\node [style=none] (69) at (8, -19) {};
		\node [style=none] (71) at (8, -24) {$(\widehat \Phi, F \stackrel{\alpha}{\multimap} G, \widehat \Psi; \Gamma)$};
		\node [style=none] (72) at (8, -18) {$(G^{\beta}, \widehat \Psi; \Gamma)$};
		\node [style=none] (73) at (7.5, -22.5) {};
		\node [style=none] (74) at (8.5, -22.5) {};
		\node [style=none] (75) at (7.5, -19.5) {};
		\node [style=none] (76) at (8.5, -19.5) {};
		\node [style=none] (77) at (6.75, -22.5) {$P$};
		\node [style=none] (78) at (6.25, -19.5) {$P\sqcup p_L$};
		\node [style=rule dot] (79) at (5.5, -21) {$\pi_L$};
		\node [style=none] (80) at (12.25, -21) {\rArUnPro:};
		\node [style=none] (93) at (-12.25, -5) {\rBtL:};
		\node [style=none] (94) at (-8.25, -7) {};
		\node [style=none] (95) at (-8.25, -3) {};
		\node [style=rule dot] (96) at (-8.25, -5) {$\boxtimes$};
		\node [style=none] (97) at (-8.25, -8) {$(\widehat \Phi; A \boxtimes B, \Gamma)$};
		\node [style=none] (98) at (-8.25, -2) {$(\widehat \Phi; A, B, \Gamma)$};
		\node [style=none] (99) at (-8.75, -6.5) {};
		\node [style=none] (100) at (-7.75, -6.5) {};
		\node [style=none] (101) at (-8.75, -3.5) {};
		\node [style=none] (102) at (-7.75, -3.5) {};
		\node [style=none] (103) at (-9.5, -6.5) {$P$};
		\node [style=none] (104) at (-9.5, -3.5) {$P$};
		\node [style=none] (105) at (-4, -5) {\rOpL:};
		\node [style=none] (106) at (0, -7) {};
		\node [style=none] (107) at (-1, -3) {};
		\node [style=rule dot] (108) at (0, -5) {$\oplus$};
		\node [style=none] (109) at (0, -8) {$(\widehat \Phi; A \oplus B, \Gamma)$};
		\node [style=none] (110) at (-1.5, -2.5) {$(\widehat \Phi;A, \Gamma)$};
		\node [style=none] (111) at (-0.5, -6.5) {};
		\node [style=none] (112) at (0.5, -6.5) {};
		\node [style=none] (113) at (-1.25, -3.5) {};
		\node [style=none] (114) at (1.25, -3.5) {};
		\node [style=none] (115) at (-1.25, -6.5) {$P$};
		\node [style=none] (116) at (-2, -3.5) {$P$};
		\node [style=none] (117) at (1, -3) {};
		\node [style=none] (118) at (1.5, -2.5) {$(\widehat \Phi;B, \Gamma)$};
		\node [style=none] (119) at (4, -5) {\rRotTh:};
		\node [style=none] (120) at (7, -7) {};
		\node [style=none] (121) at (8, -3) {};
		\node [style=rule dot] (122) at (8, -5) {$R_\theta$};
		\node [style=none] (123) at (7, -8) {$(\widehat \Phi; \Gamma)$};
		\node [style=none] (124) at (8, -2) {$(\widehat \Phi; (\unit \oplus \unit), \Gamma)$};
		\node [style=none] (125) at (6.75, -6.5) {};
		\node [style=none] (126) at (9.25, -6.5) {};
		\node [style=none] (127) at (7.5, -3.5) {};
		\node [style=none] (128) at (8.5, -3.5) {};
		\node [style=none] (129) at (6.25, -6.5) {$P$};
		\node [style=none] (130) at (6.75, -3.5) {$P$};
		\node [style=none] (131) at (9, -8) {$(\widehat \Phi; \Gamma)$};
		\node [style=none] (132) at (9, -7) {};
		\node [style=none] (133) at (12, -4.75) {\rPhTh:};
		\node [style=none] (134) at (15, -7) {};
		\node [style=none] (135) at (15, -3) {};
		\node [style=rule dot] (136) at (15, -5) {$P_\theta$};
		\node [style=none] (137) at (15, -8) {$\omega$};
		\node [style=none] (138) at (15, -2) {$\omega$};
		\node [style=none] (139) at (14.5, -6.5) {};
		\node [style=none] (140) at (15.5, -6.5) {};
		\node [style=none] (141) at (14.5, -3.5) {};
		\node [style=none] (142) at (15.5, -3.5) {};
		\node [style=none] (143) at (13.75, -6.5) {$P$};
		\node [style=none] (144) at (13.75, -3.5) {$P$};
		\node [style=none] (145) at (8, -21) {\rArUnNew $\alpha$};
		\node [style=none] (146) at (6.75, -20.25) {};
		\node [style=none] (147) at (9.25, -20.25) {};
		\node [style=none] (148) at (6.75, -21.75) {};
		\node [style=none] (149) at (9.25, -21.75) {};
		\node [style=none] (150) at (7, -21) {};
		\node [style=none] (151) at (8, -20.25) {};
		\node [style=none] (152) at (8, -21.75) {};
		\node [style=none] (153) at (15, -23) {};
		\node [style=none] (154) at (15, -19) {};
		\node [style=none] (155) at (15, -24) {$(\widehat \Phi, F \stackrel{\alpha}{\multimap} G, \widehat \Psi; \Gamma)$};
		\node [style=none] (156) at (15, -18) {$(G^{\beta}, \widehat \Psi; \Gamma)$};
		\node [style=none] (157) at (14.5, -22.5) {};
		\node [style=none] (158) at (15.5, -22.5) {};
		\node [style=none] (159) at (14.5, -19.5) {};
		\node [style=none] (160) at (15.5, -19.5) {};
		\node [style=none] (162) at (13.25, -19.5) {$P\sqcup p_L$};
		\node [style=none] (164) at (15, -21) {\rArUnNew $\alpha$};
		\node [style=none] (165) at (13.75, -20.25) {};
		\node [style=none] (166) at (16.25, -20.25) {};
		\node [style=none] (167) at (13.75, -21.75) {};
		\node [style=none] (168) at (16.25, -21.75) {};
		\node [style=none] (169) at (14, -21) {};
		\node [style=none] (170) at (15, -20.25) {};
		\node [style=none] (171) at (15, -21.75) {};
		\node [style=none] (172) at (13.25, -22.5) {$P\sqcup p_L$};
		\node [style=none] (173) at (-9, -11.5) {};
		\node [style=none] (174) at (-5, -11.5) {};
		\node [style=none] (175) at (-9, -13) {};
		\node [style=none] (176) at (-5, -13) {};
		\node [style=none] (177) at (-7, -12.25) {\rUnitUn $\alpha$};
		\node [style=none] (178) at (-8.5, -13) {};
		\node [style=none] (179) at (-5.5, -13) {};
		\node [style=none] (180) at (-8.5, -14.25) {};
		\node [style=none] (181) at (-5.5, -14.25) {};
		\node [style=none] (182) at (-7, -11.5) {};
		\node [style=none] (183) at (-7, -10.25) {};
		\node [style=none] (184) at (-7, -15) {$(\widehat{\Phi_i}, A \stackrel{\alpha}{\uarrow} B, \widehat{\Psi}; \Gamma_i, \Delta)_i$};
		\node [style=none] (185) at (-7, -9.75) {$(\widehat \Psi; B, \Delta)$};
		\node [style=none] (186) at (-9, -12.25) {};
		\node [style=rule dot] (187) at (-10.25, -12.25) {$\pi_U$};
		\node [style=none] (188) at (-12, -12.25) {\rUnitUn:};
		\node [style=none] (189) at (-5, -13.75) {};
		\node [style=none] (190) at (-9, -13.75) {};
		\node [style=none] (191) at (-6.5, -10.75) {};
		\node [style=none] (192) at (-7.5, -10.75) {};
		\node [style=none] (193) at (-9.75, -13.75) {$P$};
		\node [style=none] (194) at (-8.25, -10.75) {$P$};
		\node [style=none] (195) at (-7, -13.5) {...};
		\node [style=none] (196) at (2, -11.5) {};
		\node [style=none] (197) at (6, -11.5) {};
		\node [style=none] (198) at (2, -13) {};
		\node [style=none] (199) at (6, -13) {};
		\node [style=none] (200) at (4, -12.25) {\rCutUnG};
		\node [style=none] (201) at (2.5, -13) {};
		\node [style=none] (202) at (5.5, -13) {};
		\node [style=none] (203) at (2.5, -14.25) {};
		\node [style=none] (204) at (5.5, -14.25) {};
		\node [style=none] (205) at (4, -11.5) {};
		\node [style=none] (206) at (4, -10.25) {};
		\node [style=none] (207) at (4, -15) {$(\widehat{\Phi_i}, \widehat{\Psi}; \Gamma_i, \Delta)_i$};
		\node [style=none] (208) at (4, -9.75) {$(\widehat \Psi; B, \Delta)$};
		\node [style=none] (209) at (2, -12.25) {};
		\node [style=rule dot] (210) at (0.5, -12.25) {$\pi'$};
		\node [style=none] (211) at (-1.5, -12.25) {\rCutUnG:};
		\node [style=none] (212) at (6, -13.75) {};
		\node [style=none] (213) at (2, -13.75) {};
		\node [style=none] (214) at (4.5, -10.75) {};
		\node [style=none] (215) at (3.5, -10.75) {};
		\node [style=none] (216) at (1.25, -13.75) {$P$};
		\node [style=none] (217) at (2.75, -10.75) {$P$};
		\node [style=none] (218) at (4, -13.5) {...};
		\node [style=none] (219) at (12.75, -11.5) {};
		\node [style=none] (220) at (16.75, -11.5) {};
		\node [style=none] (221) at (12.75, -13) {};
		\node [style=none] (222) at (16.75, -13) {};
		\node [style=none] (223) at (14.75, -12.25) {\rCutUnH};
		\node [style=none] (224) at (13.25, -13) {};
		\node [style=none] (225) at (16.25, -13) {};
		\node [style=none] (226) at (13.25, -14.25) {};
		\node [style=none] (227) at (16.25, -14.25) {};
		\node [style=none] (228) at (13.25, -11.5) {};
		\node [style=none] (229) at (13.25, -10.25) {};
		\node [style=none] (230) at (14.75, -15) {$(\widehat{\Phi}, \widehat{\Psi_i}; \Gamma_i)_i$};
		\node [style=none] (231) at (14.75, -9.75) {$(F^\alpha, \widehat{\Psi_i};\Gamma_i)_i$};
		\node [style=none] (232) at (12.75, -12.25) {};
		\node [style=rule dot] (233) at (11.5, -12.25) {$\pi'$};
		\node [style=none] (234) at (9.5, -12.25) {\rCutUnH:};
		\node [style=none] (235) at (16.75, -13.75) {};
		\node [style=none] (236) at (12.75, -13.75) {};
		\node [style=none] (237) at (16.75, -10.75) {};
		\node [style=none] (238) at (12.75, -10.75) {};
		\node [style=none] (239) at (12, -13.75) {$P$};
		\node [style=none] (240) at (11.75, -10.75) {$P$};
		\node [style=none] (241) at (14.75, -13.5) {...};
		\node [style=none] (242) at (16.25, -11.5) {};
		\node [style=none] (243) at (16.25, -10.25) {};
		\node [style=none] (244) at (15, 1) {};
		\node [style=none] (245) at (15, 5) {};
		\node [style=rule dot] (246) at (15, 3) {\rOneL};
		\node [style=none] (247) at (15, 0) {$(\widehat \Phi; \unit, \Gamma)$};
		\node [style=none] (249) at (14.5, 1.5) {};
		\node [style=none] (250) at (15.5, 1.5) {};
		\node [style=none] (251) at (14.5, 4.5) {};
		\node [style=none] (252) at (15.5, 4.5) {};
		\node [style=none] (253) at (13.75, 1.5) {$P$};
		\node [style=none] (254) at (13.75, 4.5) {$P$};
		\node [style=none] (256) at (15, 6) {$(\widehat \Phi; \Gamma)$};
		\node [style=none] (257) at (12.25, 3) {\rOneL:};
	\end{pgfonlayer}
	\begin{pgfonlayer}{edgelayer}
		\draw (7.center) to (2.center);
		\draw [in=90, out=-90, looseness=0.75] (2.center) to (3.center);
		\draw (3.center) to (6.center);
		\draw (5.center) to (1.center);
		\draw [in=-90, out=90, looseness=0.75] (1.center) to (4.center);
		\draw (4.center) to (8.center);
		\draw [style=dashedline] (13.center) to (14.center);
		\draw [style=dashedline] (16.center) to (15.center);
		\draw (20.center) to (22);
		\draw (22) to (21.center);
		\draw [style=dashedline] (26.center) to (25.center);
		\draw [style=dashedline] (28.center) to (27.center);
		\draw (32.center) to (34);
		\draw (34) to (33.center);
		\draw [style=dashedline] (38.center) to (37.center);
		\draw [style=dashedline] (40.center) to (39.center);
		\draw (44.center) to (46);
		\draw (46) to (45.center);
		\draw [style=dashedline] (50.center) to (49.center);
		\draw [style=dashedline] (52.center) to (51.center);
		\draw (56.center) to (58);
		\draw (58) to (57.center);
		\draw [style=dashedline] (62.center) to (61.center);
		\draw [style=dashedline] (64.center) to (63.center);
		\draw [style=dashedline] (74.center) to (73.center);
		\draw [style=dashedline] (76.center) to (75.center);
		\draw (94.center) to (96);
		\draw (96) to (95.center);
		\draw [style=dashedline] (100.center) to (99.center);
		\draw [style=dashedline] (102.center) to (101.center);
		\draw (106.center) to (108);
		\draw [in=-90, out=135] (108) to (107.center);
		\draw [style=dashedline] (112.center) to (111.center);
		\draw [style=dashedline] (114.center) to (113.center);
		\draw [in=-90, out=45] (108) to (117.center);
		\draw [in=-135, out=90] (120.center) to (122);
		\draw (122) to (121.center);
		\draw [style=dashedline] (126.center) to (125.center);
		\draw [style=dashedline] (128.center) to (127.center);
		\draw [in=-45, out=90] (132.center) to (122);
		\draw (134.center) to (136);
		\draw (136) to (135.center);
		\draw [style=dashedline] (140.center) to (139.center);
		\draw [style=dashedline] (142.center) to (141.center);
		\draw [style=box edge] (149.center)
			 to (148.center)
			 to (146.center)
			 to (147.center)
			 to cycle;
		\draw (150.center) to (79);
		\draw (152.center) to (68.center);
		\draw (151.center) to (69.center);
		\draw [style=dashedline] (158.center) to (157.center);
		\draw [style=dashedline] (160.center) to (159.center);
		\draw [style=box edge] (168.center)
			 to (167.center)
			 to (165.center)
			 to (166.center)
			 to cycle;
		\draw (171.center) to (153.center);
		\draw (170.center) to (154.center);
		\draw [style=box edge] (175.center)
			 to (173.center)
			 to (174.center)
			 to (176.center)
			 to cycle;
		\draw (178.center) to (180.center);
		\draw (179.center) to (181.center);
		\draw (182.center) to (183.center);
		\draw (187) to (186.center);
		\draw [style=dashedline] (190.center) to (189.center);
		\draw [style=dashedline] (191.center) to (192.center);
		\draw [style=box edge] (198.center)
			 to (196.center)
			 to (197.center)
			 to (199.center)
			 to cycle;
		\draw (201.center) to (203.center);
		\draw (202.center) to (204.center);
		\draw (205.center) to (206.center);
		\draw (210) to (209.center);
		\draw [style=dashedline] (213.center) to (212.center);
		\draw [style=dashedline] (214.center) to (215.center);
		\draw [style=box edge] (222.center)
			 to (221.center)
			 to (219.center)
			 to (220.center)
			 to cycle;
		\draw (224.center) to (226.center);
		\draw (225.center) to (227.center);
		\draw (228.center) to (229.center);
		\draw (233) to (232.center);
		\draw [style=dashedline] (236.center) to (235.center);
		\draw [style=dashedline] (237.center) to (238.center);
		\draw (242.center) to (243.center);
		\draw (244.center) to (246);
		\draw (246) to (245.center);
		\draw [style=dashedline] (250.center) to (249.center);
		\draw [style=dashedline] (252.center) to (251.center);
	\end{pgfonlayer}
\end{tikzpicture}

%% file: cutred-prf-sigma-1.tikz
\begin{tikzpicture}
	\begin{pgfonlayer}{nodelayer}
		\node [style=none] (0) at (-2, -1.5) {};
		\node [style=none] (1) at (-2, -2.5) {};
		\node [style=none] (2) at (1.25, -1.5) {};
		\node [style=none] (3) at (1.25, -2.5) {};
		\node [style=none] (4) at (-0.25, 0.5) {};
		\node [style=none] (5) at (-0.25, -0.5) {};
		\node [style=none] (6) at (3, 0.5) {};
		\node [style=none] (7) at (3, -0.5) {};
		\node [style=none] (8) at (-1.75, 2.5) {};
		\node [style=none] (9) at (-1.75, -1.5) {};
		\node [style=none] (10) at (-0.75, 2.5) {};
		\node [style=none] (11) at (-0.75, -1.5) {};
		\node [style=none] (12) at (0, -0.5) {};
		\node [style=none] (13) at (0, -1.5) {};
		\node [style=none] (14) at (1, -0.5) {};
		\node [style=none] (15) at (1, -1.5) {};
		\node [style=none] (16) at (-0.375, -2) {\rCutUnH};
		\node [style=none] (17) at (1.375, 0) {$B$};
		\node [style=none] (18) at (0, 2.5) {};
		\node [style=none] (19) at (0, 0.5) {};
		\node [style=none] (20) at (1, 2.5) {};
		\node [style=none] (21) at (1, 0.5) {};
		\node [style=none] (22) at (1.75, 2.5) {};
		\node [style=none] (23) at (1.75, 0.5) {};
		\node [style=none] (24) at (2.75, 2.5) {};
		\node [style=none] (25) at (2.75, 0.5) {};
		\node [style=none] (26) at (1.75, -0.5) {};
		\node [style=none] (27) at (1.75, -3.5) {};
		\node [style=none] (28) at (2.75, -0.5) {};
		\node [style=none] (29) at (2.75, -3.5) {};
		\node [style=none] (30) at (-1.75, -2.5) {};
		\node [style=none] (31) at (-1.75, -3.5) {};
		\node [style=none] (32) at (-0.75, -2.5) {};
		\node [style=none] (33) at (-0.75, -3.5) {};
		\node [style=none] (34) at (0, -2.5) {};
		\node [style=none] (35) at (0, -3.5) {};
		\node [style=none] (36) at (1, -2.5) {};
		\node [style=none] (37) at (1, -3.5) {};
		\node [style=rule dot] (38) at (-2.5, -1) {$\pi_1$};
		\node [style=none] (39) at (-2, -2) {};
		\node [style=rule dot] (40) at (-0.75, 1.5) {$\sigma$};
		\node [style=none] (41) at (-2, 3.5) {};
		\node [style=none] (42) at (-2, 2.5) {};
		\node [style=none] (43) at (3, 3.5) {};
		\node [style=none] (44) at (3, 2.5) {};
		\node [style=none] (45) at (0.5, 3) {$\pi_2'$};
		\node [style=none] (46) at (-1.25, 0) {...};
		\node [style=none] (47) at (0.5, 1.5) {...};
		\node [style=none] (48) at (2.25, 1.5) {...};
		\node [style=none] (49) at (0.5, -1) {...};
		\node [style=none] (50) at (2.25, -2.25) {...};
		\node [style=none] (51) at (-1.25, -3) {...};
		\node [style=none] (52) at (0.5, -3) {...};
	\end{pgfonlayer}
	\begin{pgfonlayer}{edgelayer}
		\draw [style=box edge] (0.center)
			 to (1.center)
			 to (3.center)
			 to (2.center)
			 to cycle;
		\draw [style=box edge] (4.center)
			 to (5.center)
			 to (7.center)
			 to (6.center)
			 to cycle;
		\draw (8.center) to (9.center);
		\draw (10.center) to (11.center);
		\draw (12.center) to (13.center);
		\draw (14.center) to (15.center);
		\draw (18.center) to (19.center);
		\draw (20.center) to (21.center);
		\draw (22.center) to (23.center);
		\draw (24.center) to (25.center);
		\draw (26.center) to (27.center);
		\draw (28.center) to (29.center);
		\draw (30.center) to (31.center);
		\draw (32.center) to (33.center);
		\draw (34.center) to (35.center);
		\draw (36.center) to (37.center);
		\draw [in=180, out=-90, looseness=1.25] (38) to (39.center);
		\draw [style=box edge] (43.center)
			 to (41.center)
			 to (42.center)
			 to (44.center)
			 to cycle;
	\end{pgfonlayer}
\end{tikzpicture}

%% file: cutred-prf-sigma-2.tikz
\begin{tikzpicture}
	\begin{pgfonlayer}{nodelayer}
		\node [style=none] (0) at (1.5, -1.5) {};
		\node [style=none] (1) at (1.5, -2.5) {};
		\node [style=none] (2) at (4.75, -1.5) {};
		\node [style=none] (3) at (4.75, -2.5) {};
		\node [style=none] (4) at (3.25, 1.5) {};
		\node [style=none] (5) at (3.25, 0.5) {};
		\node [style=none] (6) at (6.5, 1.5) {};
		\node [style=none] (7) at (6.5, 0.5) {};
		\node [style=none] (8) at (1.75, 2.5) {};
		\node [style=none] (9) at (1.75, -1.5) {};
		\node [style=none] (10) at (2.75, 2.5) {};
		\node [style=none] (11) at (2.75, -1.5) {};
		\node [style=none] (12) at (3.5, 0.5) {};
		\node [style=none] (13) at (3.5, -1.5) {};
		\node [style=none] (14) at (4.5, 0.5) {};
		\node [style=none] (15) at (4.5, -1.5) {};
		\node [style=none] (16) at (3.125, -2) {\rCutUnH};
		\node [style=none] (17) at (4.875, 1) {$B$};
		\node [style=none] (18) at (3.5, 2.5) {};
		\node [style=none] (19) at (3.5, 1.5) {};
		\node [style=none] (20) at (4.5, 2.5) {};
		\node [style=none] (21) at (4.5, 1.5) {};
		\node [style=none] (22) at (5.25, 2.5) {};
		\node [style=none] (23) at (5.25, 1.5) {};
		\node [style=none] (24) at (6.25, 2.5) {};
		\node [style=none] (25) at (6.25, 1.5) {};
		\node [style=none] (26) at (5.25, 0.5) {};
		\node [style=none] (27) at (5.25, -3.5) {};
		\node [style=none] (28) at (6.25, 0.5) {};
		\node [style=none] (29) at (6.25, -3.5) {};
		\node [style=none] (30) at (1.75, -2.5) {};
		\node [style=none] (31) at (1.75, -3.5) {};
		\node [style=none] (32) at (2.75, -2.5) {};
		\node [style=none] (33) at (2.75, -3.5) {};
		\node [style=none] (34) at (3.5, -2.5) {};
		\node [style=none] (35) at (3.5, -3.5) {};
		\node [style=none] (36) at (4.5, -2.5) {};
		\node [style=none] (37) at (4.5, -3.5) {};
		\node [style=rule dot] (38) at (1, -1) {$\pi_1$};
		\node [style=none] (39) at (1.5, -2) {};
		\node [style=rule dot] (40) at (2.75, -0.5) {$\sigma$};
		\node [style=none] (41) at (1.5, 3.5) {};
		\node [style=none] (42) at (1.5, 2.5) {};
		\node [style=none] (43) at (6.5, 3.5) {};
		\node [style=none] (44) at (6.5, 2.5) {};
		\node [style=none] (45) at (4, 3) {$\pi_2'$};
		\node [style=none] (46) at (2.25, 1) {...};
		\node [style=none] (47) at (4, 2) {...};
		\node [style=none] (48) at (5.75, 2) {...};
		\node [style=none] (49) at (4, -0.5) {...};
		\node [style=none] (50) at (5.75, -1.25) {...};
		\node [style=none] (51) at (2.25, -3) {...};
		\node [style=none] (52) at (4, -3) {...};
	\end{pgfonlayer}
	\begin{pgfonlayer}{edgelayer}
		\draw [style=box edge] (2.center)
			 to (0.center)
			 to (1.center)
			 to (3.center)
			 to cycle;
		\draw [style=box edge] (5.center)
			 to (7.center)
			 to (6.center)
			 to (4.center)
			 to cycle;
		\draw (8.center) to (9.center);
		\draw (10.center) to (11.center);
		\draw (12.center) to (13.center);
		\draw (14.center) to (15.center);
		\draw (18.center) to (19.center);
		\draw (20.center) to (21.center);
		\draw (22.center) to (23.center);
		\draw (24.center) to (25.center);
		\draw (26.center) to (27.center);
		\draw (28.center) to (29.center);
		\draw (30.center) to (31.center);
		\draw (32.center) to (33.center);
		\draw (34.center) to (35.center);
		\draw (36.center) to (37.center);
		\draw [in=180, out=-90, looseness=1.25] (38) to (39.center);
		\draw [style=box edge] (43.center)
			 to (41.center)
			 to (42.center)
			 to (44.center)
			 to cycle;
	\end{pgfonlayer}
\end{tikzpicture}

%% file: cutred-prf-sigma-3.tikz
\begin{tikzpicture}
	\begin{pgfonlayer}{nodelayer}
		\node [style=none] (0) at (-3, -0.5) {};
		\node [style=none] (1) at (-3, -1.5) {};
		\node [style=none] (2) at (0.25, -0.5) {};
		\node [style=none] (3) at (0.25, -1.5) {};
		\node [style=none] (4) at (-1.25, 1.5) {};
		\node [style=none] (5) at (-1.25, 0.5) {};
		\node [style=none] (6) at (2, 1.5) {};
		\node [style=none] (7) at (2, 0.5) {};
		\node [style=none] (8) at (-2.75, 2.5) {};
		\node [style=none] (9) at (-2.75, -0.5) {};
		\node [style=none] (10) at (-1.75, 2.5) {};
		\node [style=none] (11) at (-1.75, -0.5) {};
		\node [style=none] (12) at (-1, 0.5) {};
		\node [style=none] (13) at (-1, -0.5) {};
		\node [style=none] (14) at (0, 0.5) {};
		\node [style=none] (15) at (0, -0.5) {};
		\node [style=none] (16) at (-1.375, -1) {\rCutUnH};
		\node [style=none] (17) at (0.375, 1) {$B$};
		\node [style=none] (18) at (-1, 2.5) {};
		\node [style=none] (19) at (-1, 1.5) {};
		\node [style=none] (20) at (0, 2.5) {};
		\node [style=none] (21) at (0, 1.5) {};
		\node [style=none] (22) at (0.75, 2.5) {};
		\node [style=none] (23) at (0.75, 1.5) {};
		\node [style=none] (24) at (1.75, 2.5) {};
		\node [style=none] (25) at (1.75, 1.5) {};
		\node [style=none] (26) at (0.75, 0.5) {};
		\node [style=none] (27) at (0.75, -3.5) {};
		\node [style=none] (28) at (1.75, 0.5) {};
		\node [style=none] (29) at (1.75, -3.5) {};
		\node [style=none] (30) at (-2.75, -1.5) {};
		\node [style=none] (31) at (-2.75, -3.5) {};
		\node [style=none] (32) at (-1.75, -1.5) {};
		\node [style=none] (33) at (-1.75, -3.5) {};
		\node [style=none] (34) at (-1, -1.5) {};
		\node [style=none] (35) at (-1, -3.5) {};
		\node [style=none] (36) at (0, -1.5) {};
		\node [style=none] (37) at (0, -3.5) {};
		\node [style=rule dot] (38) at (-3.5, 0) {$\pi_1$};
		\node [style=none] (39) at (-3, -1) {};
		\node [style=rule dot] (40) at (-1.75, -2.5) {$\sigma$};
		\node [style=none] (41) at (-3, 3.5) {};
		\node [style=none] (42) at (-3, 2.5) {};
		\node [style=none] (43) at (2, 3.5) {};
		\node [style=none] (44) at (2, 2.5) {};
		\node [style=none] (45) at (-0.5, 3) {$\pi_2'$};
		\node [style=none] (46) at (-2.25, 1) {...};
		\node [style=none] (47) at (-0.5, 2) {...};
		\node [style=none] (48) at (1.25, 2) {...};
		\node [style=none] (49) at (-0.5, 0) {...};
		\node [style=none] (50) at (1.25, -1.25) {...};
		\node [style=none] (51) at (-2.25, -3.25) {...};
		\node [style=none] (52) at (-0.5, -2.5) {...};
	\end{pgfonlayer}
	\begin{pgfonlayer}{edgelayer}
		\draw [style=box edge] (0.center)
			 to (1.center)
			 to (3.center)
			 to (2.center)
			 to cycle;
		\draw [style=box edge] (4.center)
			 to (5.center)
			 to (7.center)
			 to (6.center)
			 to cycle;
		\draw (8.center) to (9.center);
		\draw (10.center) to (11.center);
		\draw (12.center) to (13.center);
		\draw (14.center) to (15.center);
		\draw (18.center) to (19.center);
		\draw (20.center) to (21.center);
		\draw (22.center) to (23.center);
		\draw (24.center) to (25.center);
		\draw (26.center) to (27.center);
		\draw (28.center) to (29.center);
		\draw (30.center) to (31.center);
		\draw (32.center) to (33.center);
		\draw (34.center) to (35.center);
		\draw (36.center) to (37.center);
		\draw [in=180, out=-90, looseness=1.25] (38) to (39.center);
		\draw [style=box edge] (43.center)
			 to (41.center)
			 to (42.center)
			 to (44.center)
			 to cycle;
	\end{pgfonlayer}
\end{tikzpicture}

%% file: cutred-prf-rUnitUn-1.tikz
\begin{tikzpicture}
	\begin{pgfonlayer}{nodelayer}
		\node [style=none] (0) at (-5.25, -3) {};
		\node [style=none] (1) at (-5.25, -4) {};
		\node [style=none] (2) at (-0.25, -3) {};
		\node [style=none] (3) at (-0.25, -4) {};
		\node [style=none] (4) at (-1.75, -1) {};
		\node [style=none] (5) at (-1.75, -2) {};
		\node [style=none] (6) at (1.5, -1) {};
		\node [style=none] (7) at (1.5, -2) {};
		\node [style=none] (8) at (-5, 3.5) {};
		\node [style=none] (9) at (-5, -3) {};
		\node [style=none] (10) at (-4, 3.5) {};
		\node [style=none] (11) at (-4, -3) {};
		\node [style=none] (12) at (-1.5, -2) {};
		\node [style=none] (13) at (-1.5, -3) {};
		\node [style=none] (14) at (-0.5, -2) {};
		\node [style=none] (15) at (-0.5, -3) {};
		\node [style=none] (16) at (-2.625, -3.5) {\rCutUnH};
		\node [style=none] (17) at (-0.125, -1.5) {$B$};
		\node [style=none] (18) at (-1.5, 1) {};
		\node [style=none] (19) at (-1.5, -1) {};
		\node [style=none] (20) at (-0.5, 1) {};
		\node [style=none] (21) at (-0.5, -1) {};
		\node [style=none] (22) at (0.25, 3.5) {};
		\node [style=none] (23) at (0.25, -1) {};
		\node [style=none] (24) at (1.25, 3.5) {};
		\node [style=none] (25) at (1.25, -1) {};
		\node [style=none] (26) at (0.25, -2) {};
		\node [style=none] (27) at (0.25, -5) {};
		\node [style=none] (28) at (1.25, -2) {};
		\node [style=none] (29) at (1.25, -5) {};
		\node [style=none] (30) at (-5, -4) {};
		\node [style=none] (31) at (-5, -5) {};
		\node [style=none] (32) at (-4, -4) {};
		\node [style=none] (33) at (-4, -5) {};
		\node [style=none] (34) at (-1.5, -4) {};
		\node [style=none] (35) at (-1.5, -5) {};
		\node [style=none] (36) at (-0.5, -4) {};
		\node [style=none] (37) at (-0.5, -5) {};
		\node [style=rule dot] (38) at (-5.75, -2.5) {$\pi_1$};
		\node [style=none] (39) at (-5.25, -3.5) {};
		\node [style=none] (41) at (-5.25, 4.5) {};
		\node [style=none] (42) at (-5.25, 3.5) {};
		\node [style=none] (43) at (1.5, 4.5) {};
		\node [style=none] (44) at (1.5, 3.5) {};
		\node [style=none] (45) at (-1.75, 4) {$\pi_2'$};
		\node [style=none] (46) at (-4.5, -1.5) {...};
		\node [style=none] (47) at (-1, -0.5) {...};
		\node [style=none] (48) at (0.75, -0.5) {...};
		\node [style=none] (49) at (-1, -2.5) {...};
		\node [style=none] (50) at (0.75, -3.75) {...};
		\node [style=none] (51) at (-4.5, -4.5) {...};
		\node [style=none] (52) at (-1, -4.5) {...};
		\node [style=none] (53) at (-2.25, -4) {};
		\node [style=none] (54) at (-2.25, -5) {};
		\node [style=none] (55) at (-2.25, 1) {};
		\node [style=none] (56) at (-2.25, -3) {};
		\node [style=none] (57) at (-2.5, 2) {};
		\node [style=none] (58) at (-2.5, 1) {};
		\node [style=none] (59) at (-0.25, 2) {};
		\node [style=none] (60) at (-0.25, 1) {};
		\node [style=none] (61) at (-1.375, 1.5) {\rUnitUn};
		\node [style=none] (62) at (-1.375, 3.5) {};
		\node [style=none] (63) at (-1.375, 2) {};
		\node [style=rule dot] (64) at (-3, 2.5) {$\pi_2$};
		\node [style=none] (65) at (-2.5, 1.5) {};
		\node [style=none, font={\scriptsize}] (66) at (-2.625, 0) {$\omega_i$};
	\end{pgfonlayer}
	\begin{pgfonlayer}{edgelayer}
		\draw [style=box edge] (0.center)
			 to (1.center)
			 to (3.center)
			 to (2.center)
			 to cycle;
		\draw [style=box edge] (4.center)
			 to (5.center)
			 to (7.center)
			 to (6.center)
			 to cycle;
		\draw (8.center) to (9.center);
		\draw (10.center) to (11.center);
		\draw (12.center) to (13.center);
		\draw (14.center) to (15.center);
		\draw (18.center) to (19.center);
		\draw (20.center) to (21.center);
		\draw (22.center) to (23.center);
		\draw (24.center) to (25.center);
		\draw (26.center) to (27.center);
		\draw (28.center) to (29.center);
		\draw (30.center) to (31.center);
		\draw (32.center) to (33.center);
		\draw (34.center) to (35.center);
		\draw (36.center) to (37.center);
		\draw [in=180, out=-90, looseness=1.25] (38) to (39.center);
		\draw [style=box edge] (43.center)
			 to (41.center)
			 to (42.center)
			 to (44.center)
			 to cycle;
		\draw (53.center) to (54.center);
		\draw (55.center) to (56.center);
		\draw [style=box edge] (59.center)
			 to (57.center)
			 to (58.center)
			 to (60.center)
			 to cycle;
		\draw (62.center) to (63.center);
		\draw [in=180, out=-90, looseness=1.25] (64) to (65.center);
	\end{pgfonlayer}
\end{tikzpicture}

%% file: cutred-prf-rUnitUn-2.tikz
\begin{tikzpicture}
	\begin{pgfonlayer}{nodelayer}
		\node [style=none] (0) at (-6, -3) {};
		\node [style=none] (1) at (-6, -4) {};
		\node [style=none] (2) at (0.25, -3) {};
		\node [style=none] (3) at (0.25, -4) {};
		\node [style=none] (4) at (-1.75, -1) {};
		\node [style=none] (5) at (-1.75, -2) {};
		\node [style=none] (6) at (2, -1) {};
		\node [style=none] (7) at (2, -2) {};
		\node [style=none] (8) at (-5.75, 3.5) {};
		\node [style=none] (9) at (-5.75, -3) {};
		\node [style=none] (10) at (-4.75, 3.5) {};
		\node [style=none] (11) at (-4.75, -3) {};
		\node [style=none] (12) at (-1.5, -2) {};
		\node [style=none] (13) at (-1.5, -3) {};
		\node [style=none] (14) at (0, -2) {};
		\node [style=none] (15) at (0, -3) {};
		\node [style=none] (16) at (-2.625, -3.5) {\rCutUnH};
		\node [style=none] (17) at (0.125, -1.5) {$B$};
		\node [style=none] (18) at (-1.5, 1) {};
		\node [style=none] (19) at (-1.5, -1) {};
		\node [style=none] (20) at (0, 1) {};
		\node [style=none] (21) at (0, -1) {};
		\node [style=none] (22) at (0.75, 3.5) {};
		\node [style=none] (23) at (0.75, -1) {};
		\node [style=none] (24) at (1.75, 3.5) {};
		\node [style=none] (25) at (1.75, -1) {};
		\node [style=none] (26) at (0.75, -2) {};
		\node [style=none] (27) at (0.75, -5) {};
		\node [style=none] (28) at (1.75, -2) {};
		\node [style=none] (29) at (1.75, -5) {};
		\node [style=none] (30) at (-5.75, -4) {};
		\node [style=none] (31) at (-5.75, -5) {};
		\node [style=none] (32) at (-4.75, -4) {};
		\node [style=none] (33) at (-4.75, -5) {};
		\node [style=none] (34) at (-1.5, -4) {};
		\node [style=none] (35) at (-1.5, -5) {};
		\node [style=none] (36) at (0, -4) {};
		\node [style=none] (37) at (0, -5) {};
		\node [style=rule dot] (38) at (-6.5, -2.5) {$\pi_1$};
		\node [style=none] (39) at (-6, -3.5) {};
		\node [style=none] (41) at (-6, 4.5) {};
		\node [style=none] (42) at (-6, 3.5) {};
		\node [style=none] (43) at (2, 4.5) {};
		\node [style=none] (44) at (2, 3.5) {};
		\node [style=none] (45) at (-1.75, 4) {$\pi_2'$};
		\node [style=none] (46) at (-5.25, -1.5) {...};
		\node [style=none] (47) at (-0.75, -0.5) {...};
		\node [style=none] (48) at (1.25, -0.5) {...};
		\node [style=none] (49) at (-0.75, -2.5) {...};
		\node [style=none] (50) at (1.25, -3.75) {...};
		\node [style=none] (51) at (-5.25, -4.5) {...};
		\node [style=none] (52) at (-0.75, -4.5) {...};
		\node [style=none] (53) at (-3, -4) {};
		\node [style=none] (54) at (-3, -5) {};
		\node [style=none] (55) at (-3, 1) {};
		\node [style=none] (56) at (-3, -3) {};
		\node [style=none] (57) at (-3.25, 2) {};
		\node [style=none] (58) at (-3.25, 1) {};
		\node [style=none] (59) at (0.25, 2) {};
		\node [style=none] (60) at (0.25, 1) {};
		\node [style=none] (61) at (-1.5, 1.5) {\rUnitUn};
		\node [style=none] (62) at (-1.5, 3.5) {};
		\node [style=none] (63) at (-1.5, 2) {};
		\node [style=rule dot] (64) at (-3.75, 2.5) {$\pi_2$};
		\node [style=none] (65) at (-3.25, 1.5) {};
		\node [style=none, font={\scriptsize}] (66) at (-3.375, 0) {$\omega_i$};
		\node [style=none, font={\scriptsize}] (67) at (-1.875, 0) {$\omega_{j_1}$};
		\node [style=none, font={\scriptsize}] (68) at (-0.375, 0) {$\omega_{j_n}$};
	\end{pgfonlayer}
	\begin{pgfonlayer}{edgelayer}
		\draw [style=box edge] (0.center)
			 to (1.center)
			 to (3.center)
			 to (2.center)
			 to cycle;
		\draw [style=box edge] (4.center)
			 to (5.center)
			 to (7.center)
			 to (6.center)
			 to cycle;
		\draw (8.center) to (9.center);
		\draw (10.center) to (11.center);
		\draw (12.center) to (13.center);
		\draw (14.center) to (15.center);
		\draw (18.center) to (19.center);
		\draw (20.center) to (21.center);
		\draw (22.center) to (23.center);
		\draw (24.center) to (25.center);
		\draw (26.center) to (27.center);
		\draw (28.center) to (29.center);
		\draw (30.center) to (31.center);
		\draw (32.center) to (33.center);
		\draw (34.center) to (35.center);
		\draw (36.center) to (37.center);
		\draw [in=180, out=-90, looseness=1.25] (38) to (39.center);
		\draw [style=box edge] (43.center)
			 to (41.center)
			 to (42.center)
			 to (44.center)
			 to cycle;
		\draw (53.center) to (54.center);
		\draw (55.center) to (56.center);
		\draw [style=box edge] (59.center)
			 to (57.center)
			 to (58.center)
			 to (60.center)
			 to cycle;
		\draw (62.center) to (63.center);
		\draw [in=180, out=-90, looseness=1.25] (64) to (65.center);
	\end{pgfonlayer}
\end{tikzpicture}

%% file: cutred-prf-rUnitUn-3.tikz
\begin{tikzpicture}
	\begin{pgfonlayer}{nodelayer}
		\node [style=none] (0) at (-6, -3) {};
		\node [style=none] (1) at (-6, -4) {};
		\node [style=none] (2) at (0.25, -3) {};
		\node [style=none] (3) at (0.25, -4) {};
		\node [style=none] (4) at (0.5, -1) {};
		\node [style=none] (5) at (0.5, -2) {};
		\node [style=none] (6) at (2, -1) {};
		\node [style=none] (7) at (2, -2) {};
		\node [style=none] (8) at (-5.75, 3.5) {};
		\node [style=none] (9) at (-5.75, -3) {};
		\node [style=none] (10) at (-4.75, 3.5) {};
		\node [style=none] (11) at (-4.75, -3) {};
		\node [style=none] (16) at (-2.625, -3.5) {\rCutUnH};
		\node [style=none] (17) at (1.25, -1.5) {$B$};
		\node [style=none] (18) at (-1.5, 1) {};
		\node [style=none] (19) at (-1.5, -3) {};
		\node [style=none] (20) at (0, 1) {};
		\node [style=none] (21) at (0, -3) {};
		\node [style=none] (22) at (0.75, 3.5) {};
		\node [style=none] (23) at (0.75, -1) {};
		\node [style=none] (24) at (1.75, 3.5) {};
		\node [style=none] (25) at (1.75, -1) {};
		\node [style=none] (26) at (0.75, -2) {};
		\node [style=none] (27) at (0.75, -5) {};
		\node [style=none] (28) at (1.75, -2) {};
		\node [style=none] (29) at (1.75, -5) {};
		\node [style=none] (30) at (-5.75, -4) {};
		\node [style=none] (31) at (-5.75, -5) {};
		\node [style=none] (32) at (-4.75, -4) {};
		\node [style=none] (33) at (-4.75, -5) {};
		\node [style=none] (34) at (-1.5, -4) {};
		\node [style=none] (35) at (-1.5, -5) {};
		\node [style=none] (36) at (0, -4) {};
		\node [style=none] (37) at (0, -5) {};
		\node [style=rule dot] (38) at (-6.5, -2.5) {$\pi_1$};
		\node [style=none] (39) at (-6, -3.5) {};
		\node [style=none] (40) at (-6, 4.5) {};
		\node [style=none] (41) at (-6, 3.5) {};
		\node [style=none] (42) at (2, 4.5) {};
		\node [style=none] (43) at (2, 3.5) {};
		\node [style=none] (44) at (-1.75, 4) {$\pi_2'$};
		\node [style=none] (45) at (-5.25, -1.5) {...};
		\node [style=none] (47) at (1.25, -0.5) {...};
		\node [style=none] (48) at (-0.75, -2.5) {...};
		\node [style=none] (49) at (1.25, -3.75) {...};
		\node [style=none] (50) at (-5.25, -4.5) {...};
		\node [style=none] (51) at (-0.75, -4.5) {...};
		\node [style=none] (52) at (-3, -4) {};
		\node [style=none] (53) at (-3, -5) {};
		\node [style=none] (54) at (-3, 1) {};
		\node [style=none] (55) at (-3, -3) {};
		\node [style=none] (56) at (-3.25, 2) {};
		\node [style=none] (57) at (-3.25, 1) {};
		\node [style=none] (58) at (0.25, 2) {};
		\node [style=none] (59) at (0.25, 1) {};
		\node [style=none] (60) at (-1.5, 1.5) {\rUnitUn};
		\node [style=none] (61) at (-1.5, 3.5) {};
		\node [style=none] (62) at (-1.5, 2) {};
		\node [style=rule dot] (63) at (-3.75, 2.5) {$\pi_2$};
		\node [style=none] (64) at (-3.25, 1.5) {};
		\node [style=none, font={\scriptsize}] (65) at (-3.375, 0) {$\omega_i$};
		\node [style=none, font={\scriptsize}] (66) at (-1.875, 0) {$\omega_{j_1}$};
		\node [style=none, font={\scriptsize}] (67) at (-0.375, 0) {$\omega_{j_n}$};
	\end{pgfonlayer}
	\begin{pgfonlayer}{edgelayer}
		\draw [style=box edge] (0.center)
			 to (1.center)
			 to (3.center)
			 to (2.center)
			 to cycle;
		\draw [style=box edge] (4.center)
			 to (5.center)
			 to (7.center)
			 to (6.center)
			 to cycle;
		\draw (8.center) to (9.center);
		\draw (10.center) to (11.center);
		\draw (18.center) to (19.center);
		\draw (20.center) to (21.center);
		\draw (22.center) to (23.center);
		\draw (24.center) to (25.center);
		\draw (26.center) to (27.center);
		\draw (28.center) to (29.center);
		\draw (30.center) to (31.center);
		\draw (32.center) to (33.center);
		\draw (34.center) to (35.center);
		\draw (36.center) to (37.center);
		\draw [in=180, out=-90, looseness=1.25] (38) to (39.center);
		\draw [style=box edge] (42.center)
			 to (40.center)
			 to (41.center)
			 to (43.center)
			 to cycle;
		\draw (52.center) to (53.center);
		\draw (54.center) to (55.center);
		\draw [style=box edge] (58.center)
			 to (56.center)
			 to (57.center)
			 to (59.center)
			 to cycle;
		\draw (61.center) to (62.center);
		\draw [in=180, out=-90, looseness=1.25] (63) to (64.center);
	\end{pgfonlayer}
\end{tikzpicture}

%% file: cutred-prf-rUnitUn-4.tikz
\begin{tikzpicture}
	\begin{pgfonlayer}{nodelayer}
		\node [style=none] (0) at (-5.25, -3) {};
		\node [style=none] (1) at (-5.25, -4) {};
		\node [style=none] (2) at (1, -3) {};
		\node [style=none] (3) at (1, -4) {};
		\node [style=none] (4) at (-1, 0) {};
		\node [style=none] (5) at (-1, -1) {};
		\node [style=none] (6) at (2.75, 0) {};
		\node [style=none] (7) at (2.75, -1) {};
		\node [style=none] (8) at (-5, 4.5) {};
		\node [style=none] (9) at (-5, -3) {};
		\node [style=none] (10) at (-4, 4.5) {};
		\node [style=none] (11) at (-4, -3) {};
		\node [style=none] (12) at (-0.75, -1) {};
		\node [style=none] (13) at (-0.75, -3) {};
		\node [style=none] (14) at (0.75, -1) {};
		\node [style=none] (15) at (0.75, -3) {};
		\node [style=none] (16) at (-1.875, -3.5) {\rCutUnH};
		\node [style=none] (17) at (0.875, -0.5) {$B$};
		\node [style=none] (18) at (-0.75, 2) {};
		\node [style=none] (19) at (-0.75, 0) {};
		\node [style=none] (20) at (0.75, 2) {};
		\node [style=none] (21) at (0.75, 0) {};
		\node [style=none] (22) at (1.5, 4.5) {};
		\node [style=none] (23) at (1.5, 0) {};
		\node [style=none] (24) at (2.5, 4.5) {};
		\node [style=none] (25) at (2.5, 0) {};
		\node [style=none] (26) at (1.5, -1) {};
		\node [style=none] (27) at (1.5, -5) {};
		\node [style=none] (28) at (2.5, -1) {};
		\node [style=none] (29) at (2.5, -5) {};
		\node [style=none] (30) at (-5, -4) {};
		\node [style=none] (31) at (-5, -5) {};
		\node [style=none] (32) at (-4, -4) {};
		\node [style=none] (33) at (-4, -5) {};
		\node [style=none] (34) at (-0.75, -4) {};
		\node [style=none] (35) at (-0.75, -5) {};
		\node [style=none] (36) at (0.75, -4) {};
		\node [style=none] (37) at (0.75, -5) {};
		\node [style=rule dot] (38) at (-5.75, -2.5) {$\pi_1$};
		\node [style=none] (39) at (-5.25, -3.5) {};
		\node [style=none] (40) at (-5.25, 5.5) {};
		\node [style=none] (41) at (-5.25, 4.5) {};
		\node [style=none] (42) at (2.75, 5.5) {};
		\node [style=none] (43) at (2.75, 4.5) {};
		\node [style=none] (44) at (-1, 5) {$\pi_2'$};
		\node [style=none] (45) at (-4.5, -0.5) {...};
		\node [style=none] (46) at (0, 0.5) {...};
		\node [style=none] (47) at (2, 0.5) {...};
		\node [style=none] (48) at (0, -2.5) {...};
		\node [style=none] (49) at (2, -3.75) {...};
		\node [style=none] (50) at (-4.5, -4.5) {...};
		\node [style=none] (51) at (0, -4.5) {...};
		\node [style=none] (52) at (-2.25, -4) {};
		\node [style=none] (53) at (-2.25, -5) {};
		\node [style=none] (54) at (-2.25, 2) {};
		\node [style=none] (55) at (-2.25, -3) {};
		\node [style=none] (56) at (-2.5, 3) {};
		\node [style=none] (57) at (-2.5, 2) {};
		\node [style=none] (58) at (1, 3) {};
		\node [style=none] (59) at (1, 2) {};
		\node [style=none] (60) at (-0.75, 2.5) {\rUnitUn};
		\node [style=none] (61) at (-0.75, 4.5) {};
		\node [style=none] (62) at (-0.75, 3) {};
		\node [style=rule dot] (63) at (-3, 3.5) {$\pi_2$};
		\node [style=none] (64) at (-2.5, 2.5) {};
		\node [style=rule dot, font={\scriptsize}] (68) at (-0.75, -2) {$r$};
		\node [style=none, font={\scriptsize}] (69) at (-2.6, 1) {$\omega_i$};
		\node [style=none, font={\scriptsize}] (70) at (-1.1, 1) {$\omega_{j_1}$};
		\node [style=none, font={\scriptsize}] (71) at (0.4, 1) {$\omega_{j_n}$};
	\end{pgfonlayer}
	\begin{pgfonlayer}{edgelayer}
		\draw [style=box edge] (0.center)
			 to (1.center)
			 to (3.center)
			 to (2.center)
			 to cycle;
		\draw [style=box edge] (4.center)
			 to (5.center)
			 to (7.center)
			 to (6.center)
			 to cycle;
		\draw (8.center) to (9.center);
		\draw (10.center) to (11.center);
		\draw (12.center) to (13.center);
		\draw (14.center) to (15.center);
		\draw (18.center) to (19.center);
		\draw (20.center) to (21.center);
		\draw (22.center) to (23.center);
		\draw (24.center) to (25.center);
		\draw (26.center) to (27.center);
		\draw (28.center) to (29.center);
		\draw (30.center) to (31.center);
		\draw (32.center) to (33.center);
		\draw (34.center) to (35.center);
		\draw (36.center) to (37.center);
		\draw [in=180, out=-90, looseness=1.25] (38) to (39.center);
		\draw [style=box edge] (42.center)
			 to (40.center)
			 to (41.center)
			 to (43.center)
			 to cycle;
		\draw (52.center) to (53.center);
		\draw (54.center) to (55.center);
		\draw [style=box edge] (58.center)
			 to (56.center)
			 to (57.center)
			 to (59.center)
			 to cycle;
		\draw (61.center) to (62.center);
		\draw [in=180, out=-90, looseness=1.25] (63) to (64.center);
	\end{pgfonlayer}
\end{tikzpicture}

%% file: cutred-prf-rUnitUn-3f.tikz
\begin{tikzpicture}
	\begin{pgfonlayer}{nodelayer}
		\node [style=none] (0) at (-6, -3) {};
		\node [style=none] (1) at (-6, -4) {};
		\node [style=none] (2) at (0.25, -3) {};
		\node [style=none] (3) at (0.25, -4) {};
		\node [style=none] (4) at (0.5, 2) {};
		\node [style=none] (5) at (0.5, 1) {};
		\node [style=none] (6) at (2, 2) {};
		\node [style=none] (7) at (2, 1) {};
		\node [style=none] (8) at (-5.75, 3) {};
		\node [style=none] (9) at (-5.75, -3) {};
		\node [style=none] (10) at (-4.75, 3) {};
		\node [style=none] (11) at (-4.75, -3) {};
		\node [style=none] (16) at (-2.625, -3.5) {\rCutUnH};
		\node [style=none] (17) at (1.25, 1.5) {$B$};
		\node [style=none] (18) at (-1.5, -1) {};
		\node [style=none] (19) at (-1.5, -3) {};
		\node [style=none] (20) at (0, -1) {};
		\node [style=none] (21) at (0, -3) {};
		\node [style=none] (22) at (0.75, 3) {};
		\node [style=none] (23) at (0.75, 2) {};
		\node [style=none] (24) at (1.75, 3) {};
		\node [style=none] (25) at (1.75, 2) {};
		\node [style=none] (26) at (0.75, 1) {};
		\node [style=none] (27) at (0.75, -5) {};
		\node [style=none] (28) at (1.75, 1) {};
		\node [style=none] (29) at (1.75, -5) {};
		\node [style=none] (30) at (-5.75, -4) {};
		\node [style=none] (31) at (-5.75, -5) {};
		\node [style=none] (32) at (-4.75, -4) {};
		\node [style=none] (33) at (-4.75, -5) {};
		\node [style=none] (34) at (-1.5, -4) {};
		\node [style=none] (35) at (-1.5, -5) {};
		\node [style=none] (36) at (0, -4) {};
		\node [style=none] (37) at (0, -5) {};
		\node [style=rule dot] (38) at (-6.5, -2.5) {$\pi_1$};
		\node [style=none] (39) at (-6, -3.5) {};
		\node [style=none] (40) at (-6, 4) {};
		\node [style=none] (41) at (-6, 3) {};
		\node [style=none] (42) at (2, 4) {};
		\node [style=none] (43) at (2, 3) {};
		\node [style=none] (44) at (-1.75, 3.5) {$\pi_2'$};
		\node [style=none] (45) at (-5.25, -1.5) {...};
		\node [style=none] (47) at (1.25, 2.5) {...};
		\node [style=none] (48) at (-0.75, -2.5) {...};
		\node [style=none] (49) at (1.25, -3.75) {...};
		\node [style=none] (50) at (-5.25, -4.5) {...};
		\node [style=none] (51) at (-0.75, -4.5) {...};
		\node [style=none] (52) at (-3, -4) {};
		\node [style=none] (53) at (-3, -5) {};
		\node [style=none] (54) at (-3, -1) {};
		\node [style=none] (55) at (-3, -3) {};
		\node [style=none] (56) at (-3.25, 0) {};
		\node [style=none] (57) at (-3.25, -1) {};
		\node [style=none] (58) at (0.25, 0) {};
		\node [style=none] (59) at (0.25, -1) {};
		\node [style=none] (60) at (-1.5, -0.5) {\rUnitUn};
		\node [style=none] (61) at (-1.5, 3) {};
		\node [style=none] (62) at (-1.5, 0) {};
		\node [style=rule dot] (63) at (-3.75, 0.5) {$\pi_2$};
		\node [style=none] (64) at (-3.25, -0.5) {};
		\node [style=none, font={\scriptsize}] (65) at (-3.375, -2) {$\omega_i$};
		\node [style=none, font={\scriptsize}] (66) at (-1.875, -2) {$\omega_{j_1}$};
		\node [style=none, font={\scriptsize}] (67) at (-0.375, -2) {$\omega_{j_n}$};
	\end{pgfonlayer}
	\begin{pgfonlayer}{edgelayer}
		\draw [style=box edge] (0.center)
			 to (1.center)
			 to (3.center)
			 to (2.center)
			 to cycle;
		\draw [style=box edge] (4.center)
			 to (5.center)
			 to (7.center)
			 to (6.center)
			 to cycle;
		\draw (8.center) to (9.center);
		\draw (10.center) to (11.center);
		\draw (18.center) to (19.center);
		\draw (20.center) to (21.center);
		\draw (22.center) to (23.center);
		\draw (24.center) to (25.center);
		\draw (26.center) to (27.center);
		\draw (28.center) to (29.center);
		\draw (30.center) to (31.center);
		\draw (32.center) to (33.center);
		\draw (34.center) to (35.center);
		\draw (36.center) to (37.center);
		\draw [in=180, out=-90, looseness=1.25] (38) to (39.center);
		\draw [style=box edge] (42.center)
			 to (40.center)
			 to (41.center)
			 to (43.center)
			 to cycle;
		\draw (52.center) to (53.center);
		\draw (54.center) to (55.center);
		\draw [style=box edge] (58.center)
			 to (56.center)
			 to (57.center)
			 to (59.center)
			 to cycle;
		\draw (61.center) to (62.center);
		\draw [in=180, out=-90, looseness=1.25] (63) to (64.center);
	\end{pgfonlayer}
\end{tikzpicture}

%% file: cutred-prf-rUnitUn-3ff.tikz
\begin{tikzpicture}
	\begin{pgfonlayer}{nodelayer}
		\node [style=none] (0) at (-6, 0) {};
		\node [style=none] (1) at (-6, -1) {};
		\node [style=none] (2) at (-1.25, 0) {};
		\node [style=none] (3) at (-1.25, -1) {};
		\node [style=none] (4) at (0.5, 2) {};
		\node [style=none] (5) at (0.5, 1) {};
		\node [style=none] (6) at (2, 2) {};
		\node [style=none] (7) at (2, 1) {};
		\node [style=none] (8) at (-5.75, 3) {};
		\node [style=none] (9) at (-5.75, 0) {};
		\node [style=none] (10) at (-4.75, 3) {};
		\node [style=none] (11) at (-4.75, 0) {};
		\node [style=none] (16) at (-3.625, -0.5) {\rCutUnH};
		\node [style=none] (17) at (1.25, 1.5) {$B$};
		\node [style=none] (18) at (-1.5, -3.5) {};
		\node [style=none] (19) at (-1.5, -5) {};
		\node [style=none] (20) at (0, -3.5) {};
		\node [style=none] (21) at (0, -5) {};
		\node [style=none] (22) at (0.75, 3) {};
		\node [style=none] (23) at (0.75, 2) {};
		\node [style=none] (24) at (1.75, 3) {};
		\node [style=none] (25) at (1.75, 2) {};
		\node [style=none] (26) at (0.75, 1) {};
		\node [style=none] (27) at (0.75, -5) {};
		\node [style=none] (28) at (1.75, 1) {};
		\node [style=none] (29) at (1.75, -5) {};
		\node [style=none] (30) at (-5.75, -1) {};
		\node [style=none] (31) at (-5.75, -5) {};
		\node [style=none] (32) at (-4.75, -1) {};
		\node [style=none] (33) at (-4.75, -5) {};
		\node [style=none] (34) at (-1.5, -1) {};
		\node [style=none] (35) at (-1.5, -2.5) {};
		\node [style=rule dot] (38) at (-6.5, 0.5) {$\pi_1$};
		\node [style=none] (39) at (-6, -0.5) {};
		\node [style=none] (40) at (-6, 4) {};
		\node [style=none] (41) at (-6, 3) {};
		\node [style=none] (42) at (2, 4) {};
		\node [style=none] (43) at (2, 3) {};
		\node [style=none] (44) at (-1.75, 3.5) {$\pi_2'$};
		\node [style=none] (45) at (-5.25, 0.75) {...};
		\node [style=none] (47) at (1.25, 2.5) {...};
		\node [style=none] (48) at (-0.75, -4.75) {...};
		\node [style=none] (49) at (1.25, -3.75) {...};
		\node [style=none] (50) at (-5.25, -4.5) {...};
		\node [style=none] (54) at (-3, -3.5) {};
		\node [style=none] (55) at (-3, -5) {};
		\node [style=none] (56) at (-3.25, -2.5) {};
		\node [style=none] (57) at (-3.25, -3.5) {};
		\node [style=none] (58) at (0.25, -2.5) {};
		\node [style=none] (59) at (0.25, -3.5) {};
		\node [style=none] (60) at (-1.5, -3) {\rUnitUn};
		\node [style=none] (61) at (-1.5, 3) {};
		\node [style=none] (62) at (-1.5, 0) {};
		\node [style=rule dot] (63) at (-3.75, -2) {$\pi_2$};
		\node [style=none] (64) at (-3.25, -3) {};
		\node [style=none, font={\scriptsize}] (65) at (-3.375, -4.25) {$\omega_i$};
		\node [style=none, font={\scriptsize}] (66) at (-1.875, -4.25) {$\omega_{j_1}$};
		\node [style=none, font={\scriptsize}] (67) at (-0.375, -4.25) {$\omega_{j_n}$};
	\end{pgfonlayer}
	\begin{pgfonlayer}{edgelayer}
		\draw [style=box edge] (0.center)
			 to (1.center)
			 to (3.center)
			 to (2.center)
			 to cycle;
		\draw [style=box edge] (4.center)
			 to (5.center)
			 to (7.center)
			 to (6.center)
			 to cycle;
		\draw (8.center) to (9.center);
		\draw (10.center) to (11.center);
		\draw (18.center) to (19.center);
		\draw (20.center) to (21.center);
		\draw (22.center) to (23.center);
		\draw (24.center) to (25.center);
		\draw (26.center) to (27.center);
		\draw (28.center) to (29.center);
		\draw (30.center) to (31.center);
		\draw (32.center) to (33.center);
		\draw (34.center) to (35.center);
		\draw [in=180, out=-90, looseness=1.25] (38) to (39.center);
		\draw [style=box edge] (42.center)
			 to (40.center)
			 to (41.center)
			 to (43.center)
			 to cycle;
		\draw (54.center) to (55.center);
		\draw [style=box edge] (58.center)
			 to (56.center)
			 to (57.center)
			 to (59.center)
			 to cycle;
		\draw (61.center) to (62.center);
		\draw [in=180, out=-90, looseness=1.25] (63) to (64.center);
	\end{pgfonlayer}
\end{tikzpicture}